\documentclass[12pt,draftclsnofoot,onecolumn]{IEEEtran}
\usepackage[colorlinks,bookmarksopen,bookmarksnumbered,citecolor=red,urlcolor=red,]{hyperref}

\usepackage{color}
\definecolor{MyDarkBlue}{rgb}{0,0.08,0.45}
\definecolor{yellow}{rgb}{0.99,0.99,0.70}
\definecolor{myback}{RGB}{204,232,207}  
\definecolor{white}{rgb}{1.0,1.0,1.0}                             
\definecolor{black}{rgb}{0.00,0.00,0.00}
\usepackage{verbatim}
\usepackage{epsfig,bbm}
\usepackage[colorlinks,bookmarksopen,bookmarksnumbered,citecolor=red,urlcolor=red,]{hyperref}
\usepackage{CJK}
\usepackage{indentfirst}
\usepackage{multirow}
\usepackage{epstopdf}
\usepackage{graphicx}
\usepackage{footmisc}
\graphicspath{{./figures/}}
\usepackage{amsfonts}
\usepackage{mathrsfs}
\usepackage{setspace}
\usepackage{amsmath}
\usepackage{algorithm,algorithmic,amsbsy,amsmath,amssymb,epsfig,bbm,mathrsfs, bbm} %,mathabx
\usepackage{amsthm}
\usepackage{verbatim}
\usepackage[noadjust]{cite}
\usepackage[subfigure]{graphfig}

\newtheorem{theorem}{Theorem}

\newtheorem{corollary}{Corollary}
\newtheorem{prop}{Proposition}

 \usepackage{geometry}
\usepackage{xcolor}
\allowdisplaybreaks
\usepackage{xcolor}
\begin{document}

\title{Ambient Backscatter-Assisted Wireless-Powered Relaying %for Wireless IoT %with ing: 
%Performance Modeling and Analysis%
%\vspace{-3mm}
}

%\begin{comment}
\author{Xiao Lu, Dusit Niyato, Hai Jiang, Ekram Hossain and Ping Wang \\ 
 \thanks{\hspace{-5mm} Some initial results of this paper have been accepted by the IEEE International Conference on Communications in 2018~\cite{X.May_2018Lu}}
}
%\end{comment}

\markboth{}{Shell \MakeLowercase{\textit{et al.}}: Bare Demo of
IEEEtran.cls for Journals}
  
\maketitle

%\vspace{-10mm}

\begin{abstract}
Internet-of-Things (IoT) is featured with low-power communications among a massive number of ubiquitously-deployed and energy-constrained electronics, e.g., sensors and actuators. To cope with the demand, wireless-powered cooperative relaying emerges as a promising communication paradigm %approach 
to extend data transmission coverage and solve energy scarcity for the IoT devices. In this paper, we propose a novel hybrid relaying strategy by combining wireless-powered communication and ambient backscattering functions to improve applicability and performance of data transfer. In particular, the hybrid relay can harvest energy from radio frequency (RF) signals and use the energy for active transmission. Alternatively, the hybrid relay can choose to perform ambient backscattering of incident RF signals for passive transmission. %To well adapt the hybrid relay to the network environments, 
To efficiently utilize the ambient RF resource, we design mode selection protocols to coordinate between the active and passive relaying in circumstances with and without instantaneous channel gain. With different mode selection protocols, we characterize the success probability and ergodic capacity of a dual-hop relaying system with the hybrid relay in the field of randomly located ambient transmitters. 
%The analytical models effectively reveal some important impacts of system parameters, thus enhancing understanding of the network operation and enabling performance optimization. 
The analytical and the numerical results demonstrate the effectiveness of the mode selection protocols in adapting the hybrid relaying into the network environment and reveal the impacts of system parameters on the performance gain of the hybrid relaying. As applications of our analytical framework which is computationally tractable, we formulate optimization problems based on the derived expressions to optimize the system parameters with different objectives. The optimal solutions exhibit a tradeoff between the maximum energy efficiency and target success  probability.
%the energy harvesting time and transmit power with the objective to maximize the system energy efficiency, which is the normalized ergodic capacity over the transmit power. The optimal solutions exhibit a tradeoff between the maximum energy efficiency and target success  probability.

\end{abstract}

\emph{Index terms- Large-scale IoT, wireless-powered relaying, ambient backscatter communications, %wireless energy harvesting,
performance analysis, energy efficiency}.

%\vspace{-3mm}

\section{Introduction}

Wireless-powered communication \cite{Q.Nov.2015Wu,H.April2016Lee,Q.Jul.2016Wu,E.May_2017Boshkovska,D.May2016Niyato} is a green networking paradigm that allows wireless devices to harvest energy from ambient signals, especially radio-frequency (RF) waves. The harvested energy is then used for data transmission without requiring an external energy supply. Therefore, it is considered to be an enabling technology to support %self-sustaintable
machine-type communications (MTC) for Internet-of-Things (IoT) applications \cite{S.April2016Bi,X.April2015Lu,D.Aug2016Niyato}, such as smart sensors and wearable devices~\cite{Y.2014Toh}.
%due to the limited resources of wireless-powered devices, relay transmission is particularly necessary for IoT to save the transmit power and increase the reliability of communication networks
%In traditional wireless networking, 
It has also been shown that for traditional wireless networking, cooperative relaying can improve the reliability of wireless data transmission~\cite{X.2014Lu}. This concept is particularly attractive for wireless-powered communication, due to the resource limitation of the devices. In this regard, %the same concept can be applied to
wireless-powered relaying has attracted a great deal of recent research attention. % has been attracted into this technique. %For example, 
Existing literature \cite{Y.April_2015Zeng,Y.May_2016Zeng,L.2016Tang,Z.2017Chen} has shown that the use of wireless-powered relaying can provide remarkable performance gains in terms of spatial efficiency and energy efficiency.

However, employing wireless-powered relaying poses some notable challenge. Specifically, a wireless-powered relay involves active transmission, i.e., RF signals are actively generated at the RF front-end. %, which consumes high circuit power consumption. 
This will require the relay to adopt power-hungry analog circuit components such as oscillator, amplifier, filter, and mixer~\cite{V.2013Liu}.
Consequently, the relay may need to spend a long time to harvest and accumulate energy for the relaying operation~\cite{X.2nd2015Lu}. The situation becomes much worse when the ambient RF signals are weak or intermittent. 
%Moreover, due to active transmission of the relay, co-channel interference can adversely affect the relay transmission, degrading the benefit of the cooperative relay considerably. 
To solve such issues, some existing literature mainly emphasizes on designing network protocols
% \cite{Z.2017Chen,Z.2017Zheng,A.2015Nasir} 
to utilize network resources more efficiently. For example, the authors in~\cite{H.April2015Chen} introduced a harvest-then-cooperate protocol for the relay to optimize the time allocation for energy harvesting in a downlink and information forwarding in an uplink to maximize throughput. 
%Reference~\cite{A.2015Nasir} designed time-switching based protocols for a decode-and-forward relay and an amplify-and-forward relay powered by energy harvesting to maximize throughput. 
These network protocols %are to optimize the use of available wireless resources in a more effective way. 
manage to enhance the performance of a wireless-powered relay under different system configurations. However, the performance is still limited by the high circuit power consumption because of the active transmission nature.
%However, the classical solutions are still limited by the active transmission nature of cooperative relaying.

%\subsection{Related Work}

%\subsection{Motivation and Contributions} 
 
To tackle the power consumption limitation caused by active transmission, the work in~\cite{X.2017Lu} %introduce an approach of improving 
proposed to integrate the wireless-powered transmission %by %capitalizing on
%integrating 
with 
a passive communication technology, called ambient backscattering. Ambient backscattering, first introduced in~\cite{V.2013Liu}, can transmit data 
by passively reflecting the existing RF signals in the air.
In particular, ambient backscatter is generated by a
load modulator. The load modulator tunes its load impedance to match and mismatch with that of the antenna to absorb and reflect the incident RF signals. Either an absorbing state or reflecting
state can be used to indicate a ``0" or ``1" bit to the backscatter receiver.
The incident RF signals can be from existing RF sources such as TV stations, cellular base stations, and Wi-Fi access points. There are three major advantages of ambient backscattering. Firstly, an ambient backscatter transmitter does not actively generate any RF signals and thus require much lower operational power compared with traditional wireless-powered communication based on active transmission\footnote{Ambient backscatter communication typically incurs a few micro-Watts %to hundreds of micro-Watts 
power consumption rate at the transmitter side and can achieve a transmission rate ranging from 1 kbps to hundreds of kbps~\cite{V.2013Liu,N.2014Parks}.}. Secondly, an ambient backscatter transmitter utilizes ambient emitters as its carrier sources. %Thus, it is applicable where ambient RF signals exist. 
Thus, it does not require a dedicated spectrum.
Thirdly, an ambient backscatter receiver performs data decoding by averaging the received signals plus modulated backscatter and detect the variations therein for demodulation. Thus, self-interference from the ambient emitters does not significantly degrade the performance of ambient backscatter communication. 
Therefore, ambient backscattering becomes an effective alternative transmission approach to 
improve the applicability and energy efficiency of wireless-powered devices~\cite{V.2018Huynh}.
%wireless-powered transmission, especially when the harvested energy is insufficient or when the interference is high~\cite{X.2017Lu}.  

Reference~\cite{T.2017Hoang} investigated a cognitive radio network where the wireless-powered secondary user is equipped with ambient backscattering capability. 
In this network, when the primary user is on transmission, the secondary user can either harvest energy or perform ambient backscattering. When the primary user is off, the secondary user can perform active transmission with the harvested energy. 
Optimal transmission policies are designed to maximize the throughput of the secondary network. 
Additionally, 
the work in~\cite{X.March2018Lu} investigated an integrated wireless-powered transmitter with ambient backscattering for device-to-device communication. %The performance of two simple mode selection protocols based on the local information of the transmitter and feedback information from receiver, respectively, are examined.
It was shown that ambient backscattering %can assist 
is suitable to work as an alternative to wireless-powered transmission when the ambient energy is not sufficient and/or when the interference level is high.   
However, to the best of our knowledge, none of the existing studies has investigated cooperative relaying enabled with both wireless-powered transmission and ambient backscattering.

%===================================

%These merits motivate us to adopt ambient backscatter to assist a wireless-powered relay when active transmission is not an optimal option. 

This paper investigates a dual-hop cooperative relay system with a hybrid relay that integrates wireless-powered communication and ambient backscattering for data forwarding. %The hybrid relay adopts a dual-hop decode-and-forward cooperative relaying scheme. 
Specifically, the hybrid relay first harvests energy from the transmission of %surrounding transmitters, i.e., carrier
ambient emitters, then decodes and forwards the information from the source node to the destination node through either  
wireless-powered communication or ambient backscattering. The former and the latter are referred to as the wireless-powered relaying (WPR) mode and  ambient backscatter relaying (ABR) mode, respectively.
As the power for both circuit operation and transmission comes from the ambient RF signals, the performance of the hybrid relay largely depends on the environment factors, e.g., spatial density, distribution, and transmit power of surrounding transmitters and how the relay adapts to them.
To analyze and gain insight into operation of the hybrid relaying, % together with the mode selection protocols,
we develop analytical models built upon stochastic geometry analysis. 
In our system model, the ambient emitters and interferers are spatially randomly located following general point processes, i.e., $\alpha$-Ginibre point processes (GPPs)~\cite{L.2015Decreusefond}.  
%Compared with the widely adopted PPP, $\alpha$-GPP characterizes the spatial correlation among randomly located points by a repulsion factor $\alpha \in (0,1]$.
The processes are able to capture the spatial repulsion among the randomly located points by a repulsion factor $\alpha \in (0,1]$ in which the typical Poisson point process (PPP) cannot do\footnote{In practical communication systems, transmitters may not be located very close to each other because of network planning, e.g., for interference mitigation, and physical constraints, e.g., obstacles. 
Due to the tractability in characterizing the correlation among the spatial points, 
$\alpha$-GPP has recently being applied to model different networks with repulsion, e.g., cellular networks~\cite{X.2016JLu} and 
wireless-powered networks~\cite{I.May_2015Flint}.
These references have demonstrated that spatial repulsion can cause a non-trivial impact on network performance.}. 
As aforementioned, the performance of the hybrid relay heavily relies on its ambient environment. Therefore, the repulsion among the ambient transmitters is an important factor that cannot be ignored. %In this work, we choose $\alpha$-GPP for its tractability .
In this paper, we adopt $\alpha$-GPP due to its flexibility to evaluate different degrees of repulsion. %Therefore, it is imperative to devise mode selection protocols %operational protocols to coordinate the hybrid relay among energy harvesting, wireless-powered transmission and ambient backscattering based on the system environment.
%To this end,
%The mode selection is based on the system conditions, i.e., the instantaneously detected energy harvesting rate at the relay and receive signal-to-noise-plus-interference ratio (SINR) at the destination. 

%For the performance analysis, we consider two practical scenarios of the interferers, i.e., the quasi-static scenario and fast-varying scenario~\cite{Z.Feb.2009Win}.  In particular, with quasi-static interference and fast-varying interference, respectively, the source-to-relay transmission and relay-to-destination transmission are affected by the same and different set of the interferers.  For both scenarios, 
Based on the $\alpha$-GPP modeling and analytical framework, 
our main contributions in this paper are summarized as follows.

\begin{itemize}

\item  We propose  
two mode selection protocols that dynamically instruct the hybrid relay to %perform different relaying modes 
switch between wireless-powered communication and ambient backscattering
for the cases with and without instantaneous information of the system. %operate %in these two modes. 
For the former case, % with prior information,
the hybrid relay selects the mode based on predefined knowledge and current network conditions. For the latter case, %For case without prior information, 
the hybrid relay %uses previously detected 
explores network conditions and uses the learned knowledge to infer future mode selection. 

\item We derive analytically and computationally tractable expressions 
to characterize the end-to-end success probability and ergodic capacity of the hybrid relaying system. 
We demonstrate analytically and numerically that, with the proposed mode selection protocols, wireless-powered communication and ambient backscattering can  complement each other to achieve improved performance for cooperative relaying.  
In particular, for the case with instantaneous channel state information (CSI), the hybrid relaying  strictly outperforms both the WPR and the ABR. For the case without instantaneous CSI, the hybrid relaying renders a better performance than the cooperative relaying that randomly selects between the WPR and the ABR. %worse one between a wireless-powered relay and an ambient backscatter relay.

\item We demonstrate an example of our analytical framework in %optimizing the transmit power of the source node and energy harvesting time
optimal resource allocation 
of the hybrid relaying system with an objective to maximize the energy efficiency (in bits/Joule), which is the normalized ergodic capacity averaged over the transmit power of the source node. The optimal solutions reveal the tradeoff between the maximum energy efficiency and the achieved success probability. %of the relaying system

\end{itemize}
%Some insightful result is discovered from the models, for example, the energy harvesting efficiency can be adjusted to maximize the capacity. 

%{\bf Modify}
  
The remainder of the paper is organized as follows. Section II introduces the system model and the modeling framework. % describes the fundamental properties of the adopted analytical framework
Section III designs mode selection protocols of the hybrid relaying. Section IV and Section V perform the analysis of the hybrid relaying in terms of end-to-end success probabilities and ergodic capacity, respectively, in various settings. %Section VI demonstrates an application of our analytical framework in maximizing the system energy efficiency. 
Section VI %validates our analytical results via Monte Carlo simulations and 
presents the numerical results and demonstrates an applications of our analytical framework in optimizing the system parameters with different objectives. Finally, Section VII concludes our work.

%We measure the performance of the hybrid relaying in term of two important metrics, namely, energy outage probability, coverage

{\bf Notations:}  In this paper, %we use $\mathbb{E}[\cdot]$ to denote the average over all the random variables in $[\cdot]$, $\mathbb{E}_{X}[\cdot]$ to denote the expectation over the random variable $X$, and $\mathbb{P}[Z]$ to denote the probability that an event $Z$ occurs.
%the operators 
$\mathbb{E}_{X}[\cdot]$ and $\mathbb{E}[\cdot]$ are used to denote the average over the random variable $X$ and all the random variables in $[\cdot]$, respectively. $\mathbb{P}[E]$ denotes the probability that an event $E$ occurs and $\bf{1}$$_{E}$ is an indicator function that equals 1 if event $E$ occurs and equals 0 otherwise. 
$\mathbf x_{a}$ represents the location of $a$ and $\bar{\mathbf x}$ denotes the conjugate of a complex scalar $\mathbf x$. $i$ is the imaginary unit, i.e., $i=\sqrt{-1}$. 
The operator $\|\cdot \|$ represents the Euclidean 2-norm. % between the coordinates of $a$ and $b$.
Besides, let ${n \choose i}=\frac{n(n-1)\cdots(n-i-1)}{i(i-1)\cdots 1}$ denote the binomial coefficient and  $[a,b]^{+}\triangleq\max(a,b)$.

\begin{figure}
\centering
\includegraphics[width=0.5\textwidth]{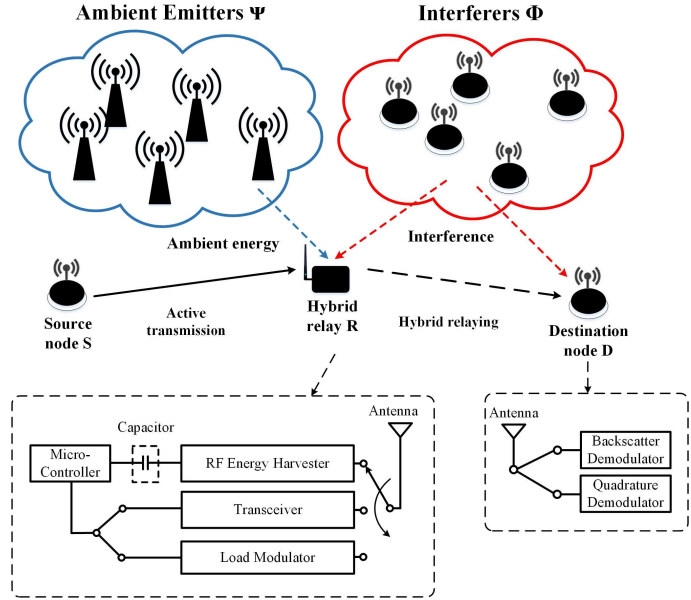} %{multihop_models.eps} %\vspace{-5mm}
\caption{A dual-hop relaying system in an $\alpha$-GPP field of ambient emitters and interferers.} \label{fig:system_model}
%\vspace{-5mm}
\end{figure}

%\vspace{-2mm}

\section{System Model and Stochastic Geometry Characterization} 
%\subsection{Network Model}
\label{sec:network_model}
%in coexistence with a large number of randomly distributed interferers $\Phi$

This section introduces the system model under consideration and then describes some %fundamental properties 
preliminary results of the adopted stochastic geometry framework. 

%\vspace{-3mm}

\subsection{System and Relaying Protocol}
\label{subsection:system_model}

%The relay selection is beyond the scope of this paper
 
%We consider an IoT scenario where exists a massive number of sensors that have communication demand. A transmitting sensor has limited transmit power and may rely on another device to relay its transmission to the destination device. %The vacant sensor is willing to relay the transmission only with the harvested energy, not the energy from its own battery. 
%Due to the self-sustainability feature, our proposed hybrid relay can be used to enable such a cooperative relaying. 
As shown in Fig.~\ref{fig:system_model}, we consider that a source node $\mathrm{S}$ needs to transmit to a destination node $\mathrm{D}$ which 
is located far away %from $\mathrm{S}$ 
and cannot achieve a reliable direct communication under its transmit power budget~\cite{A.2014Aalo}. As such, a hybrid relay  $\mathrm{R}$ composed of an active transceiver and an ambient backscatter transmitter is adopted to forward the information from $\mathrm{S}$ to $\mathrm{D}$\footnote{An ambient backscatter transmitter can be combined with an active transceiver with low implementation complexity and cost~\cite{X.2017Lu}. Only an additional load modulator is needed to be integrated to the RF front-end. Other components such as an antenna, capacitor, and logic circuit can be shared from the active transceiver circuit. A recent prototype with both Bluetooth and backscattering functions,  namely {\em Briadio}~\cite{P.2016Hu}, has demonstrated the feasibility of such an integration. }. 
Accordingly, %to decode the transmission from $\mathrm{R}$ in both modes, 
$\mathrm{D}$ is equipped with both a quadrature demodulator and backscatter demodulator to decode the  information from $\mathrm{R}$.  
All three nodes $\mathrm{S}$, $\mathrm{R}$, and $\mathrm{D}$ are equipped with a single antenna and work in a half-duplex fashion.  Similar to the system models in  \cite{A.2013Nasir,A.2015Nasir}, both $\mathrm{S}$ and $\mathrm{D}$ are assumed to have sufficient energy to supply their operations. Thus, the focus is on $\mathrm{R}$ which is equipped with an energy harvester and an on-board 
capacitor with capacity $E_{C}$
to store the harvested energy for the relaying operation. We assume the harvested
energy can only be consumed in the current time slot and there is no energy left for the subsequent time slot due to capacitor leakage~\cite{K.Dec2013Huang,I.2014Krikidis}.

%To enable the hybrid relaying,
As shown in the left block diagram in Fig.~\ref{fig:system_model}, 
a time switching-based receiver architecture~\cite{XLuSurvey} is adopted to enable $\mathrm{R}$ %coordinate among 
to work in either
energy harvesting, information decoding or transmission.
$\mathrm{R}$ performs energy harvesting/ambient backscattering and active transmission on two different frequency bands. For example, the hybrid relay can be designed to harvest energy and perform ambient backscattering by using the downlink signals of ambient base stations (BSs) %one one band 
while the hybrid relay can actively transmit data on %another band which is the same as that for 
the operating frequency of
the source node $\mathrm{S}$. Let $\Psi$ and $\Phi$ denote the point processes, i.e., the sets, of ambient emitters and the %point process of
interferers, respectively. The ambient emitters are the surrounding transmitters working on the frequency band that the hybrid relay performs energy harvesting/ambient backscattering. The interferers are the surrounding transmitters working on the band that the hybrid relay performs active transmission. Let $\widetilde{\zeta}$ ($\zeta$ ) and $\widetilde{P}_{T}$ ($P_{T}$) denote the spatial density and transmit power of the transmitters in $\Psi$ ($\Phi$), respectively. We assume that the distributions of $\Psi$ and $\Phi$ follow independent homogeneous $\alpha$-GPPs~\cite{L.2015Decreusefond}, which are of a general class of repulsive point process. The mathematical details about $\alpha$-GPP can be found in~\cite{L.2015Decreusefond,L.August2015Decreusefond}.

Fig.~\ref{fig:relay_protocol} shows the time slots of the relaying protocol adopted by the considered hybrid relaying system to 
coordinate among energy harvesting and the dual-hop relaying. Let $T$ denote the duration of a time slot. At the beginning of each time slot, an $\omega$ ($0 \!<\! \omega \!<\! 1$) fraction of $T$ is allocated for $\mathrm{R}$ to harvest ambient energy. % from the ambient emitters in $\Psi$. 
Then, the first half and second half of the remaining time (i.e., with duration $\frac{(1-\omega) T}{2}$) are allocated for the source-to-relay transmission and relay-to-destination transmission, respectively. During the relay-to-destination transmission, %$\mathrm{R}$ operates in the mode chosen by the mode selection protocols. %chooses to operate between the WPR mode and the ABR
%With the proposed architecture as shown in Fig.~\ref{fig:system_model}, $\mathrm{R}$ is capable of forwarding the information from $\mathrm{S}$ to $\mathrm{D}$ through either
%WPR or ABR, denoted as $\mathrm{W}$ and   $\mathrm{A}$\footnote{$\mathrm{W}$ and $\mathrm{A}$ appear in the subscripts of variables to represent the operation mode.}, respectively. 
%Under different system environments,
%Enabled with the hybrid relaying capability,
$\mathrm{R}$ %forwards the information from $\mathrm{S}$ to $\mathrm{D}$ through 
operates in either the WPR mode or the ABR mode, denoted as $\mathrm{W}$ and   $\mathrm{A}$\footnote{$\mathrm{W}$ and $\mathrm{A}$ appear in the subscripts of variables to represent the operation mode.}, respectively, based on the system environment. The mode selection protocols are to be introduced in Sec.~\ref{sec:MP}. 
If the WPR mode is chosen, $\mathrm{R}$ adopts the decode-and-forward protocol for relaying the data from $\mathrm{S}$ to $\mathrm{D}$ which incurs a circuit power consumption of $E_{\mathrm{W}}  \! < \! E_{C}$ per time slot. Otherwise, $\mathrm{R}$ first decodes the active transmission from $\mathrm{S}$ and then forwards the decoded message to $\mathrm{D}$ through ambient backscattering which incurs a circuit power consumption of $E_{\mathrm{A}} \! \ll \! E_{\mathrm{W}}$ per time slot.
The hybrid relay is able to function in the WPR mode and ABR mode only if the harvested energy during the energy harvesting phase of each time slot $E_{\mathrm{R}}$ exceeds $E_{\mathrm{W}}$ and $E_{\mathrm{A}}$, %(or equivalently, $Q_{\mathrm{R}}$ exceeds $\varrho_{\mathrm{W}}$ and $\varrho_{\mathrm{A}}$),
respectively.

\begin{figure*}
\centering
\includegraphics[width=1\textwidth]{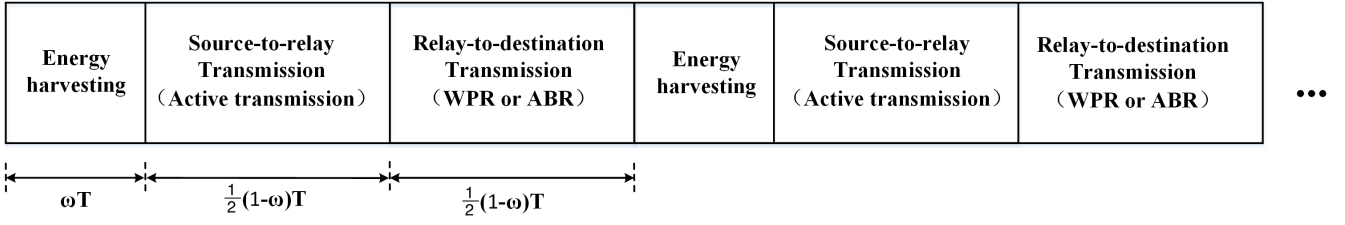} %{multihop_protocol.eps} 
%\vspace{-8mm}
\caption{A time-slot based relaying protocol.} \label{fig:relay_protocol}
%\vspace{-8mm}
\end{figure*}

\subsection{Channel Model and SINR/SNR Formulation}
All the channels in the systems are considered to experience independent and identically distributed (i.i.d.) block Rayleigh fading. In particular, the fading channel gains follow exponential distributions. Besides, we consider i.i.d. zero-mean additive white Gaussian noise (AWGN) with variance $\widetilde{\sigma}^2$ and $\sigma^2$ on the transmit frequency of $\Psi$ and $\Phi$, respectively.

\subsubsection{Harvested Energy} 
The hybrid relay harvests energy from the ambient emitters in $\Psi$ for its operation. 
The power of the RF signals from $\Psi$ received by the hybrid relay can be computed as
\begin{equation} \label{eqn:Q_R}
	Q_{\mathrm{R}}= \widetilde{P}_{T} \sum_{\mathbf{x}_{k} \in \Psi} \widetilde{h}_{k,\mathrm{R}} \|\mathbf{x}_{k}-\mathbf{x}_{\mathrm{R}}\|^{-\widetilde{\mu}}, 
%\vspace{-2mm}
\end{equation}
where $\widetilde{h}_{a,b}\sim \exp(1)$ is the fading channel gain between nodes $a$ and $b$ on the transmit frequency of $\Psi$, and $\widetilde{\mu}$ represents %the path loss exponent for the RF carrier from $\Phi$.
%$\widetilde{\mu}$ is 
the path-loss exponent for the signals from the ambient emitters in $\Psi$.

Then, the instantaneous energy harvesting rate (in Watt) at the hybrid relay can be expressed as $ \rho_{\mathrm{R}} = \beta  Q_{\mathrm{R}}$, where $\beta \in (0,1]$ denotes the RF-to-DC conversion efficiency. And the amount of energy harvested during the energy harvesting phase is obtained as $E_{\mathrm{R}} = \omega  T %\omega \beta Q_{\mathrm{R}}
\rho_{\mathrm{R}}$. 

\subsubsection{Source-to-Relay Transmission} 

During the source-to-relay transmission phase, the receive SINR at the relay node $\mathrm{R}$ is given by %~\cite{X.2013Zhou}
\begin{equation} 
	\nu_{\mathrm{R}}= \frac{ P_{\mathrm{S}} h_{\mathrm{S},\mathrm{R}} d^{-\mu }_{\mathrm{S},\mathrm{R}} }{  I_{\mathrm{R}} + \sigma^2  }, 
\end{equation}
where $P_{\mathrm{S}}$ is the transmit power of the source node, $h_{a,b}\sim \exp(1)$ denotes the fading channel gain between $a$ and $b$ on the transmit frequency of $\Phi$, $d_{x,y}$ represents the distance between $x$ and $y$, $\mu$ is the path-loss exponent for signals from the interferers $\Phi$, $\mathrm{S}$ and $\mathrm{R}$, $I_{\mathrm{R}} = \sum_{j \in \Phi} P_{T} h_{j,\mathrm{R}} \|\mathbf{x}_{j}-\mathbf{x}_{\mathrm{R}}\|^{-\mu} $ is the aggregated interference received at $\mathrm{R}$.

\subsubsection{Relay-to-Destination Transmission} 
Let $\rho_{\mathrm{W}}= \frac{E_{\mathrm{W}}}{T}$ and $\rho_{\mathrm{A}}= \frac{E_{\mathrm{A}}}{T}$ denote the average circuit power consumption rate of $\mathrm{R}$ in 
the WPR mode and the ABR mode, respectively, and $\rho_{C} \triangleq \frac{E_{C}}{T}$ denotes the normalized energy capacity over a time slot duration.  
Moreover, let us define $\varrho_{\mathrm{W}} \triangleq \frac{ \rho_{\mathrm{W}}  }{\omega \beta }$, $\varrho_{\mathrm{A}} \triangleq \frac{ \rho_{\mathrm{A}}  }{\omega \beta }$ and $\varrho_{C} \triangleq \frac{\rho_{C} }{\omega \beta }$.
If the WPR mode is chosen, %all stored energy is considered to be consumed 
$\mathrm{R}$ uses all the available energy for active transmission during the relay-to-destination transmission phase. Then, the transmit power can be calculated as: 
\begin{comment}
\begin{equation} 
\label{eqn:AR_R}
 	P^{A}_{\mathrm{R}} \!=\!
	\begin{cases}
  0 & \text{if} \quad \!\!\! Q_{\mathrm{R}} \!<\! \frac{ \rho_{\mathrm{A}} (N+1) }{\beta (N-1)}, \\
  \frac{ \frac{ (N-1) T \beta Q_{\mathrm{R}} }{(N+1)} \! - \! E_{A} }{ \frac{2T}{N+1} } = \frac{1}{2} \big( (N\!-\!1) \beta Q_{\i} - \rho_{\mathrm{A}} (N\!+\!1) \big) & \text{if} \quad \!\!\! \frac{ ( \rho_{\mathrm{A}}) (N+1) }{\beta (N-1)} \!\leq \! Q_{\mathrm{R}} \! \leq \! \frac{ (\rho_{\mathrm{A}} + \rho_{C}) (N+1) }{\beta (N-1)}, \\
	\frac{E_{C}}{\frac{2T}{(N+1)}} = \frac{ \delta_{C} (N+1) }{2} & \text{ otherwise},  
	\end{cases}
\end{equation}
\end{comment}

\begin{equation} 
\label{eqn:AR_R_dual}
 	P^{\mathrm{W}}_{\mathrm{R}} \!=\!
	\begin{cases}
  0 & \text{if} \quad \!\!\! Q_{\mathrm{R}} \! \leq \! \varrho_{\mathrm{W}}, \\
  \frac{ T \omega \beta Q_{\mathrm{R}}  - E_{\mathrm{W}} }{ \frac{T(1-\omega)}{2} } = \frac{2(\omega \beta Q_{\mathrm{R}} - \rho_{\mathrm{W}})}{ 1-\omega} & \text{if} \quad \!\!\! %\frac{ \rho_{\mathrm{W}}}{ \omega \beta} 
  \varrho_{\mathrm{W}} \! < \! Q_{\mathrm{R}} \! \leq \! 
  \varrho_{\mathrm{W}} + \varrho_{C}
  %\frac{ \rho_{\mathrm{W}} + \rho_{C} }{\omega \beta }
  , \\
	\frac{E_{C}}{\frac{T(1-\omega)}{2}} = \frac{2\rho_{C}}{1-\omega}  & \text{ otherwise},  
	\end{cases}
\end{equation}
%where and . %Physically, $2\rho_{C} $ indicates the maximum discharging rate of the capacitor.

Correspondingly, the receive SINR at $\mathrm{D}$ is 
\begin{equation}
\nu^{\mathrm{W}}_{\mathrm{D}}= \frac{ P^{\mathrm{W}}_{\mathrm{R}} h_{\mathrm{R},\mathrm{D}} d^{-\mu}_{\mathrm{R},\mathrm{D}} }{ \sum_{j \in \Phi} P_{T} h_{j,\mathrm{D}} \|\mathbf{x}_{\mathrm{R}}-\mathbf{x}_{\mathrm{D}}\|^{-\mu} \! + \! \sigma^2  }.  
\end{equation}

Then, the end-to-end capacity achieved can be computed as 
 \begin{equation} 
 	C_{\mathrm{W}} \!=\!
	\begin{cases}
  \frac{(1-\omega)}{2}W \log_{2} \big ( 1 + \nu \big) & \text{if} \quad \!\!\! \nu \!\geq\! \tau_{\mathrm{W}}, \\
  0 & \text{ otherwise},  
	\end{cases}
\end{equation} 
where $W$ denotes the transmission bandwidth of $\mathrm{R}$ in the WPR mode, $\nu = \min( \nu_{\mathrm{R}}, \nu^{\mathrm{W}}_{\mathrm{D}})$, and $\tau_{\mathrm{W}}$ is the minimum SINR threshold to decode data from active transmission at $\mathrm{R}$ and $\mathrm{D}$.

Otherwise, if $\mathrm{R}$ decides to forward data through ambient backscattering, the transmit power of modulated backscatter can be calculated as~\cite{C.2012Boyer} %\vspace{-3mm}
\begin{equation}  \label{eqn:MB_R}
	P^{\mathrm{A}}_{\mathrm{R}} = \eta \xi Q_{\mathrm{R}}, %\vspace{-3mm}
\end{equation}
where $\eta$ is the fraction of the incoming RF signals reflected during backscattering and $\xi \in (0,1]$ is the backscattering efficiency of the transmit antenna \cite{S.2020Gurucharya}. $\xi$ represents the portion of the reflected signals that are effectively used to carry the modulated data. The values of $\eta$ and $\xi$ are dependent on the tag-encoding scheme~\cite{C.2014Boyer} and tag antenna aperture~\cite{Y.2009Kim}, respectively. 
 
The receive signal-to-noise ratio (SNR) from ambient backscatter at $\mathrm{D}$ is 
\begin{align}
	\nu^{\mathrm{A}}_{\mathrm{D}} = \frac{ P^{\mathrm{A}}_{ \mathrm{R} } \widetilde{h}_{\mathrm{R},\mathrm{D}} d^{-\mu}_{\mathrm{R},\mathrm{D}} }{\widetilde{\sigma}^2}. 
\end{align}
 
Due to the adopted simple amplitude demodulation based on envelope detection~\cite{V.2013Liu}, an ambient backscatter receiver typically requires much higher receive SNR to achieve a low bit error rate compared with a quadrature demodulator. For example, a bit error rate around $10^{-3}$ can be obtained with $10$ dB receive SNR~\cite{J.Nov.2012Kimionis}.    
Once the receive SNR at $\mathrm{D}$ is above the threshold $\tau_{\mathrm{A}}$, a pre-designed channel capacity $C_{\mathrm{A}}$ can be attained over the relay-to-destination link. The value of $C_{\mathrm{A}}$ is dependent on the encoding scheme adopted and the setting of the resistor-capacitor of the ambient backscattering circuit~\cite{V.2013Liu,T.2017Hoang}. As ambient backscattering usually adopts very simple modulation schemes, such as amplitude shift keying and phase shift keying, it is reasonable to assume that $C_{\mathrm{A}} < C_{\mathrm{W}}$.

Table~\ref{notation} summaries the main notations used in this paper.

\begin{table*}
\centering
\caption{\footnotesize NOTATIONs.} \label{notation} %\vspace{-3mm}
\begin{tabular}{|l|l|}
\hline
Symbol & Definition\\ \hline
\hline
 $\Psi$, $\Phi$ & The point processes representing the ambient emitters and interferers, respectively \\
 $\widetilde{\alpha}$, $\alpha$ & Repulsion factors for $\Psi$ and $\Phi$, respectively \\
 $\widetilde{P}_{T}$, $P_{T}$ & The transmit power of the transmitters in $\Psi$ and $\Phi$, respectively \\ 
 %$\widetilde{\mu}$, $\mu$ & Pass-loss exponent for the signals from $\Psi$ and $\Phi$, respectively \\   
 $P_{\mathrm{S}}$ & The transmit power of the source node $\mathrm{S}$ \\
% $E_{\mathrm{A}}$, $E_{\mathrm{B}}$ & The circuit power consumption of the hybrid relay in a time slot with wireless-powered relaying\\
% &  and ambient backscatter relaying, respectively \\
% $L$ & The number of antennas at a massive MIMO-enabled BS \\ 
$\widetilde{\sigma}^{2}$, $\sigma^{2}$ & The variance of AWGN in the transmit frequency of $\Psi$ and $\Phi$, respectively \\
$E_{\mathrm{R}}$ & The amount of harvested energy at the hybrid relay $\mathrm{R}$    \\
$E_{C}$ & The capacitor capacity of $\mathrm{R}$    \\
%$W$ & The bandwidth of the channel between the access point and RF-powered device \\
$\xi$ & Backscattering efficiency \\
$\nu^{\mathrm{W}}_{\mathrm{R}}$, $\nu^{\mathrm{A}}_{\mathrm{R}}$ & The receive SINR and SNR at $\mathrm{R}$ from active transmission and ambient backscatter, respectively \\
$\tau_{\mathrm{W}}$, $\tau_{\mathrm{A}}$ & The SINR threshold to decode from active transmission and ambient backscatter, respectively \\
\hline
\end{tabular}
%\vspace{-8mm}
\end{table*}

\subsection{Preliminaries}
This subsection describes some %fundamental features and properties 
primarily results 
of $\alpha$-GPP modeling which are applied later in the analysis of this paper.
We consider that the hybrid relay $\mathrm{R}$ locates at the origin of the Euclidean space surrounded by the $\alpha$-GPP distributed ambient emitters $\Psi$ with transmit power $\widetilde{P}_{T}$.  %randomly located following an $\alpha$-GPP %introduced in Section~\ref{subsection:system_model}
%where the spatial points represent transmitters with transmit power $\widetilde{P}_{T}$. 
In the Rayleigh fading environment, the distribution of the received signal power from $\Psi$ at $\mathrm{R}$, i.e., $Q_{\mathrm{R}}$ in (\ref{eqn:Q_R}), are presented in the following proposition.
 
\begin{prop} {\em \cite[Theorem 1]{X.March2018Lu} }
The probability density function (PDF) and cumulative distribution function (CDF) of $Q_{\mathrm{R}}$ are given, respectively, as: %\vspace{-2mm}
\begin{align}   \label{pdf} 
f_{Q_{\mathrm{R}}}(q) =  \mathcal{L}^{-1}\Big\{ \mathrm{Det}\big(\mathrm{Id}+\widetilde{\alpha}  %\mathbb{Q}_{\Psi}(s)
\mathbb{G}_{\Psi}(\mathbf{x},\mathbf{y}) 
\varpi_\mathbf{x}(s)  \varpi_\mathbf{y}(s) 
\big)^{\!-\frac{1}{\widetilde{\alpha}}}\Big\}(q),
\end{align} 
%\vspace{-10mm}
\begin{align} \label{CDF} 
\text{and} \quad & F_{Q_{\mathrm{R}}}(q) %\int^{x}_{-\infty} f_{Q_{\mathrm{R}}}(q) \mathrm{d}q = \int^{x}_{-\infty} \mathcal{L}^{-1}\left\{\mathcal{L}_{P_{T}}(s)\right \}(q) \mathrm{d}q \nonumber \\ & 
 = \mathcal{L}^{-1} \!\left \{ \! \frac{1}{s} \mathrm{Det}\big(\mathrm{Id}+\widetilde{\alpha} %\mathbb{Q}_{\Psi}(s)
 \mathbb{G}_{\Psi}(\mathbf{x},\mathbf{y}) 
 \varpi_\mathbf{x}(s)  \varpi_\mathbf{y}(s) 
 \big)^{-\frac{1}{\widetilde{\alpha}} } \!\right \} \! (q),
 \end{align} 
where $\mathcal{L}^{-1}\{\cdot\}(x)$ represents the inverse Laplace transform which can be evaluated by the Mellin's inverse formula~\cite{P.1995Flajolet}, %and $\mathbb{Q}_{\Psi}(s)$ is $\mathbb{Q}_{\Psi}(s) \! \triangleq  \psi_{1}(\mathbf{x})  G_{\Psi}(\mathbf{x},\mathbf{y}) \psi_{1}(\mathbf{y})  $
\begin{comment}
\begin{align} 
\label{eq:Q}
 \hspace{2mm} \mathbb{Q}_{\Psi}(s) \! \triangleq \! \! \sqrt{\!1\!-\!\Big( \! 1\!+\! s \widetilde{P}_{T} \|\mathbf{x} \|^{-\widetilde{\mu}} % \frac{s \widetilde{P}_{T}}{\|\mathbf{x} \|^{\widetilde{\mu}}} 
 \Big)^{\!\! -1}} 
\, G_{\Psi}(\mathbf{x},\mathbf{y})  
\sqrt{\!1\!-\! \Big(\! 1 \!+\! % \frac{s \widetilde{P}_{T}}{ \|\mathbf{y} \|^{\widetilde{\mu}}} 
s \widetilde{P}_{T} \|\mathbf{y} \|^{-\widetilde{\mu}}  \! \Big)^{ \! \!-1} }, \hspace{2mm}
\mathbf x,\mathbf y \in \Psi.
\end{align}
\end{comment}
$ \varpi_\mathbf{z}(s) \! \triangleq \! \sqrt{\!1\!-\!\big(   1\!+\! s \widetilde{P}_{T} \|\mathbf{z} \|^{-\widetilde{\mu}} \big)^{\!-1} } $, and $\mathbb{G}_{\Psi}(\mathbf{x},\mathbf{y})
\triangleq  \widetilde{\zeta} \exp \! \left( \!  - \frac{\pi\widetilde{\zeta}}{2} (\|\mathbf x\|^2 \! \! + \! \! \|\mathbf y\|^2 \! \! - \! 2 \mathbf x \bar{\mathbf y} ) \! \right) $ is the Ginibre kernel of $\Psi$, which represents the correlation force among different spatial points in $\Psi$. %, defined as $\$. 
\end{prop} 

%\vspace{-5mm} 

\section{Operational Model Selection Protocols}
\label{sec:MP}
 
The advantage of the hybrid relay lies in the fact that it can select the more suitable operational mode to achieve better performance  
under different network conditions. However, as the hybrid relay is a wireless-powered device, any mode selection protocols that have heavy computational complexity, e.g., based on online optimizations \cite{G.May2019Li,G.2019Li}, are not applicable.   
For implementation practicality, low complexity mode selection based on obtainable information and limited communication overhead needs to be devised.
For the mode selection of the hybrid relay, we consider the situations with and without instantaneous CSI. For the former and the latter situations, we propose two protocols, namely, {\em energy and SINR-aware Protocol (ESAP) } and {\em explore-then-commit protocol (ETCP)}, respectively, for the hybrid relay adapting to the network environment.
The operational procedures of the mode selection protocols are specified as follows.

\begin{itemize}
 
\item ESAP: 
%as predefined information., while the network-dependent parameters can be obtained by a detection period at the be. 
The idea of the ESAP is using ABR to assist data forwarding when the WPR is not feasible (i.e., either when $\mathrm{{R}}$ does not harvest sufficient energy or the forwarded data is not successfully decoded by $\mathrm{D}$) based on the current system conditions. 
The mode selection of ESAP is done at the end of each energy harvesting phase. In particular,
$\mathrm{S}$ first sends preamble signals to $\mathrm{R}$, and $\mathrm{R}$ detects its receive SINR $\nu_{\mathrm{R}}$ and the amount of harvested energy  $E_{\mathrm{R}}$ in its capacitor. If $\nu_{\mathrm{R}} > \tau_{\mathrm{R}}$ and 
$E_{\mathrm{R}} > E_{\mathrm{W}}$,  
$\mathrm{R}$ transmits preambles to $\mathrm{D}$ through WPR with transmit power $P^{\mathrm{W}}_{\mathrm{R}}=\big[\frac{2( %\omega \beta Q_{\mathrm{R}}
E_{\mathrm{R}} 
- E_{\mathrm{W}} )}{(1-\omega)T}, \frac{2E_{C}}{(1-\omega)T}\big]^{+}$. Otherwise, $\mathrm{R}$ chooses the ABR mode. Upon receiving the preamble signals from $\mathrm{R}$, $\mathrm{D}$ provides a feedback of its receive SINR to $\mathrm{R}$ through signaling. If the receive SINR at $\mathrm{D}$ is greater than $\tau_{\mathrm{W}}$, $\mathrm{R}$ chooses the WPR mode for relaying. Otherwise, $\mathrm{R}$ chooses the ABR mode for relaying. % and $\mathrm{D}$ uses the backscatter demodulator.
%In particular, at the beginning of each time slot, $\mathrm{S}$ first sends preamble signals to $\mathrm{R}$, and $\mathrm{R}$ detects its instantaneous energy harvesting rate $\rho_{\mathrm{R}}$ and the receive SINR $\nu_{\mathrm{R}}$. If $\rho_{\mathrm{R}} > \rho_{\mathrm{W}}$ and $\nu_{\mathrm{R}} > \tau_{\mathrm{R}}$, $\mathrm{R}$ transmits preambles to $\mathrm{D}$ through active transmission with transmit power $P^{\mathrm{W}}_{\mathrm{R}}=\big[\frac{2( %\omega \beta Q_{\mathrm{R}} \rho_{\mathrm{R}}  - \rho_{\mathrm{W}} )}{1-\omega}, \frac{2\rho_{C}}{1-\omega}\big]^{+}$. Otherwise, $\mathrm{R}$ chooses the ABR mode. Upon receiving the preamble signals from $\mathrm{R}$, $\mathrm{D}$ provides a feedback of its receive SINR to $\mathrm{R}$ through signaling. If the receive SINR at $\mathrm{D}$ is greater than $\tau_{\mathrm{W}}$, $\mathrm{R}$ chooses the WPR mode for relaying and $\mathrm{D}$ works with the quadrature demodulator. Otherwise, $\mathrm{R}$ chooses the ABR mode for relaying. % and $\mathrm{D}$ uses the backscatter demodulator.
For ESAP, we consider the ideal case that the mode selection at the end of each energy harvesting phase causes negligible time and energy consumption. %compared with the duration of a time slot.

\item ETCP: The idea of the ETCP is to select the averagely better-performed mode in terms of success probability based on the history information. In particular, ECTP begins with an exploration period that occupies the first initial 2$n$ ($n \in \mathbb{N}+$) to learn the network conditions before committing to a certain operation mode for steady-state transmission based on the learned knowledge.  
%For the first initial 2$n$ ($n \in \mathbb{N}+$) time slots, referred to as the 
Specifically, in the exploration period, $\mathrm{R}$ works in each operational mode for $n$ time slots in an arbitrary sequence. Afterward, $\mathrm{D}$ feeds back the numbers of successful transmissions in the WPR mode and ABR mode, denoted as $N_{\mathrm{WPR}}$ and $N_{\mathrm{ABR}}$, respectively, to $\mathrm{S}$. Since then, $\mathrm{R}$ always selects the WPR mode, if $N_{\mathrm{WPR}} > N_{\mathrm{ABR}}$, and the ABR mode, if $N_{\mathrm{ABR}} > N_{\mathrm{WPR}}$, and uniformly selects between the ABR and WPR modes at random, otherwise. 
 
\end{itemize}

%Meanwhile, if $\mathrm{R}$ is successfully powered, i.e., $P^{E}_{\mathrm{R}} \geq 2 \delta_{\mathcal{P}}$, it If $P^{E}_{\mathrm{R}} \!< \! 2 \delta_{\mathcal{A}}$, $\mathrm{R}$ chooses the passive relaying mode. If $\tau_{\mathrm{R}} \!\geq\! \tau_{\mathrm{A}}$ and $P^{E}_{\mathrm{R}} \!\geq\! 2 \delta_{\mathcal{A}}$, then $\mathrm{R}$ transmits preambles to $\mathrm{D}$ by active transmission with transmit power $P^{A}_{\mathrm{R}}=\big[\beta \rho \widetilde{P}_{T}\!-\!2\delta_{\mathcal{P}}, \frac{2E_{C}}{T}\big]^{+}$. And $\mathrm{D}$ provides a feedback of the receive SINR to $\mathrm{R}$ through signaling. If the receive SINR at $\mathrm{D}$ is greater than $\tau_{\mathrm{A}}$, $\mathrm{R}$ chooses active relaying mode and $\mathrm{D}$ works with the quadrature demodulator. Otherwise, $\mathrm{R}$ chooses passive relaying mode and $\mathrm{D}$ works with the backscatter demodulator.

%ESAP operates based on prior information including the physical parameters of $\mathrm{R}$, i.e., $\rho_{A}$ and $\rho_{C}$, and network environment-dependent parameters, i.e., $\rho_{\mathrm{R}}$,
% $\rho_{\mathrm{R}}$ and  $\nu_{\mathrm{D}}$. %of $\mathrm{D}$, and knowledge of physical parameters, i.e., and . The physical parameters can be known by the relay as predefined information, while the network-dependent parameters can be obtained through detection at the beginning of each time slot. 
 
Note that 
both of the proposed mode selection protocols are %practical as they are 
adopted only at the hybrid relay. 
ESAP operates based on the information including
the physical parameters of $\mathrm{R}$, i.e., $E_{\mathrm{W}}$ and $E_{C}$, 
and network environment-dependent parameters, i.e., $E_{\mathrm{R}}$, 
and 
$\nu_{\mathrm{D}}$.  
The physical parameters can be known by the relay as predefined knowledge, while the network-dependent parameters can be obtained through detection at the beginning of each time slot. By contrast,
ETCP chooses the mode based on the history information without knowing the current channel condition. %feedback information from the destination node without any prior information of the system. 
It can be seen that both mode selection protocols are practical as they are based on information attainable at the hybrid relay.

It is worth mentioning that ESAP incurs much higher overhead than ETCP. In particular, ESAP requires  the hybrid relay to know the expression of $P^{\mathrm{W}}_{\mathrm{R}}$ in (\ref{eqn:AR_R_dual}) and values of $\rho_{\mathrm{W}}$ and $\rho_{C}$ and as a priori. Additionally, the hybrid relay with ESAP needs a feedback from the destination node and calculation of $P^{\mathrm{W}}_{\mathrm{R}}$  according to (\ref{eqn:AR_R_dual}) every time slot. By contrast, the hybrid relay with ETCP only needs one feedback from the destination node and one comparison operation between the values of $N_{\mathrm{WPR}} $ and $ N_{\mathrm{ABR}}$ after the first 2$n$ initial time slots.

%\vspace{-3mm}

\section{Analysis of Success Probability}
\label{sec:SP}

In this section, we analyze the performance of the hybrid relaying system in presence of randomly located ambient emitters and interferers.   
%In particular, the performance is characterized in terms of the end-to-end success probability by applying the stochastic geometry analysis based on the $\alpha$-GPP framework.
For this,
\begin{itemize}

\item we first derive the  mode selection probabilities of the hybrid relay with the proposed protocols, 

\item we characterize the interference distribution under the $\alpha$-GPP modeling framework and derive the success probabilities of the hybrid relaying with both ESAP and ETCP, 

\item and we also simplify the success probabilities in the special cases when one of the operational modes of the hybrid relay is disabled and when the distribution of ambient transmitters follow PPPs.

\end{itemize}

The transmission of the dual-hop hybrid relaying system is considered to be successful if 1) the relay can harvest sufficient energy for its circuit operation and for decoding information transmitted by the source and 2) the destination can decode the information forwarded by the relay %from the source 
either through active transmission or ambient backscattering. Let $\mathrm{M} \in \{\mathrm{W}, \mathrm{A}\}$ denote the operational mode indicator of the hybrid relay $\mathrm{R}$.
Mathematically, the success probability of the hybrid relaying 
is expressed as
\begin{align} \label{definition_SP}  
 & \hspace{-5mm} \mathcal{S}_{\textup{HR}}   =  %\sum_{\mathrm{M} \in \{ \mathrm{W}, \mathrm{A}\} } 
 \mathbb{P} \Big[\nu_{\mathrm{R}} \! >\! \tau_{\mathrm{W}}, \nu^{\mathrm{W}}_{\mathrm{D}} \! > \! \tau_{\mathrm{W}}, E_{\mathrm{R}} > E_{\mathrm{W}}, \mathrm{M}=\mathrm{W}\Big]  +  \mathbb P \Big[ \nu_{\mathrm{R}} \!>\! \tau_{\mathrm{W}}, 
 \nu^{\mathrm{A}}_{\mathrm{D}} \!> \! \tau_{\mathrm{A}},  E_{\mathrm{R}} > E_{\mathrm{A}}, \mathrm{M} \!=\! \mathrm{A} \Big]  \nonumber \\
 & \hspace{-5mm} \overset{\text{(a)}}{=} \mathbb{P}\Big[\nu_{\mathrm{R}} \! > \! \tau_{\mathrm{W}}, \nu^{\mathrm{W}}_{\mathrm{D}} \! >\! \tau_{\mathrm{W}},  E_{\mathrm{R}} \! > \! E_{\mathrm{W}}| \mathrm{M}\!=\!\mathrm{W}\Big]  \mathbb{P} \Big[\mathrm{M}\!=\!\mathrm{W}\Big] \!+\!  \mathbb P \Big[ \nu_{\mathrm{R}} \! >\! \tau_{\mathrm{W}}, 
  \nu^{\mathrm{A}}_{\mathrm{D}} \! > \! \tau_{\mathrm{A}},  E_{\mathrm{R}} \! > \! E_{\mathrm{A}} | \mathrm{M} \!=\! \mathrm{A} \Big]   \mathbb{P} \Big[\mathrm{M} \!=\! \mathrm{A}\Big], \!\!\!\!\!
\end{align}
where (a) follows the Bayes' theorem  \cite[page 36]{bayesian}.
%where $\mathrm{M}_{\mathrm{R}} \in \{\mathrm{W}, \mathrm{A} \}$ is the mode indicator of $\mathrm{R}$.

%$\mathcal{E}_{\mathrm{HR}}=\frac{\mathcal{C}_{\mathrm{HR}}}{P_{\mathrm{S}}}$

%Let $\mathcal{C}_{\mathrm{A}}$ and $\mathcal{C}_{\mathrm{A}}$ denote the capacity of $\mathrm{R}$ being in active radio mode and ambient backscatter mode, respectively. 

% W \log_2(1+\nu_{\mathrm{A}}) 
 
%ambient backscatter-prioritized protocol (ABP) active radio-prioritized protocol (ARP)
 
%\subsection{Quasi-static Interference} 
%\vspace{-3mm}
\subsection{General-Case Result for ESAP}

%We first consider the quasi-static interference scenario where the interferers $\Phi$ have longer transmission time than a time-slot duration of the relaying protocol $T$. More specifically, during %the source-to-relay transmission and relay-to-destination transmission,  one time slot $T$, the relay and destination node coexist with the same set of interferers with static locations. 
We first investigate the hybrid relaying with ESAP. Note that with block Rayleigh fading channels, the %relay and destination node coexist with 
source-to-relay transmission and relay-to-destination transmission are affected by the same set of interferers with static locations. In other words, the relay and destination nodes experience spatially and temporally correlated interference. 
In this scenario, we characterize the success probability of hybrid relaying defined in (\ref{definition_SP}) in the following theorem.

% T1
\begin{theorem}  \label{thm:QSI_SP_HR}
The success probability of the hybrid relaying with ESAP is 
\begin{align} \label{eqn:QI_HR}
& \hspace{0mm} \mathcal{S}^{\textup{ESAP}}_{\textup{HR}} \! = \!  \exp \! \Big(\! - \!  \kappa(\tau_{\mathrm{W}}) \sigma^2  \! \Big) \! \Bigg( \! \chi_{\mathbb{A}} (\tau_{\mathrm{W}},\rho_C)   \bigg(\! 1 \! - \! %\mathcal{L}^{-1} \bigg\{ \! \frac{1}{s}\mathrm{Det} \big( \mathrm{Id}\!+\!\alpha_{Q} \mathbb{Q}_{\Psi}( s )   \big)^{\!-\frac{1}{\alpha_{Q}}} \! \bigg\} \! \big(\varrho_{\mathrm{W}}\!+\!\varrho_{C} \big)
F_{Q_{\mathrm{R}}} \big(\varrho_{\mathrm{W}}\!+\!\varrho_{C} \big)
%\bigg( \! \! \frac{\rho_{\mathrm{W}}\!+\!\rho_{C}\!}{\omega \beta} \! \bigg)  
- \! \int^{\infty}_{ \! \varrho_{\mathrm{W}}\!+\!\varrho_{C} } \! \delta(q)  %\exp \! \bigg( \!\!-\!\frac{ d^{\mu}_{\mathrm{R},\mathrm{D}} \sigma^2_{\Psi}   \tau_{\mathrm{A}}  }{ \eta \xi q } \! \! \bigg) 
 f_{Q_{\mathrm{R}}}(q) 
 \mathrm{d}q  \! \bigg)  + \! \int^{ \varrho_{\mathrm{W}}\!+\!\varrho_{C} }_{\!  \varrho_{\mathrm{W}} } \!\!\! \big (  1 \! - \! %\exp \! \bigg(\!\!-\!\frac{ d^{\mu}_{\mathrm{R},\mathrm{D}} \sigma^2_{\Psi}   \tau_{\mathrm{A}}  }{ \eta \xi q } \! \! \bigg) 
\delta(q)  \big) 
\nonumber \\
& \hspace{3mm} \! \!   
\times \chi_{\mathbb{A}} (\tau_{\mathrm{W}}, \omega \beta q \! - \! \rho_{\mathrm{W}} )  
f_{Q_{\mathrm{R}}}(q) \mathrm{d}q   +   
%\mho(\tau_{\mathrm{W}}) 
\mathrm{Det} \Big( \mathrm{Id} \! + \! \alpha \mathbb{G}_{\Phi}(\mathbf{x},\mathbf{y}) \psi_{\mathbf{x}}\big(\kappa(\tau_{\mathrm{W}})\big)   \psi_{\mathbf{y}}\big(\kappa(\tau_{\mathrm{W}}) \big) \! \Big)^{\!\!-\!\frac{1}{\alpha}}
\!\!\!\! \int^{ \infty }_{\! \varrho_{\mathrm{A}} }  \!\!\!  \delta(q)  
f_{Q_{\mathrm{R}}}(q) \mathrm{d} q \! \Bigg), \hspace{-4mm}
\end{align}  
where  $\kappa(v)\!\triangleq\!\frac{d^{\mu}_{\mathrm{S},\mathrm{R}} v }{P_{\mathrm{S}}}$, $\delta(q) \! \triangleq \! \exp \!  \Big( \!\! -  \! \frac{ d^{\mu}_{\mathrm{R},\mathrm{D}}   \widetilde{\sigma}^2 \tau_{\mathrm{A}}}{ \eta \xi q } \!  \Big)$, $\mathbb{G}_{\Phi}(\mathbf{x},\mathbf{y}) \! \triangleq \! \widetilde{\zeta} \exp \! \Big( \! \! - \! \frac{\pi\widetilde{\zeta}}{2} (\|\mathbf x\|^2 \! \! + \! \|\mathbf y\|^2 \! \! - \! 2 \mathbf x \bar{\mathbf y} ) \!  \Big)$, $ \psi_{\mathbf{z}}(s) \! \triangleq \! \sqrt{\!1\!-\!\Big( \! 1\!+\!  s P_{T} \|\mathbf{z} \|^{-\mu} \! \Big)^{\!\! -1}} $, $f_{Q_{\mathrm{R}}}(q)$ and $F_{Q_{\mathrm{R}}}(x)$ are given in (\ref{pdf}) and (\ref{CDF}), respectively, $\chi_{\mathbb{A}}$ %and $\mho$ are given, respectively, 
is given as:  %\vspace{-3mm}
%\begin{align} \chi_{\mathbb{A}}(v,p) \triangleq   \exp \Big( \! - \ell(v,p) \sigma^2   \Big)  \mathrm{Det} \Big(\mathrm{Id} \! + \! \alpha \mathbb{A}_{\Phi}\big( \ell(v,p) \big) \Big)^{\!- \frac{1}{ \alpha}}, \quad \mathbb{I} \in \{\mathbb{A}, \mathbb{B}  \} \end{align}
\begin{align}  \label{chi1}
\chi_{\mathbb{A}}(v,p) =
\exp \! \Big( \! - \ell(v,p) \sigma^2  \Big) 
\mathrm{Det} \Big(\mathrm{Id} \! + \! \alpha  \mathbb{G}_{\Phi}(\mathbf{x},\mathbf{y})\varphi_\mathbf{x}\big(\ell(v,p)\big)\varphi_\mathbf{y}\big(\ell(v,p)\big) \! \Big)^{\!- \frac{1}{ \alpha}}, %\quad \mathbb{I} \in \{\mathbb{A}, \mathbb{B}  \}
\end{align}
\begin{comment}
%\vspace{-10mm} 
\begin{flalign} \label{mho}
 \text{and}, \hspace{10mm} \mho(v) = %\exp \! \bigg(\! \! - \!\frac{d^{\mu}_{\mathrm{S},\mathrm{R}} \sigma^2 v }{  P_{\mathrm{S}} } \! \bigg)
 \exp \Big( \! - \kappa(v) \sigma^2  \Big)
 \mathrm{Det} \Big( \mathrm{Id} \! + \! \alpha_{I} \mathbb{G}_{\Phi}(\mathbf{x},\mathbf{y}) \psi_{\mathbf{x}}\big(\kappa(v)\big)   \psi_{\mathbf{y}}\big(\kappa(v) \big) \! \Big)^{\!-\!\frac{1}{\alpha_{I}}}, \!
\end{flalign} 
\end{comment}
therein $\ell(v,p) \! \triangleq \! \frac{  d^{\mu}_{\mathrm{R},\mathrm{D}} v (1-\omega) }{ 2 p } $,   and  $\varphi_{\mathbf{z}}(s)\! \triangleq \!\! \sqrt{\!1 \! - \!\Big( \! 1 \! + \! 
 \kappa(\tau_{\mathrm{A}}) P_{T} \|\mathbf{z}  \|^{-\mu}   \!  \Big)^{\!\!\!-1}   \! \Big(\! 1 \! + \! s P_{T}\|\mathbf{z} \! - \! \mathbf{x}_{\mathrm{D}} \|^{-\mu} \!  \Big)^{\!\!\!-1}  } $.   

\end{theorem}   
 
For readability, the proof of $\textbf{Theorem}$~\ref{thm:QSI_SP_HR} is presented in $\textbf{Appendix A}$.  
  
%\begin{align} & \mathcal{R} = \frac{1}{N} \bigg( \mathbb{E} \big[ W  \log_{2}( 1 + \nu ) \mathbbm{1}_{ \{ \nu \geq \tau_{\mathrm{A}}  \}} \big] \mathbb{P} \big[ \mathrm{M}_{\mathrm{R}_{1}}=\mathrm{A},\mathrm{M}_{\mathrm{R}_{2}}=\mathrm{A},\cdots,\mathrm{M}_{\mathrm{R}_{N}}=\mathrm{A} \big] \nonumber \\ & \hspace{30mm} + T_{\mathrm{A}} \Big(1-\mathbb{P} \big[ \mathrm{M}_{\mathrm{R}_{1}}=\mathrm{A},\mathrm{M}_{\mathrm{R}_{2}}=\mathrm{A},\cdots,\mathrm{M}_{\mathrm{R}_{N}}=\mathrm{A} \big] \Big)  \bigg) \end{align} where $\nu_{\mathrm{A}} =\min(\nu^{\mathrm{A}}_{R_{1}},\nu^{\mathrm{A}}_{R_{2}},\cdots, \nu^{\mathrm{A}}_{R_{N}})$
%\begin{comment}
   
%We also note that similar to the stochastic geometry analysis based on PPP in the existing literature, e.g., [36], it is difficult to see the relationship between the performance metric and system parameters directly from the general-case results in Theorems 1 and 2 derived based on the α-GPP framework. However, these general-case results can be simplified in some special cases. We then investigate a special setting which considerably simplifies the above results.   

The analytical expression in (\ref{eqn:QI_HR}) appears in terms of the Fredholm determinant~\cite{L.2015Decreusefond}, which allows an efficient numerical evaluation of the relevant quantities~\cite{I.Dec.2014Flint,X.March2015Lu,X.2016JLu}. 
It is observed that the analytical expression of $\mathcal{S}^{\textup{ESAP}}_{\textup{HR}}$ in (\ref{eqn:QI_HR}) has multiple terms. This is due to the fact that the proposed hybrid relaying features with a two-mode operation. The analytical expression involves the joint probabilities that an operational mode is selected and the relay transmission in the selected mode is successful.  %However, these general-case results can be simplified in some special cases.
We note that the analytical expression in (\ref{eqn:QI_HR}) has a comparable computational complexity to the analytical results in~\cite{X.March2018Lu,Kong2016,B.June2017Kong}.
The terms that have the highest computational complexity (e.g., the last term in (\ref{eqn:QI_HR})) involve one integral of the inverse Laplace transform of the Fredholm determinant, which can be evaluated relatively easily with numerical integration tools.

%\vspace{-3mm}
\subsection{Special-Case Results}
Next, we investigate some special settings which considerably simplify the general-case result in (\ref{eqn:QI_HR}). 
 
\subsubsection{Pure Ambient Backscatter Relaying}
In the special case when $\mathrm{R}$ forwards information from $\mathrm{S}$ to $\mathrm{D}$ through ambient backscattering only, referred to as {\em pure ABR}, we have the corresponding success probability as follows.

\begin{corollary}   \label{corollary1}
% C1
The success probability of the %cooperative relaying with an 
pure ABR
%ambient backscatter relaying is 
\begin{align}
\hspace{-2mm} \mathcal{S}_{\textup{ABR}} \! \hspace{-0.5mm} = \!  %\exp \! \bigg(\! \! - \! \frac{d^{\mu}_{\mathrm{S},\mathrm{R}} \sigma^2   \tau_{\mathrm{A}}  }{ P_{\mathrm{S}} }  \!   \bigg) 
\exp \! \Big( \! \! - \kappa(\tau_{\mathrm{W}}) \sigma^{2} \Big) \mathrm{Det} \Big( \mathrm{Id} \! + \!   %\mathbb{B}_{\Phi} \big( \kappa(\tau_{\mathrm{W}}) \big) 
\alpha \mathbb{G}_{\Phi}(\mathbf{x},\mathbf{y}) \psi_{\mathbf{x}}\big(\kappa(\tau_{\mathrm{W}})\big)   \psi_{\mathbf{y}}\big(\kappa(\tau_{\mathrm{W}}) \big)  \! \Big)^{\!\!- \frac{1}{\alpha}} \! \!\! \int^{\infty}_{\!  \varrho_{\mathrm{A}} }  \!   %\exp \!  \bigg( \! \! -  \! \frac{ d^{\mu}_{\mathrm{R},\mathrm{D}}   \sigma^2_{\Psi} \tau_{\mathrm{A}}}{ \eta \xi q } \!  \bigg) 
\delta(q) f_{Q_{\mathrm{R}}}(q) \mathrm{d} q,   \label{eqn:theorem_ABR}  \hspace{-2mm}
\end{align}
where %$\kappa(v)=\frac{d^{\mu}_{\mathrm{S},\mathrm{R}} v }{P_{\mathrm{S}}}$, 
%$\mathbb{B}_{\Phi}$ and
$f_{Q_{\mathrm{R}}}(q)$ is given in %(\ref{eq:kernal_I_R1}), %and 
(\ref{pdf}),
respectively. 
\end{corollary}   
   
\begin{proof} % P1 
The performance of %dual-hop cooperative relaying with ambient backscattering relays 
%ambient backscatter relaying 
the pure ABR
can be derived by setting $\mathrm{R}$ exclusively in ABR mode for relaying as long as the harvested energy is enough to support the function.  %$\mathrm{R}$ relaying the information from $\mathrm{S}$ to $\mathrm{D}$ through ambient backscattering.  
Mathematically, by plugging $\mathbb{P}\big[\mathrm{M}=\mathrm{W}\big]=0$ and $\mathbb{P}\big[\mathrm{M}=\mathrm{A}\big]= \mathbb{P}\big[E_{\mathrm{R}} > E_{\mathrm{A}}]$ into the definition in (\ref{definition_SP}), the corresponding success probability can be expressed as: 
 \begin{align}
 \mathcal{S}_{\textup{ABR}}  =     \mathbb{P} \Big [\nu_{\mathrm{R}} \! > \! \tau_{\mathrm{W}} , \nu^{\mathrm{A}}_{\mathrm{D}} \! > \! \tau_{\mathrm{A}}, E_{\mathrm{R}}  \!>\! E_{\mathrm{A}}  \Big ],
%\begin{align} 
\label{proof:ABR}
 \end{align} 
which is equivalent to the second term of $\mathcal{S}^{\textup{ESAP}}_{\mathrm{A}}$ in (\ref{eqn:C_ARP}) with $E_{\mathrm{W}}$ replaced by $\infty$. Therefore, the analytical expression of $\mathcal{S}_{\textup{ABR}}$ in (\ref{eqn:theorem_ABR}) yields from the derivation of $\mathcal{S}^{\textup{ESAP}}_{\mathrm{A}}$ with the mentioned replacement.
\end{proof}           

%It is noted that for the pure ABR, only the transmission over the source-to-relay link is on the transmit frequency of the interferers $\Phi$. Thus, both quasi-static interference or fast-varying interference only affect the source-to-relay link and have the same impact on the performance of the pure ABR. %ambient backscatter relaying. % on the transmit frequency of $\Phi$  

%$\nu_{\mathrm{R}}$ and $\nu^{\mathrm{A}}_{\mathrm{D}}$ are independent. 

We note that the success probability of the hybrid relaying in (\ref{eqn:QI_HR}) can be expanded as follows: 
\begin{align}
\mathcal{S}^{\textup{ESAP}}_{\textup{HR}}  &  =  \mathbb{P} \Big [  \nu_{\mathrm{R}} \! > \! \tau_{\mathrm{W}}, \nu^{\mathrm{W}}_{\mathrm{D}} \! > \! \tau_{\mathrm{W}}, \nu^{\mathrm{A}}_{\mathrm{D}} \! > \! \tau_{\mathrm{A}}, E_{\mathrm{R}}  \! > \! E_{\mathrm{W}}   \Big ]  %\nonumber \\ & \hspace{75mm} 
+  \mathbb{P} \Big [  \nu_{\mathrm{R}} \! > \! \tau_{\mathrm{W}}, \nu^{\mathrm{W}}_{\mathrm{D}} \! > \! \tau_{\mathrm{W}}, \nu^{\mathrm{A}}_{\mathrm{D}} \! \leq \! \tau_{\mathrm{A}}, E_{\mathrm{R}}  \!>\! E_{\mathrm{W}}   \Big ] \nonumber \\
& \hspace{10mm} + \mathbb{P} \Big [ \nu_{\mathrm{R}} \! > \! \tau_{\mathrm{W}}, \nu^{\mathrm{A}}_{\mathrm{D}} \! > \! \tau_{\mathrm{A}},   \nu^{\mathrm{W}}_{\mathrm{D}} \! \leq \! \tau_{\mathrm{W}},   E_{\mathrm{R}}       \!>\! E_{\mathrm{W} } \Big ] + \mathbb P \Big[ \nu_{\mathrm{R}} \! > \! \tau_{\mathrm{W}}, \nu^{\mathrm{A}}_{\mathrm{D}} \!>\! \tau_{\mathrm{A}}, E_{\mathrm{W}} \! \geq \! E_{\mathrm{R}} \! > \! E_{\mathrm{A}} \Big]  \nonumber \\
& = \mathbb{P} \Big [  \nu_{\mathrm{R}} \! > \! \tau_{\mathrm{W}},  \nu^{\mathrm{A}}_{\mathrm{D}} \! > \! \tau_{\mathrm{A}}, E_{\mathrm{R}}  \! > \! E_{\mathrm{A}}   \Big ] +  \mathbb{P} \Big [  \nu_{\mathrm{R}} \! > \! \tau_{\mathrm{W}}, \nu^{\mathrm{W}}_{\mathrm{D}} \! > \! \tau_{\mathrm{W}}, \nu^{\mathrm{A}}_{\mathrm{D}} \! \leq \! \tau_{\mathrm{A}}, E_{\mathrm{R}}  \!>\! E_{\mathrm{W}}   \Big ],   \label{eqn:expended_SP_HR}
\end{align}
where the first equality follows by expanding the first term of (\ref{eqn:C_ARP}) into two cases when $\nu^{\mathrm{A}}_{\mathrm{D}} > \tau_{\mathrm{A}}$ and $\nu^{\mathrm{A}}_{\mathrm{D}} \leq \tau_{\mathrm{A}}$ and the last equality follows by combining the first, the third and the fourth terms before the equality.
One finds that the probability representation of $\mathcal{S}_{\textup{ABR}}$ in (\ref{proof:ABR}) is exactly the first term of (\ref{eqn:expended_SP_HR}).
Given that the second term of (\ref{eqn:expended_SP_HR}) is always positive, we have the following observation.

{\bf Remark 1}:  
The success probability of the hybrid relaying with ESAP is strictly higher than that of the pure ABR. %ambient backscatter relaying.

Let $\mathcal{G}^{\textup{ESAP}}_{\textup{ABR}} %= \mathcal{S}^{\textup{QI}}_{\textup{HR}}- \mathcal{S}^{\textup{QI}}_{\textup{ABR}}
$ denote the performance improvement of the hybrid relaying with ESAP over the pure ABR %wireless-powered relaying
in terms of the success probability, i.e., $ \mathcal{G}^{\textup{ESAP}}_{\textup{ABR}} = \mathcal{S}^{\textup{ESAP}}_{\textup{HR}}- \mathcal{S}_{\textup{ABR}}$.  
%$\mathcal{S}_{\textup{HR}}$ and $\mathcal{S}^{\textup{QI}}_{\textup{ABR}}$, indicating the improvement of  
In particular, we have 
\begin{align} \label{eqn:G_HR_ABR}
& \mathcal{G}^{\textup{ESAP}}_{\textup{ABR}} =  \exp \! \Big(\! - \!  \kappa(\tau_{\mathrm{W}}) \sigma^2   \! \Big) \! \Bigg( \! \underbrace{\chi_{\mathbb{A}} (\tau_{\mathrm{W}},\rho_C)   \bigg(\! 1 \! - \! %\mathcal{L}^{-1} \bigg\{ \! \frac{1}{s}\mathrm{Det} \big( \mathrm{Id}\!+\!\alpha_{Q} \mathbb{Q}_{\Psi}( s )   \big)^{\!-\frac{1}{\alpha_{Q}}} \! \bigg\} \! \big(\varrho_{\mathrm{W}}\!+\!\varrho_{C} \big)
F_{Q_{\mathrm{R}}} \big(\varrho_{\mathrm{W}}\!+\!\varrho_{C} \big)
%\bigg( \! \! \frac{\rho_{\mathrm{W}}\!+\!\rho_{C}\!}{\omega \beta} \! \bigg)  
- \! \int^{\infty}_{ \! \varrho_{\mathrm{W}}\!+\!\varrho_{C} } \! \delta(q)  %\exp \! \bigg( \!\!-\!\frac{ d^{\mu}_{\mathrm{R},\mathrm{D}} \sigma^2_{\Psi}   \tau_{\mathrm{A}}  }{ \eta \xi q } \! \! \bigg) 
 f_{Q_{\mathrm{R}}}(q) 
 \mathrm{d}q  \! \bigg)}_{\mathcal{G}_{1}}  
\nonumber \\
& \hspace{75mm} \!  \!   + \! \underbrace{\int^{ \varrho_{\mathrm{W}}\!+\!\varrho_{C} }_{\!  \varrho_{\mathrm{W}} } \!\!\! \big (  1 \! - \! %\exp \! \bigg(\!\!-\!\frac{ d^{\mu}_{\mathrm{R},\mathrm{D}} \sigma^2_{\Psi}   \tau_{\mathrm{A}}  }{ \eta \xi q } \! \! \bigg) 
\delta(q)  \big)   
\chi_{\mathbb{W}} (\tau_{\mathrm{W}}, \omega \beta q \! - \! \rho_{\mathrm{W}} )  
f_{Q_{\mathrm{R}}}(q) \mathrm{d}q}_{\mathcal{G}_2} \! \Bigg).
\end{align}
Based on the expansion of the Fredholm determinant in~\cite[eqn. 14]{Kong2016}, $\chi_{\mathbb{A}}(\tau_{\mathrm{W}},\rho_{C})$  can be expressed as: %$\frac{  d^{\mu}_{\mathrm{R},\mathrm{D}} \tau_{\mathrm{W}} (1-\omega) }{ 2 \rho_{C} } $
\begin{eqnarray}  \label{evaluation_chi_A}
\chi_{\mathbb{A}}(\tau_{\mathrm{W}},\rho_{C}) \!= \exp \! \bigg( \! \! - \! \frac{  d^{\mu}_{\mathrm{R},\mathrm{D}} \tau_{\mathrm{W}} (1\!-\!\omega) \sigma^2  }{ 2 \rho_{C} }   \! \bigg) \! \prod_{n \geq 0} \! \left(\! 1 \!+\! \frac{2 \alpha (\pi \zeta)^{n+1}}{n!} \!\! \int^{R}_{0} \! \! 
%\bigg(
\frac{ \exp(-\pi \zeta r^2) r^{2n+1}}{ 
\! 1\!+\! 
%\frac{  
2 \rho_{C} r^{\mu}
%}{ 
 (d^{\mu}_{\mathrm{R},\mathrm{D}} \tau_{\mathrm{W}} (1\!-\!\omega)  P_{T})^{-1} %}
  } 
%\bigg)^{\!\! -1} 
\mathrm{d} r \! \right)^{\!\!\!-\frac{1}{\alpha}}\!\!\!. %\vspace{-4mm} 
\end{eqnarray}
As the repulsion factor $\alpha \in [-1,0)$ and all the other parameters take positive values, it is readily checked that $\chi_{\mathbb{A}}$ is an increasing function of $\varrho_{C}$. Given that the physical capacity of the capacitor, i.e., %$E_{\mathrm{W}}+E_{C}$ 
$\varrho_{\mathrm{W}}+\varrho_{C}$, is fixed, both $\mathcal{G}_{1}$ and $\mathcal{G}_{2}$ in (\ref{eqn:G_HR_ABR}), and thus $\mathcal{G}^{\textup{ESAP}}_{\textup{ABR}}$, %decrease with the increase of 
increase with $\varrho_{C}$.   
As a result, $\chi_{\mathbb{A}}$ is a decreasing function of $\varrho_{\mathrm{W}}$.
%Additionally, it is noted
It is also noted that  $\varrho_{\mathrm{A}}$ does not appear in the expression of $\mathcal{G}^{\textup{ESAP}}_{\textup{ABR}}$. %Therefore, we can draw the following remark.

%$\frac{\partial{\mathcal{G}^{\textup{HR}}_{\textup{ABR}}}}{ \varrho_{\mathrm{W}}} = $

{\bf Remark 2}: According to (\ref{eqn:G_HR_ABR}), the improvement of the hybrid relaying with ESAP over the pure ABR %wireless-powered relaying
$\mathcal{G}^{\textup{HR}}_{\textup{ABR}}$ can be increased with the reduced circuit power consumption $E_{\mathrm{W}}$, while $\mathcal{G}^{\textup{HR}}_{\textup{ABR}}$ is not affected by any change of the circuit power consumption $E_{\mathrm{A}}$.

Furthermore, considering the special case of $\textbf{Corollary}$~\ref{corollary1} where the distributions of $\Psi$ and $\Phi$ exhibit no repulsion, i.e., the Poisson field of the ambient emitters and interferers with $\widetilde{\alpha} \to 0$ and $\alpha \to 0$, we can simplify $\mathcal{S}_{\textup{ABR}}$ in a closed form. %expression in the following corollary. 

%In a special case where the distribution of the ambient signal sources exhibits no repulsion, (i.e., they are deployed independently), we simplify the active probability in the following corollary.   
   
\begin{corollary}  \label{corollary2}
% C2
When the path-loss exponent $\mu$ equals $4$,
the success probability of the pure ABR % ambient backscatter relaying %cooperative relaying with an ambient backscatter relay 
in the Poisson field of ambient emitters and interferers is 
\begin{align} \label{eqn:SP_ABR}
& \mathcal{S}_{\textup{ABR}} \! = \! %\frac{\pi^{\frac{3}{2}}}{2}  \widetilde{\zeta} \sqrt{\frac{\widetilde{P}_{T} \omega \beta }{\rho_{\mathrm{A}}}} 
\frac{\pi^2}{4} \widetilde{\zeta} \sqrt{\frac{\widetilde{P}_{T}}{N}} \hspace{ 1mm} 
\exp \! \bigg(\! \! - \!\frac{d^{4}_{\mathrm{S},\mathrm{R}} \sigma^2   \tau_{\mathrm{W}}  }{ P_{\mathrm{S}} }  \! %   \bigg) \exp \bigg (\!
-  \frac{ \pi^2 \zeta  d^{2}_{\mathrm{S},\mathrm{R}} }{ 2 } \sqrt{\frac{ \tau_{\mathrm{W}} P_{T}}{P_{\mathrm{S}}} }  \bigg)  \mathrm{Erf} \Bigg( \! \sqrt{  \frac{N}{\varrho_{\mathrm{A}}}  } \Bigg) %\nonumber \\
%& \hspace{75mm} \times \hspace{0.1mm}_{1}F_{1} \bigg[\frac{1}{2},\frac{3}{2}, - \frac{\omega \beta}{\rho_{\mathrm{A}}}\bigg( \frac{d^2_{\mathrm{R},\mathrm{D}} \sigma^2_{\Psi} \tau_{\mathrm{A}}}{\eta \xi  } + \frac{\pi^4 \zeta^2_{Q} \widetilde{P}_{T} }{16} \bigg) \bigg]
,
\end{align} 
where $N \triangleq \frac{d^4_{\mathrm{R},\mathrm{D}} \widetilde{\sigma}^2 \tau_{\mathrm{A}}}{\eta \xi   } + \frac{\pi^4 \zeta^2_{Q} \widetilde{P}_{T} }{16}$ and $\mathrm{Erf}(t) = \frac{1}{\sqrt{t}} \int^{t}_{-t} \exp(-x^2) \mathrm{d}x $ is the error function~\cite{M.1972Abramowitz}. %$_{1}F_{1}(a;b;z)$ denotes the Kummer confluent hypergeometric function~\cite{M.1972Abramowitz} which can be expressed in the series expansion $_{1}F_{1}(a;b;z)=\sum^{\infty}_{n=0}\frac{(a)_{n} z^{n}}{(b)_{n} n!}$, and $(x)_{n}=x(x-1) \cdots (x-(n-1))$ presents the falling factorial.

\end{corollary}   

The proof of $\textbf{Corollary}$~\ref{corollary2} is presented in $\mathbf{Appendix}$ $\mathbf{B}$. 

The closed-form expression in (\ref{eqn:SP_ABR}) directly reveals the effects of the parameters on the success probability. As the circuit power consumption of the pure ABR is ultra-low, we have $\varrho_{\mathrm{A}} \to 0$, and thus $\mathrm{Erf} \Big( \! \sqrt{  \frac{N}{\varrho_{\mathrm{A}}}  } \Big) \to 1$ and $\mathcal{S}_{\textup{ABR}} \! \approx \! %\frac{\pi^{\frac{3}{2}}}{2}  \widetilde{\zeta} \sqrt{\frac{\widetilde{P}_{T} \omega \beta }{\rho_{\mathrm{A}}}} 
\frac{\pi^2}{4} \widetilde{\zeta} \sqrt{\frac{\widetilde{P}_{T}}{N}} \hspace{ 1mm} 
\exp \! \bigg(\! \! - \!\frac{d^{4}_{\mathrm{S},\mathrm{R}} \sigma^2  \tau_{\mathrm{W}}  }{ P_{\mathrm{S}} }  \! %   \bigg) \exp \bigg (\!
-  \frac{ \pi^2 \zeta  d^{2}_{\mathrm{S},\mathrm{R}} }{ 2 } \sqrt{\frac{ \tau_{\mathrm{W}} P_{T}}{P_{\mathrm{S}}} }  \bigg)$. One easily observes that $\mathcal{S}_{\textup{ABR}}$ is an increasing function of $\widetilde{\zeta}$, $P_{\mathrm{S}}$, and $\widetilde{P}_{T}$, and a decreasing function of $\tau_{\mathrm{W}}$, $\tau_{\mathrm{A}}$, $d_{\mathrm{S},\mathrm{R}}$, $d_{\mathrm{R},\mathrm{D}}$, $\zeta$, and $P_{T}$. Among these parameters, $P_{\mathrm{S}}$ is the only one controllable by the hybrid relaying system. %We further observe that in order to maintain a certain target value for $\mathcal{S}_{\textup{ABR}}$, $P_{\mathrm{S}}$ needs to scale linearly with $P_{T}$ and at a rate of 
%In networks with increased transmit power $P_{T}$ or the spatial density $\zeta$ of the interferers,
In order to maintain a certain target value for $\mathcal{S}_{\textup{ABR}}$,
$P_{\mathrm{S}}$ needs to scale linearly with transmit power $P_{T}$ and at a rate of $P_{\mathrm{S}} \propto \zeta^2$ with the spatial density of the interferers.

\subsubsection{Pure Wireless-Powered Relaying}
%On the other hand, 
Next, we consider another special case when $\mathrm{R}$ forwards information over the relay-to-destination link only with wireless-powered transmission referred to as the pure WPR. %{\em wireless-powered relaying}
The corresponding success probability is given in the following corollary.  

\begin{corollary}   \label{corollary3}
%WPR C3 
The success probability of the pure WPR %wireless-powered relaying
is 
\begin{align} \label{eqn:theorem_WPR} 
& \mathcal{S}_{\textup{WPR}} \! = \!  \exp \! \Big(\! \! - \! \kappa(\tau_{\mathrm{W}}) \sigma^2  %\frac{d^{\mu}_{\mathrm{S},\mathrm{R}} \sigma^2  \tau_{\mathrm{W}}  }{ P_{\mathrm{S}} }   
\!  \Big)  \!
\bigg( \! \int^{ \varrho_{\!\mathrm{W}} + \varrho_{C}  }_{  \varrho_{\!\mathrm{W}}  } \! \! \! \!   \chi_{\mathbb{A}}(\tau_{\mathrm{W}},\omega \beta q  \! - \! \rho_{\mathrm{W}}) f_{Q_{\mathrm{R}}}(q) \mathrm{d} q %\nonumber \\
%&  \hspace{45mm} 
+ \chi_{\mathbb{A}}(\tau_{\mathrm{W}},\rho_{C}) \Big(\! 1 \! - \! 
%\mathcal{L}^{-1} \bigg\{ \! \frac{1}{s}\mathrm{Det}\big(\mathrm{Id}+\alpha_{Q}\mathbb{Q}_{\Psi}(s)\big)^{\!-\frac{1}{\alpha_{Q}}} \! \bigg \} \! 
F_{Q_{\mathrm{R}}}\big(   \varrho_{\mathrm{W}}\!+\!\varrho_C   \big) \! \Big) \! \! \bigg) ,  
  \hspace{-2mm}
\end{align}  
where %$\kappa \! = \! \frac{d^{\mu}_{\mathrm{S},\mathrm{R}} \tau_{\mathrm{W}}}{P_{S}}$, 
$\chi_{\mathbb{A}}$, $F_{Q_{\mathrm{R}}}(x)$, and $f_{Q_{\mathrm{R}}}(p)$ are given in (\ref{chi1}), (\ref{CDF}), and (\ref{pdf}), respectively.
\end{corollary}

\begin{proof}
The performance of the pure WPR %wireless-powered relaying
can be obtained by letting $\mathcal{\mathrm{R}}$ forward the information from $\mathrm{S}$ to  $\mathrm{D}$ with active transmission only, once the harvested energy is sufficient for the function. Mathematically, we have   $\mathbb{P}\big[\mathrm{M}=\mathrm{W}\big] =  \mathbb{P} \Big[\nu_{\mathrm{R}} \!>\! \tau_{\mathrm{W}}, \nu^{\mathrm{W}}_{\mathrm{D}} \! > \! \tau_{\mathrm{W}} ,  E_{\mathrm{R}} \! > \! E_{\mathrm{W}}\Big]$ and $\mathbb{P}\big[\mathrm{M}=\mathrm{A}\big] = 0$. 
By assigning the above conditions to the definition in (\ref{definition_SP}), the corresponding success probability can be expressed as:  
\begin{align}
\mathcal{S}_{\textup{WPR}} \! = \!       \mathbb{P} \Big [ \nu_{\mathrm{R}} \! > \! \tau_{\mathrm{W}}, \nu^{\mathrm{W}}_{\mathrm{D}} \! > \! \tau_{\mathrm{W}}, E_{\mathrm{R}}  \! > \! E_{\mathrm{W}}  \Big ], \label{eqn:SP_WPR_QI}
\end{align}
which is exactly the probability representation of $\mathcal{S}^{\textup{ESAP}}_{\mathrm{W}}$ in (\ref{eqn:C_ARP}). Therefore, the analytical expression of  $\mathcal{S}_{\textup{WPR}}$ in (\ref{eqn:theorem_WPR}) can be directly obtained from (\ref{eqn:SP_HR_A}). 
\end{proof} 

{\bf Remark 3}: As the probability representation of $\mathcal{S}_{\textup{WPR}}$ in (\ref{eqn:SP_WPR_QI}) is exactly $\mathcal{S}^{\textup{ESAP}}_{\mathrm{W}}$ in (\ref{eqn:C_ARP}), we have $\mathcal{S}^{\textup{ESAP}}_{\textup{HR}}=\mathcal{S}^{\textup{ESAP}}_{\mathrm{W}}+\mathcal{S}^{\textup{ESAP}}_{\mathrm{A}} > \mathcal{S}_{\textup{WPR}}$, given that $\mathcal{S}^{\textup{ESAP}}_{\mathrm{A}}$ is positive. Therefore, the success probability of the hybrid relaying with ESAP is strictly higher than that of the pure WPR. %wireless-powered relaying.  

Let $\mathcal{G}^{\textup{HR}}_{\textup{WPR}} %= \mathcal{S}^{\textup{QI}}_{\textup{HR}}- \mathcal{S}^{\textup{QI}}_{\textup{ABR}}
$ denote the performance improvement of the hybrid relaying over the pure WPR %wireless-powered relaying
in terms of the success probability,  i.e., $ \mathcal{G}^{\textup{ESAP}}_{\textup{WPR}} = \mathcal{S}^{\textup{ESAP}}_{\textup{HR}}- \mathcal{S}_{\textup{WPR}}$. 
%$\mathcal{S}^{\textup{QI}}_{\textup{HR}}$ and $\mathcal{S}^{\textup{QI}}_{\textup{ABR}}$, indicating the improvement of  
In particular, we have  
\begin{align} \label{eqn:G_HR_WPR}
& \mathcal{G}^{\textup{ESAP}}_{\textup{WPR}} =  \exp \! \Big(\! - \!  \kappa(\tau_{\mathrm{W}}) \sigma^2  \! \Big)   \! \Bigg( \! %\mho(\tau_{\mathrm{W}}) 
\underbrace{ \mathrm{Det} \Big( \mathrm{Id} \! + \! \alpha \mathbb{G}_{\Phi}(\mathbf{x},\mathbf{y}) \psi_{\mathbf{x}}\big(\kappa(\tau_{\mathrm{W}})\big)   \psi_{\mathbf{y}} \big(\kappa(\tau_{\mathrm{W}}) \big) \! \Big)^{\!\!-\!\frac{1}{\alpha}} }_{\mathcal{G}_{3}}  
\!  \int^{ \infty }_{\! \varrho_{\mathrm{A}} }  \!\!  \delta(q)  
f_{Q_{\mathrm{R}}}(q) \mathrm{d} q \nonumber \\
& \hspace{25mm} \! 
 - \underbrace{ \chi_{\mathbb{A}} (\tau_{\mathrm{W}},\rho_C)    
 \bigg( \! \int^{\infty}_{ \! \varrho_{\mathrm{W}}\!+\!\varrho_{C} } \! \delta(q)  %\exp \! \bigg( \!\!-\!\frac{ d^{\mu}_{\mathrm{R},\mathrm{D}} \widetilde{\sigma}^2  \tau_{\mathrm{A}}  }{ \eta \xi q } \! \! \bigg) 
 f_{Q_{\mathrm{R}}}(q) 
 \mathrm{d}q  
%\nonumber \\
%& \hspace{70mm} \!  
  + \! \int^{ \varrho_{\mathrm{W}}\!+\!\varrho_{C} }_{\!  \varrho_{\mathrm{W}} } \!  %\exp \! \bigg(\!\!-\!\frac{ d^{\mu}_{\mathrm{R},\mathrm{D}} \sigma^2_{\Psi}   \tau_{\mathrm{A}}  }{ \eta \xi q } \! \! \bigg) 
\delta(q)    \chi_{\mathbb{A}} (\tau_{\mathrm{W}}, \omega \beta q \! - \! \rho_{\mathrm{W}} )  
f_{Q_{\mathrm{R}}}(q) \mathrm{d}q  \bigg)  }_{\mathcal{G}_{4}}  \! \Bigg) .
\end{align}

Recall that $\chi_{\mathbb{A}}$ is a decreasing function of $\varrho_{\mathrm{W}}$ from \textbf{Remark 2}. %based on the expansion of the Fredholm determinant in~\cite[eqn. 14]{Kong2016}. 
Given that the overall capacity of the capacitor is fixed, i.e., $\varrho_{\mathrm{W}}+\varrho_{C}$, it is readily checked that $\mathcal{G}_{4}$ decreases with $\varrho_{\mathrm{W}}$. Hence, $\mathcal{G}^{\textup{ESAP}}_{\textup{WPR}}$ is an increasing function of $\varrho_{\mathrm{W}}$.
In addition, we have $\frac{\partial{\mathcal{G}^{\textup{ESAP}}_{\textup{WPR}} } }{ \partial{\varrho_{\mathrm{A}}} } =  - \exp \! \big(\! - \!  \kappa(\tau_{\mathrm{W}})  \sigma^2   \big) \mathcal{G}_{3}    \delta(\varrho_{\mathrm{A}}) f_{Q_{\mathrm{A}}}(\varrho_{\mathrm{A}}) < 0 $. Thus, $\mathcal{G}^{\textup{ESAP}}_{\textup{WPR}}$ is a decreasing function of $\varrho_{\mathrm{A}}$.

{\bf Remark 4:} 
It is observed from
(\ref{eqn:G_HR_WPR}) that the improvement of the hybrid relaying with ESAP over the pure WPR %wireless-powered relaying 
becomes more remarkable with increased circuit power consumption $E_{\mathrm{W}}$ and reduced circuit power consumption $E_{\mathrm{A}}$.

%\vspace{-3mm}
\subsection{General-Case Results for ETCP}

Next, we continue to investigate the performance of hybrid relaying with ETCP at the steady states when $\mathrm{R}$ has committed to a certain mode based on its selection criteria. The mode selection probability of ETCP depends on the average success probabilities of the pure ABR and WPR, which have been obtained in $\mathbf{Corollary}$ \ref{corollary1}  and $\mathbf{Corollary}$ \ref{corollary3}, respectively. Based on these results with ETCP, we have the success probability of hybrid relaying in the following theorem.  

% T2 
\begin{theorem}  \label{thm:SP_ETC}
The success probability of the hybrid relaying with ETCP at the steady states is 
%\vspace{-2mm}
\begin{align} \label{eqn:SP_ETC}
%& \mathcal{S}^{\textup{ETCP}}_{\textup{HR}} \! = \frac{1}{2} (\mathcal{S}_{\textup{ABR}}+\mathcal{S}_{\textup{WPR}}) + \frac{1}{2} \Bigg( \sum^{n}_{i=1} \sum^{i}_{j=1} {n \choose i} \mathcal{S}^{i}_{\textup{WPR}}(1-\mathcal{S}_{\textup{WPR}})^{n-i}   {n \choose i \!-\! j} \mathcal{S}^{i-j}_{\textup{ABR}} (1-\mathcal{S}_{\textup{ABR}})^{n-i+j} \nonumber \\
%& \hspace{25mm} -\sum^{n}_{i=1} \sum^{i}_{j=1} {n \choose i} \mathcal{S}^{i}_{\textup{ABR}}(1-\mathcal{S}_{\textup{ABR}})^{n-i}   {n \choose i \!-\! j} \mathcal{S}^{i-j}_{\textup{WPR}} (1-\mathcal{S}_{\textup{WPR}})^{n-i+j}  \Bigg) (\mathcal{S}_{\textup{WPR}}-\mathcal{S}_{\textup{ABR}}) , \\
& \mathcal{S}^{\textup{ETCP}}_{\textup{HR}} \! = \frac{1}{2} (\mathcal{S}_{\textup{WPR}}+\mathcal{S}_{\textup{ABR}}) + \frac{1}{2} \Big(  \phi_{1}\big(\mathcal{S}_{\textup{WPR}}\big) \phi_{2}\big(\mathcal{S}_{\textup{ABR}}\big)  -  \phi_{1}\big(\mathcal{S}_{\textup{ABR}}\big) \phi_{2}\big(\mathcal{S}_{\textup{WPR}}\big) \Big) (\mathcal{S}_{\textup{WPR}}-\mathcal{S}_{\textup{ABR}}),   % \nonumber \\
\end{align}
where $n \in \mathbb{N}^+$, $\mathcal{S}_{\textup{ABR}}$ and $\mathcal{S}_{\textup{WPR}}$ are given in  (\ref{eqn:theorem_ABR}) and (\ref{eqn:theorem_WPR}), respectively, and $\phi_{1}$ and $\phi_{2}$ are given, respectively, as
%\vspace{-2mm}
\begin{align} 
\phi_{1}(x)=\sum^{n}_{i=1} {n \choose i} %\mathcal{S}^{i}_{\textup{ABR}}
x^{i}
(1-x)^{n-i} \quad \text{and} \quad \phi_2(x) =\sum^{i}_{j=1}  {n \choose i \!-\! j} x^{i-j} (1-x)^{n-i+j} . \nonumber 
\end{align}
\end{theorem} %\vspace{-2mm}
The proof of $\textbf{Theorem}$~\ref{thm:SP_ETC} can be found in  \textbf{Appendix C}. 

\begin{comment}
%R5
{\bf Remark 5:} 
It is noted that, when $\mathcal{S}_{\textup{ABR}}$ is inferior to  $\mathcal{S}_{\textup{WPR}}$, we have $\mathcal{S}^{\textup{ETCP}}_{\textup{HR}} = \mathcal{S}_{\textup{ABR}} + (1 - \mathbb{P} [ \mathrm{M}_{\textup{ETCP}} = \mathrm{A} ]  )(\mathcal{S}_{\textup{WPR}} - \mathcal{S}_{\textup{ABR}} ) > \mathcal{S}_{\textup{ABR}} $. 
On the contrary, when $\mathcal{S}_{\textup{WPR}}$ is lower than $\mathcal{S}_{\textup{ABR}}$, we have  $\mathcal{S}^{\textup{ETCP}}_{\textup{HR}} = (\mathcal{S}_{\textup{ABR}} - \mathcal{S}_{\textup{WPR}}) \mathbb{P} [ \mathrm{M}_{\textup{ETCP}} = \mathrm{A} ] + \mathcal{S}_{\textup{WPR}}  > \mathcal{S}_{\textup{WPR}}$.
Therefore, when $\mathcal{S}_{\textup{ABR}}\neq \mathcal{S}_{\textup{WPR}}$, the success probability of the hybrid relaying with ETCP is strictly higher than that of the worse one between those of the pure ABR and the pure WPR. 
\end{comment}

When the hybrid relay uniformly selects between the ABR and WPR modes at random, the corresponding success probability can be easily obtained by averaging the success probabilities of the two modes, i.e., $\mathcal{S}^{\textup{URMS}}_{\textup{HR}}=\frac{1}{2} (\mathcal{S}_{\textup{ABR}}+\mathcal{S}_{\textup{WPR}})$. 
Let $\Upsilon \triangleq \mathcal{S}^{\textup{ETCP}}_{\textup{HR}} - \mathcal{S}^{\textup{URMS}}_{\textup{HR}} = \frac{1}{2} \Big  ( \phi_{1}\big(\mathcal{S}_{\mathrm{R}}\big)  \phi_{2}\big(\mathcal{S}_{\textup{WPR}}\big) -  \phi_{1}\big(\mathcal{S}_{\textup{WPR}})\phi_{2}\big(\mathcal{S}_{\textup{ABR}}\big) \Big) (\mathcal{S}_{\textup{WPR}}-\mathcal{S}_{\textup{ABR}})$. It is readily checked that $\phi_{1}\big(\mathcal{S}_{\textup{WPR}}\big)  \phi_{2}\big(\mathcal{S}_{\textup{ABR}}\big) -  \phi_{1}\big(\mathcal{S}_{\textup{ABR}})\phi_{2}\big(\mathcal{S}_{\textup{WPR}}\big)$ and $\mathcal{S}_{\textup{WPR}}-\mathcal{S}_{\textup{ABR}}$ take positive, negative, or zero at the same time. Therefore, we have $\Upsilon \geq 0 $, which yields the following observation.  

%When $\mathcal{S}_{\textup{WPR}} > \mathcal{S}_{\textup{ABR}}$,
 
%R5 
{\bf Remark 5:}  The success probability of the hybrid relaying with ETCP is strictly no worse than that with uniformly random mode selection.

\section{Analysis of Ergodic Capacity}
%\vspace{-2mm}
 
%Then, we proceed to
This section investigates the end-to-end ergodic capacity that can be achieved from the hybrid relaying. Specifically, we provide the general-case results for hybrid relaying with ESAP and ETCP and the special-case results for pure ABR and WPR.

The ergodic capacity of the hybrid relaying can be defined as: 
\begin{align} \label{def:ergodic_capacity}
& \mathcal{C}_{\textup{HR}}  \triangleq  \frac{(1-\omega)}{2} \Big ( C_{\mathrm{W}} %\mathcal{S}_{\mathrm{AW}} 
\mathbb{P}\big[\mathrm{M}=\mathrm{W}\big] +    C_{\mathrm{A}} \mathbb{P}\big[\mathrm{M}=\mathrm{A}\big] \Big) \nonumber \\
& \!\!\!= \frac{(1-\omega)}{2} \Big ( \mathbb{E} \Big[ W \log_{2} (1 + \nu ) \big| \mathrm{M}  = \mathrm{W}  \Big] \mathbb{P} [ \mathrm{M} \! = \! \mathrm{W}  ]  \! + \! C_{\mathrm{A}} \mathbb{P} \Big[\nu_{\mathrm{R}} > \tau_{\mathrm{W}}, \nu^{\mathrm{A}}_{\mathrm{D}} >  \tau_{\mathrm{A}} | \mathrm{M} \! = \! \mathrm{A} \Big] \mathbb{P} \big[\mathrm{M}  \! = \! \mathrm{A} \big] \Big),\!\!
\end{align}  
where the coefficient $\frac{1-\omega}{2}$ comes from the fact that the transmission for each hop occupies $\frac{1-\omega}{2}$ fraction of a time slot duration and the last equality follows %the assumption that  
the assumption that $C_{\mathrm{A}} \ll C_{\mathrm{W}} $.
 
%\vspace{-3mm}
\subsection{General-Case Results for ESAP} 
 
According to the definition in (\ref{def:ergodic_capacity}), the ergodic capacity of the hybrid relaying with ESAP is presented in the following theorem.

% T3 
\begin{theorem}\label{the:3}
The ergodic capacity of the hybrid relaying with ESAP is   
\begin{align} 
& \hspace{-2mm} \mathcal{C}_{\textup{HR}}   \! = \!  \frac{1\!-\!\omega}{2}   \Bigg(  \frac{W }{ \ln(2)}  \! \int^{\infty}_{\tau_{\mathrm{W}}}  
\frac{ \exp \! \big( \! -   \kappa (v)   \sigma^2     \big)}{1+v} \bigg( \! \chi_{\mathbb{A}}(v,\rho_{C})  
\Big( 1 \! - \! 
%\mathcal{L}^{-1} \bigg\{ \! \frac{1}{s}\mathrm{Det}\big(\mathrm{Id}\!+\!\alpha_{Q}\mathbb{Q}_{\Psi}(s)\big)^{\!-\frac{1}{\alpha_{Q}}} \! \bigg \} \! 
F_{Q_{\mathrm{R}}}
\big( \varrho_{\mathrm{W}}\!+\!\varrho_C \big) \! \Big)
%\int^{\infty}_{\frac{\rho_{\mathrm{W}} \! + \! \rho_{C} }{\omega \beta} } \!    f_{Q_{\mathrm{R}}}(q) \mathrm{d} q 
\! + \! \int^{  \varrho_{\mathrm{W}}+\varrho_C }_{\!  \varrho_{\mathrm{W}} } \!\! \! \!  
\chi_{\mathbb{A}}(v,\omega \beta q \!- \!  \rho_{\mathrm{W}}) 
\nonumber  \\ 
& \hspace{110mm}  
\times  f_{Q_{\mathrm{R}}}(q) \mathrm{d} q \! \bigg)   
\mathrm{d} v + \!  C_{\mathrm{A}}   \mathcal{S}^{\textup{ESAP}}_{\mathrm{A}} \Bigg),
% \! \Bigg( \!   \mho(\tau_{\mathrm{W}}) \!  \int^{\infty}_{\!  \varrho_{\mathrm{A}}    }  \!  \delta(q) f_{Q_{\mathrm{R}}}(q) \mathrm{d} q -   \exp \! \big( \! -   \kappa (v)   \sigma^2     \big) \bigg(\! \int^{\infty}_{ \varrho_{\mathrm{W}}+\varrho_C } \! \! \delta(q) \chi_{\mathbb{A}} (\tau_{\mathrm{W}},\rho_C) \nonumber \\ 
%& \hspace{60mm}  \times \!     f_{Q_{\mathrm{R}}} (q)  \mathrm{d} q - \! \int^{  \varrho_{\mathrm{A}}+\varrho_C }_{\! \varrho_{\mathrm{A}} } \! \delta (q)     %\nonumber \\    & \hspace{100mm}   \!  \times \! \chi_{\mathbb{A}} \big(\tau_{\mathrm{A}}, \omega \beta q \! - \!  \rho_{\!\mathrm{A}}  \big)   f_{Q_{\mathrm{R}}} (q) \mathrm{d} q \! \bigg) \!   \Bigg),
 \label{the:QI_Capacity_HR}
\end{align}
\begin{comment}
\begin{align} 
& \hspace{-2mm} \mathcal{C}^{\textup{QI}}_{\textup{HR}}  \! = \!  \frac{1\!-\!\omega}{2}   \bigg( \frac{W }{ \ln(2)} \! \int^{\infty}_{\tau_{\mathrm{A}}} \!  \frac{ \mathcal{S}^{\textup{QI}}_{\textup{HR} } (\tau_{v})
 }{1+v}   \mathrm{d} v   %\times \! %\frac{1}{1+v}
 + \!  C_{\mathrm{A}}   \mathcal{S}_{\mathrm{A}} \bigg) %\Bigg( \!  \mho(\tau_{\mathrm{A}}) \! \int^{\infty}_{\!  \varrho_{\mathrm{A}}    }  \!  \delta(q) f_{Q_{\mathrm{R}}}(q) \mathrm{d} q -   \exp \! \big( \! -   \kappa (v)   \sigma^2     \big)  \nonumber  
\end{align}
\end{comment}
%where $\mathcal{S}_{\mathrm{A}}$
where %$\delta(q)\triangleq\exp \! \Big(\!\!-\!\frac{ d^{\mu}_{\mathrm{R},\mathrm{D}} \widetilde{\sigma}^2  \tau_{\mathrm{A}}  }{ \eta \xi q } \! \Big)$, 
$\chi_{\mathbb{A}}$, $F_{Q_{\mathrm{R}}}(x)$,  $f_{Q_{\mathrm{R}}}(q)$, and $\mathcal{S}^{\textup{ESAP}}_{\mathrm{A}}$ are given in  (\ref{chi1}), (\ref{CDF}), (\ref{pdf}), and (\ref{eqn:SP_HR_B}), respectively.

%,   and $\chi_{\mathbb{A}}$ is \begin{align} \label{varpiA}  \chi_{\mathbb{A}}(\nu,p) = \exp \!   \bigg(\! \! -\!\frac{d^{\mu}_{\mathrm{R},\mathrm{D}} \sigma^2_{\Phi} v (1\!-\!\omega) }{2p} \! \bigg)  \mathrm{Det} \big(\mathrm{Id} \! + \! \alpha_{I} \mathbb{A}_{\Phi}( \varsigma(v,p) ) \big)^{\!- \frac{1}{ \alpha_{I}}} , \end{align}
%therein $\mathbb{A}_{\Phi}$ is given in (\ref{eq:kernal_I_R1}) and  $\ell(v,p)=\frac{d^{\mu}_{\mathrm{R},\mathrm{D}} v (1-\omega) }{ 2 p }$. 
\end{theorem}    

The proof of $\mathbf{Theorem}$~\ref{the:3} is presented in  $\mathbf{Appendix}$ $\mathbf{D}$.

\subsection{Special-Case results} 

For the analysis of the ergodic capacity, we also investigate the special cases when $\mathrm{R}$ relays the information from $\mathrm{S}$ to  $\mathrm{D}$ by using only ambient backscattering or only wireless-powered transmission. In particular, the result of the former case is presented  in the following corollary.

\begin{corollary}  \label{corrollary4}% C4 
The capacity of the pure ABR is
\begin{align} \label{Coro:capacity_ABR}
\mathcal{C}_{\textup{ABR}} \! = \! \frac{ (1\!-\!\omega) C_{\mathrm{A}} }{2}  %\exp \! \bigg(\! \! - \!\frac{d^{\mu}_{\mathrm{S},\mathrm{R}} \sigma^2_{\Phi}  \tau_{\mathrm{A}}  }{ P_{\mathrm{S}} }   \! \! \bigg)  
%\exp \! \Big( \! -   \kappa (v)   \sigma^2_{\Phi}   \!  \Big) \mathrm{Det} \Big( \mathrm{Id} \! + \! \alpha_{I} \mathbb{B}_{\Phi} \big(\kappa(\tau_{\mathrm{A}})\big) \! \Big)^{\!-\!\frac{1}{\alpha_{I}}} %\nonumber \\  & \hspace{80mm} 
\exp \! \Big(\! - \!  \kappa(\tau_{\mathrm{W}}) \sigma^2  \! \Big) \mathrm{Det} \Big( \mathrm{Id} \! + \! \alpha \mathbb{G}_{\Phi}(\mathbf{x},\mathbf{y}) \psi_{\mathbf{x}}\big(\kappa(\tau_{\mathrm{W}})\big)   \psi_{\mathbf{y}}\big(\kappa(\tau_{\mathrm{W}}) \big) \! \Big)^{\!\!-\!\frac{1}{\alpha}}
\! \! \!  \int^{ \infty }_{\! \varrho_{\mathrm{A}} } \! \! \!
%\exp \!  \bigg( \!\! -  \! \frac{ d^{\mu}_{\mathrm{R},\mathrm{D}}   \widetilde{\sigma}^2 \tau_{\mathrm{A}}}{ \eta \xi q } \! \! \bigg)
\delta(q)  f_{Q_{\mathrm{R}}}(q) \mathrm{d} q, \!  \!
\end{align}
where %$\kappa(v)=\frac{d^{\mu}_{\mathrm{S},\mathrm{R}} v }{ P_{\mathrm{S}} }$, 
$f_{Q_{\mathrm{R}}}(q)$ is given in %(\ref{mho}) and 
(\ref{pdf}). %, respectively.
\end{corollary}   
   
\begin{proof}
If the relay-to-destination transmission is performed  only through ambient backscattering, we have
\begin{align}\label{eqn:conditions} \mathbb{P}[\mathrm{M}=\mathrm{W}]=0 \quad \text{and} \quad \mathbb{P}[\mathrm{M}=\mathrm{A}]=\mathbb{P} \big [\nu_{\mathrm{R}} \! > \! \tau_{\mathrm{A}} , \nu^{\mathrm{A}}_{\mathrm{D}} \! > \! \tau_{\mathrm{A}}, E_{\mathrm{R}}  \!>\! E_{\mathrm{A}}  \big ]=\mathcal{S}_{\textup{ABR}},
\end{align}
By assigning (\ref{eqn:conditions}) into the definition in (\ref{def:ergodic_capacity}), 
we have the corresponding ergodic capacity as: 
\begin{align} \label{eqn:capacity_ABR}
\mathcal{C}_{\textup{ABR}} = \frac{1-\omega}{2} C_{\mathrm{A}} \mathcal{S}_{\textup{ABR}},
\end{align}
where $\mathcal{S}_{\textup{ABR}}$ has been obtained in $\textbf{Corollary}$ \ref{corollary1}. 

Subsequently, $\mathcal{C}_{\textup{ABR}}$ in (\ref{Coro:capacity_ABR})  directly yields by inserting  $\mathcal{S}_{\textup{ABR}}$ in (\ref{eqn:theorem_ABR}) into (\ref{eqn:capacity_ABR}).
\end{proof}

Moreover, if $\mathrm{R}$ performs relaying with only the wireless-powered transmission, we have the capacity of the pure WPR in the following corollary.  
   
\begin{corollary} \label{corollary5} % C5  
The capacity of the pure WPR %wireless-powered relaying 
is
\begin{align} \label{eqn:QI_WPR} 
& \hspace{-2mm} \mathcal{C}_{\textup{WPR}} \! = \!  \frac{W(1\!-\!\omega)}{2\ln(2)} \! \int^{\infty}_{\tau_{\mathrm{W}}} \! \frac{  %\exp (  -d^{\mu}_{\mathrm{S},\mathrm{R}} \sigma^2  v P^{-1}_{\mathrm{S}}    )
 \exp \! \big( \! -   \kappa (v)   \sigma^2      \big)
 }{1+v} \! \bigg( \!  %\int^{\infty}_{\!\frac{\rho_{\mathrm{W}} \! + \! \rho_{C} }{\omega \beta} }
\chi_{\mathbb{A}}(v,\rho_{C})  
%f_{Q_{\mathrm{R}}}(q) \mathrm{d} q  
\Big( 1 \!- \! 
%\mathcal{L}^{-1} \bigg\{ \! \frac{1}{s}\mathrm{Det}\big(\mathrm{Id}\!+\!\alpha_{Q}\mathbb{Q}_{\Psi}(s)\big)^{\!-\frac{1}{\alpha_{Q}}} \! \bigg \} \! 
F_{Q_{\mathrm{R}}}
\big(  \varrho_{\mathrm{W}}\!+\!\varrho_C \big)  \! \Big)
\nonumber \\
& \hspace{90mm} + \!\!\!  \int^{ \varrho_{\mathrm{W}}+\varrho_C }_{\! \varrho_{\mathrm{C}}}  \!\! \!  
\chi_{\mathbb{A}}(v,\omega \beta q \! - \! \rho_{\mathrm{W}}) f_{Q_{\mathrm{R}}}(q) \mathrm{d} q \! \bigg) 
  \mathrm{d} v,
\end{align}
where $\chi_{\mathbb{A}}$, $F_{Q_{\mathrm{R}}}(x)$, and $f_{Q_{\mathrm{R}}}(q)$ are given in (\ref{chi1}), (\ref{CDF}), and (\ref{pdf}), respectively.
\end{corollary}   
   
\begin{proof}
From the proof of $\textbf{Corollary}$~\ref{corollary5}, we have %the probability that the 
the conditions that $\mathrm{R}$ only performs %wireless-powered relaying 
WPR
as:   
%\vspace{-2mm}
\begin{align} \label{eqn:condition2}  \mathbb{P}\big[\mathrm{M}=\mathrm{W}\big] =  \mathbb{P} \Big[\nu_{\mathrm{R}} \!>\! \tau_{\mathrm{W}}, E_{\mathrm{R}} \! > \! E_{\mathrm{W}}, \nu^{\mathrm{W}}_{\mathrm{D}} \! > \! \tau_{\mathrm{W}}  \Big] \quad \text{and} \quad \mathbb{P}\big[\mathrm{M}=\mathrm{A}\big] = 0.   
\end{align} 
By inserting (\ref{eqn:condition2}) into %the definition of
(\ref{def:ergodic_capacity}), we have   
\begin{align} \label{eqn:QI_Capacity_WPR}
\mathcal{C}_{\textup{WPR}} = \frac{(1-\omega)}{2}  \mathbb{E} \Big[ W \log_{2} (1 + \nu )  \mathbf{1}_{\{\nu_{\mathrm{R}} > \tau_{\mathrm{W}}, \nu^{\mathrm{W}}_{\mathrm{D}} >  \tau_{\mathrm{W}}, E_{\mathrm{R}} >  E_{\mathrm{W}}  \}} \Big]    ,  
\end{align}
where  $\mathbb{E} \Big[ W \log_{2} (1 + \nu )  \mathbf{1}_{\{\nu_{\mathrm{R}} > \tau_{\mathrm{W}}, \nu^{\mathrm{W}}_{\mathrm{D}} >  \tau_{\mathrm{W}}, E_{\mathrm{R}} >  E_{\mathrm{W}}  \}} \Big] $  has been obtained in (\ref{eqn:C_A_QI}). We, therefore, have $\mathcal{C}_{\textup{WPR}}$ in (\ref{eqn:QI_WPR}) by plugging (\ref{eqn:C_A_QI}) into (\ref{eqn:QI_Capacity_WPR}).
\end{proof}      
  
%\vspace{-5mm}
\subsection{General-Case Results for ETCP}   
% T4 
\begin{theorem}\label{the:4}
The ergodic capacity of the hybrid relaying with ETCP at the steady states is    %\vspace{-2mm}
\begin{equation} \label{thm:EC_ETC}  \mathcal{C}^{\textup{ETCP}}_{\textup{HR}}  =   \frac{1}{2} (\mathcal{C}_{\textup{ABR}}+\mathcal{C}_{\textup{WPR}}) + \frac{1}{2} \Big( %\sum^{n}_{i=1} \sum^{i}_{j=1} {n \choose i} \mathcal{S}^{i}_{\textup{WPR}}(1-\mathcal{S}_{\textup{WPR}})^{n-i}   {n \choose i \!-\! j} \mathcal{S}^{i-j}_{\textup{ABR}} (1-\mathcal{S}_{\textup{ABR}})^{n-i+j}
\phi_{1}(\mathcal{S}_{\textup{WPR}})\phi_{2}(\mathcal{S}_{\textup{ABR}}) - \phi_{1}(\mathcal{S}_{\textup{ABR}})\phi_{2}(\mathcal{S}_{\textup{WPR}}) \Big) (\mathcal{C}_{\textup{WPR}}-\mathcal{C}_{\textup{ABR}}), %\vspace{-2mm} %\nonumber \\
%& \hspace{25mm} -\sum^{n}_{i=1} \sum^{i}_{j=1} {n \choose i} \mathcal{S}^{i}_{\textup{ABR}}(1-\mathcal{S}_{\textup{ABR}})^{n-i}   {n \choose i \!-\! j} \mathcal{S}^{i-j}_{\textup{WPR}} (1-\mathcal{S}_{\textup{WPR}})^{n-i+j}  \Bigg) (\mathcal{C}_{\textup{WPR}}-\mathcal{C}_{\textup{ABR}}) ,  
\end{equation} 
where $\mathcal{C}_{\textup{ABR}}$, $\mathcal{C}_{\textup{WPR}}$, $\mathcal{S}_{\textup{ABR}}$, and $\mathcal{S}_{\textup{WPR}}$ have been obtained in (\ref{Coro:capacity_ABR}), (\ref{eqn:QI_WPR}), (\ref{eqn:theorem_ABR}), and (\ref{eqn:theorem_WPR}), respectively.
\end{theorem}    
% P4
\begin{proof}
Recall that, from the proof of Theorem \ref{thm:SP_ETC}, we have obtained the probability of $\mathrm{R}$ selecting the ABR mode and WPR mode under ETCP at the steady states %as $\mathrm{P}_{\textup{ABR}}$
in (\ref{P_METCP_B}) and (\ref{P_METCP_A}), respectively. According to the mode selection criteria of ETCP described in Section \ref{sec:network_model}, the ergodic capacity of the hybrid relaying with ETCP at the steady states can be expressed as %\vspace{-2mm}
\begin{equation} \label{EC_ETCP}
\mathcal{C}^{\textup{ETCP}}_{\textup{HR}}=   \frac{1}{2}( \mathcal{C}_{\textup{ABR}}+\mathcal{C}_{\textup{WPR}} ) + \frac{1}{2} ( \mathbb{P} [N_{\textup{WPR}} >  N_{\textup{ABR}} ] - \mathbb{P} [N_{\textup{ABR}} >  N_{\textup{WPR}} ]  )(\mathcal{C}_{\textup{WPR}} - \mathcal{C}_{\textup{ABR}} ), %\vspace{-2mm}
\end{equation}
By inserting the expressions of $\mathbb{P} [ N_{\textup{ABR}} >  N_{\textup{WPR}}]$, $\mathbb{P} [ N_{\textup{WPR}} >  N_{\textup{ABR}}]$, $\mathcal{C}_{\textup{ABR}}$, and $\mathcal{C}_{\textup{WPR}}$ obtained in (\ref{P_METCP_B}), (\ref{P_METCP_A}), (\ref{Coro:capacity_ABR}), and (\ref{eqn:QI_WPR}), respectively, into (\ref{EC_ETCP}), we have the expression of $\mathcal{C}^{\textup{ETCP}}_{\textup{HR}}$ in (\ref{thm:EC_ETC}).
\end{proof}

%\section{Application of the Analytical Results}

%\vspace{-3mm}
\section{Numerical Results}

In this section, we show numerical results to validate and evaluate the success probabilities and ergodic capacity of the hybrid relaying system analyzed in Section~\ref{sec:SP}. % with the introduced hybrid relay. 
To demonstrate the advantage of the proposed hybrid relaying with the mode selection protocols,
we compare their performance with that of the pure WPR, and the pure ABR. 
%The performance results of the dual-hop cooperative relaying with a hybrid relay, wireless-powered relay and ambient backscattering relay are labeled as ``HR", ``WPR" and ``ABR", respectively. 
The performance results of the hybrid relaying with ESAP and those with ETCP, the pure WPR and the pure ABR are labeled as ``HR-EASP", ``HR-ETCP", ``WPR" and ``ABR", respectively. 
In the simulation, the ambient emitters $\Psi$ and interferers $\Phi$ are distributed on a circular disc of radius $R = 500$ m with the relay node $\mathrm{R}$ centered at the origin. Besides, the source node $\mathrm{S}$ and destination node $\mathrm{D}$ are placed at (-$d_{\mathrm{S}, \mathrm{R}}$,0) and ($d_{\mathrm{R}, \mathrm{D}}$,0), respectively.    
$\Psi$ and $\Phi$ are considered as base stations and sensor devices with transmit power $\widetilde{P}_{T}=40$ dBm and $P_{T}=20$ dBm, respectively.  The frequency bandwidth of $\Psi$ and $\Phi$ are set as 20 MHz~\cite{B2014Kellogg} and 50 kHz~\cite{C.2010Gomez}, respectively. The noise variance is -120 dBm/Hz. We set the transmit power of the source node as $P_{\mathrm{S}}=P_{T}$.
If the WPR is adopted, the average circuit power consumption rate of %active relaying 
$\mathrm{R}$ 
is set at $\rho_{\mathrm{W}}=50\mu W$, which is within the typical power consumption range of a wireless-powered transmitter~\cite{X.2016Lu}.   For the ABR, we set $\rho_{\mathrm{A}}=5\mu W$, an order of magnitude smaller than $\rho_{\mathrm{W}}$. The normalized capacitor capacity $\rho_{C}$ is $0.02$ Joules/second.
The other system parameters adopted are listed in Table \ref{parameter_setting} unless otherwise stated.

%\begin{figure}\centering \includegraphics[width=0.45\textwidth]{tau_A4.eps} \caption{Coverage probability as a function of $\tau_A$.} \label{fig:tau_A} % \end{figure}

%\vspace{-5mm}
\begin{table*}[htp]
 \centering
 \caption{\footnotesize Parameter Setting.} \label{parameter_setting}
 %\vspace{-2mm}
 \begin{tabular}{|l|l|l|l|l|l|l|l|l|l|l|l|l|l|l|l|} 
 \hline
 Symbol & %$P_{\mathrm{S}}$, $P_{T}$  & 
 $\alpha$  %& $\widetilde{P}_{T}$ 
 & $\widetilde{\alpha}$  &  $\mu $ & $\widetilde{\mu}$ & $d_{\mathrm{S},\mathrm{R}}$, $d_{\mathrm{R},\mathrm{D}}$  & %$\rho_{\mathrm{A}}$ &
  $\tau_{\mathrm{W}} $ & %$\rho_{\mathrm{W}}$ &
  $\tau_{\mathrm{A}}$   & $\eta$ & $\beta$   & $\xi$   & $\omega$  & $ C_{\mathrm{A}}$ \\ 
 \hline 
 Value %& 10 mW 
 & -0.5 %& 1 W 
 & -1  & 3.5  & 3.0 & 5 m &    0 dB %5 $\mu$W 
  &  20 dB  %& 50 $\mu$W 
  & 0.375 & 0.5 & 0.25  & 0.4  & 50 kbps
      \\
 \hline 
\end{tabular}
%\vspace{-5mm}
\end{table*}

\subsubsection{SINR Threshold on Success Probability}
In Fig.~\ref{fig:tauAB}, the success probabilities $\mathcal{S}^{\textup{ESAP}}_{\textup{HR}}$ and $\mathcal{S}^{\textup{ETCP}}_{\textup{HR}}$ obtained in $\mathbf{Theorem}$ \ref{thm:QSI_SP_HR} and $\mathbf{Theorem}$ \ref{thm:SP_ETC}, respectively, are shown as functions of the SINR threshold $\tau_{\mathrm{W}}$. To demonstrate the accuracy of the analytical expressions, we compare them with the results generated by Monte Carlo simulations. It can be seen that for both $\mathcal{S}^{\textup{ESAP}}_{\textup{HR}}$ and $\mathcal{S}^{\textup{ETCP}}_{\textup{HR}}$, the analytical results match closely with the simulation results over a wide range of $\tau_{\mathrm{W}}$ and $\tau_{\mathrm{A}}$. 
Fig.~\ref{fig:ETCP_tauA} depicts $\mathcal{S}^{\textup{ETCP}}_{\textup{HR}}$ under different settings of $n$. %When $n=0$, ETCP uniformly chooses between the ABR and WPT modes at random.
For comparison, we also show the success probability of uniform random mode selection, i.e., $\mathcal{S}^{\textup{URMS}}_{\textup{HR}}=\frac{1}{2}(\mathcal{S}_{\textup{ABR}}+\mathcal{S}_{\textup{WPR}})$, labeled as ``URMS". 
It can be found that $\mathcal{S}_{\textup{ETCP}}$ under different settings of $n$ outperforms all $\mathcal{S}^{\textup{URMS}}_{\textup{HR}}$, which agrees with $\mathbf{Remark }$ $\mathbf{5}$.
Moreover, $\mathcal{S}_{\textup{ETCP}}$ monotonically increases with $n$. This is due to the fact that the more number of the time slots to explore, the higher chance the hybrid relay finds the averagely better-performed mode. %which achieves better performance by average.
We can see that the performance gap resulted from the increase of $n$ decreases when $n$ is large.
In the following simulations, the value of $n$ is set as 5 to avoid a lengthy exploration period.

%under different settings of the repulsion factors $\alpha_{Q}$ and   $\alpha_{I}$.  
%It is clear that the hybrid relaying achieves a better performance with the quasi-static interference than that of the fast-varying interference, i.e., $\mathcal{S}^{\textup{QI}}_{\textup{HR}}$ exceeds $\mathcal{S}^{\textup{FI}}_{\textup{HR}}$. The main reason is that, in the WPR mode, the SINR at the relay and destination nodes are correlated with the quasi-static interference. The information relaying from the relay to destination is likely to succeed if the transmission over the source-to-relay link is successful. By contrast, with the fast-varying interference, a successful transmission over the source-to-relay link would not imply the relaying has a higher chance to succeed. Moreover, the performance gap between $\mathcal{S}_{\textup{HR}}$ and $\mathcal{S}^{\textup{FI}}_{\textup{HR}}$ increases with $\tau_{\mathrm{A}}$. This is due to the fact that, with quasi-static interference, the successful transmissions of the two hops are more correlated when $\tau_{\mathrm{A}}$ becomes larger. 

%(which collaborates with the remark in section )

%a larger difference between $\tau_{\mathrm{A}}$ and $\tau_{\mathrm{A}}$ causes a larger gap between $\mathcal{S}^{\textup{QI}}_{\textup{HR}}$ and  $\mathcal{S}^{\textup{FI}}_{\textup{HR}}$.

 \begin{figure} 
 \centering
  %\vspace{-5mm}  
  \subfigure [ESAP\vspace{0mm}]%Comparison of coverage probabilities. ($\xi=0.2$)  
   {
 \label{fig:ESAP_tauA}
  \centering   
  \includegraphics[width=0.48 \textwidth]{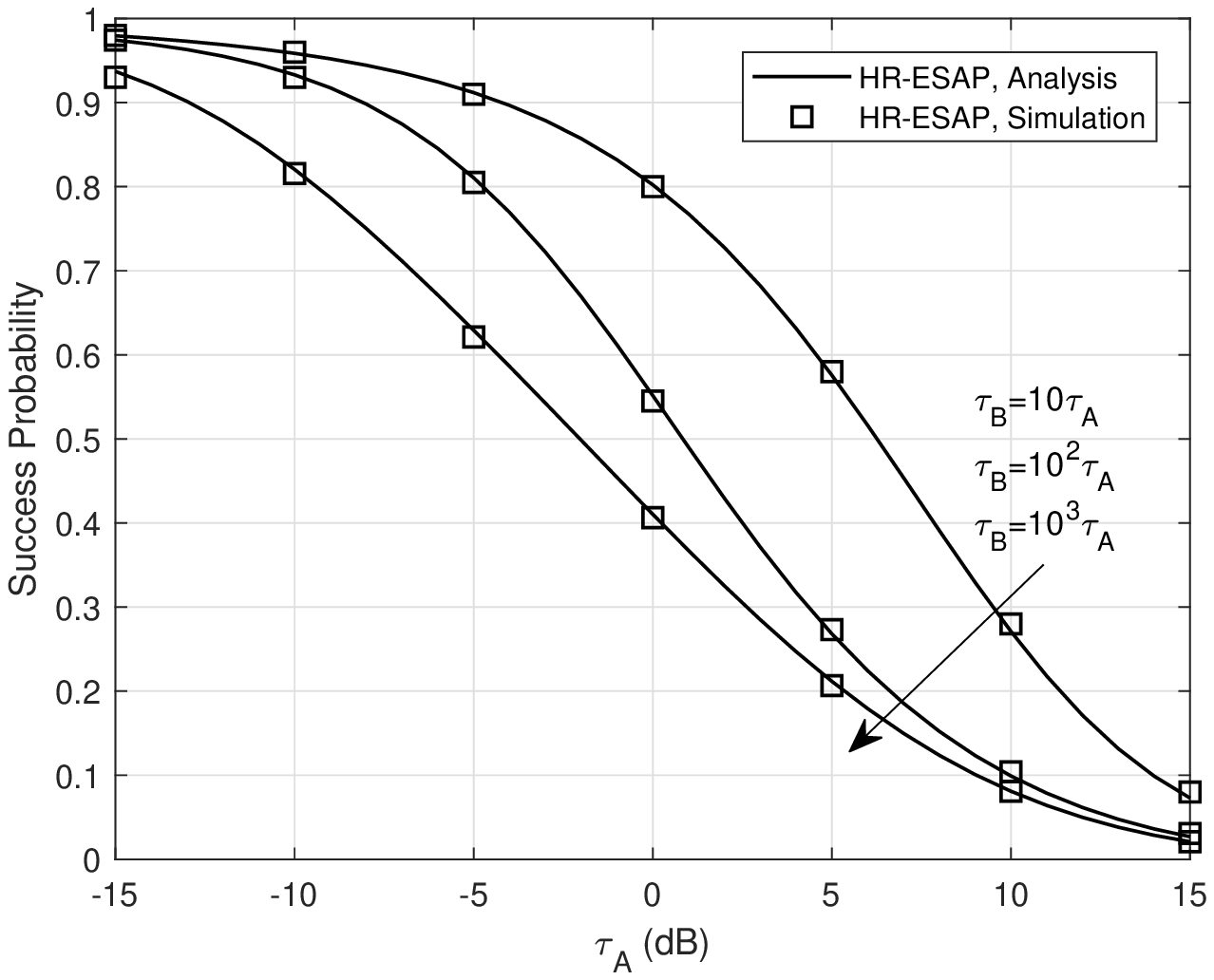}} 
  \centering  
  \subfigure  [ %Comparison of coverage probabilities. ($\xi=0.8$) 
 ETCP ($\tau_{\mathrm{A}}=10\tau_{\mathrm{W}} $) ] {
 \label{fig:ETCP_tauA}
  \centering
 \includegraphics[width=0.48 \textwidth]{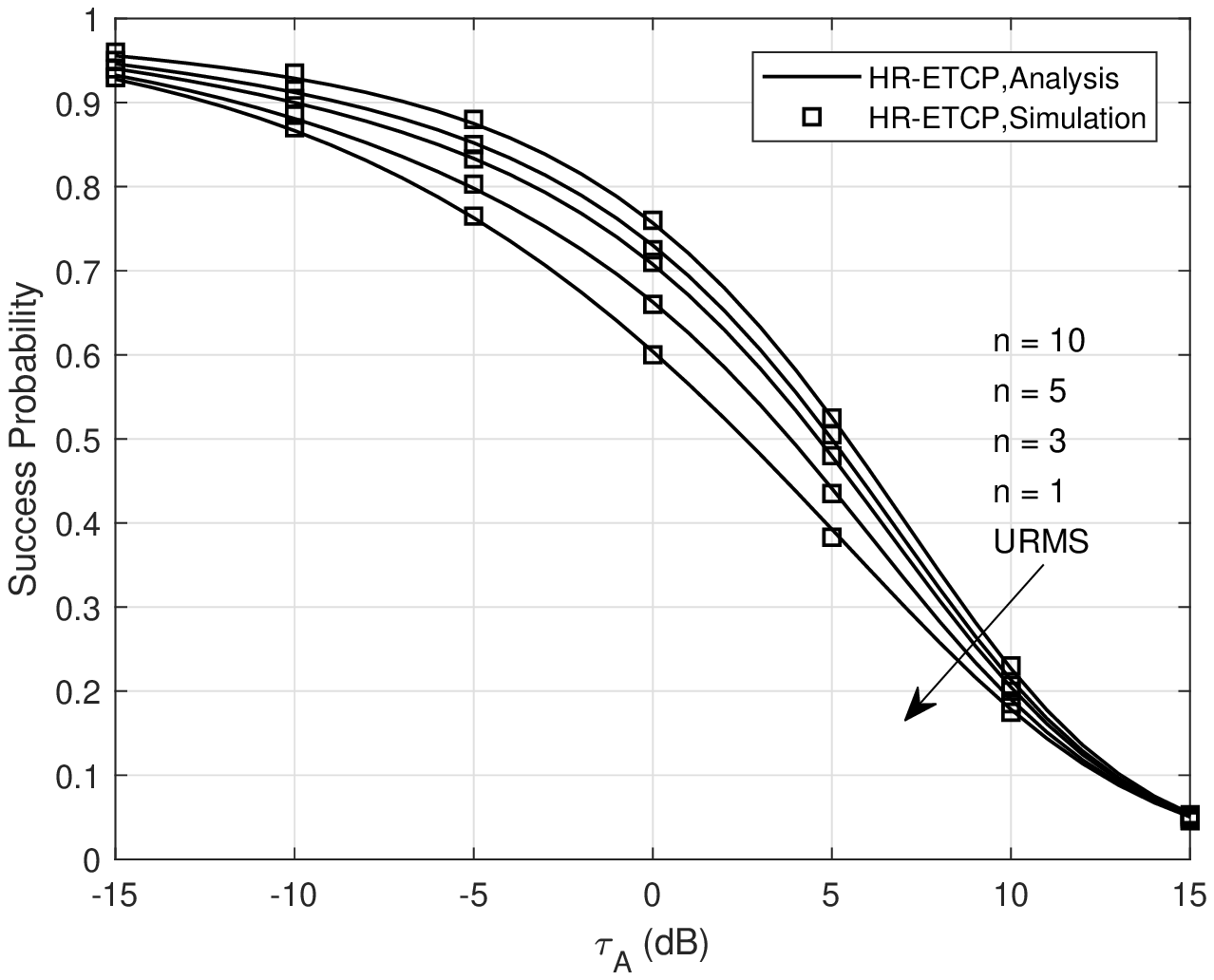}}
  %\vspace{-3mm}
 \caption{Success probability as a function of $\tau_{\mathrm{W}}$.   } 
 \centering
 \label{fig:tauAB}
 %\vspace{-8mm}
 \end{figure}
 
 \begin{figure}  
 %	\vspace{-3mm}
 \centering 
  \begin{minipage}[c]{0.48\textwidth}
  \includegraphics[width=0.95\textwidth]{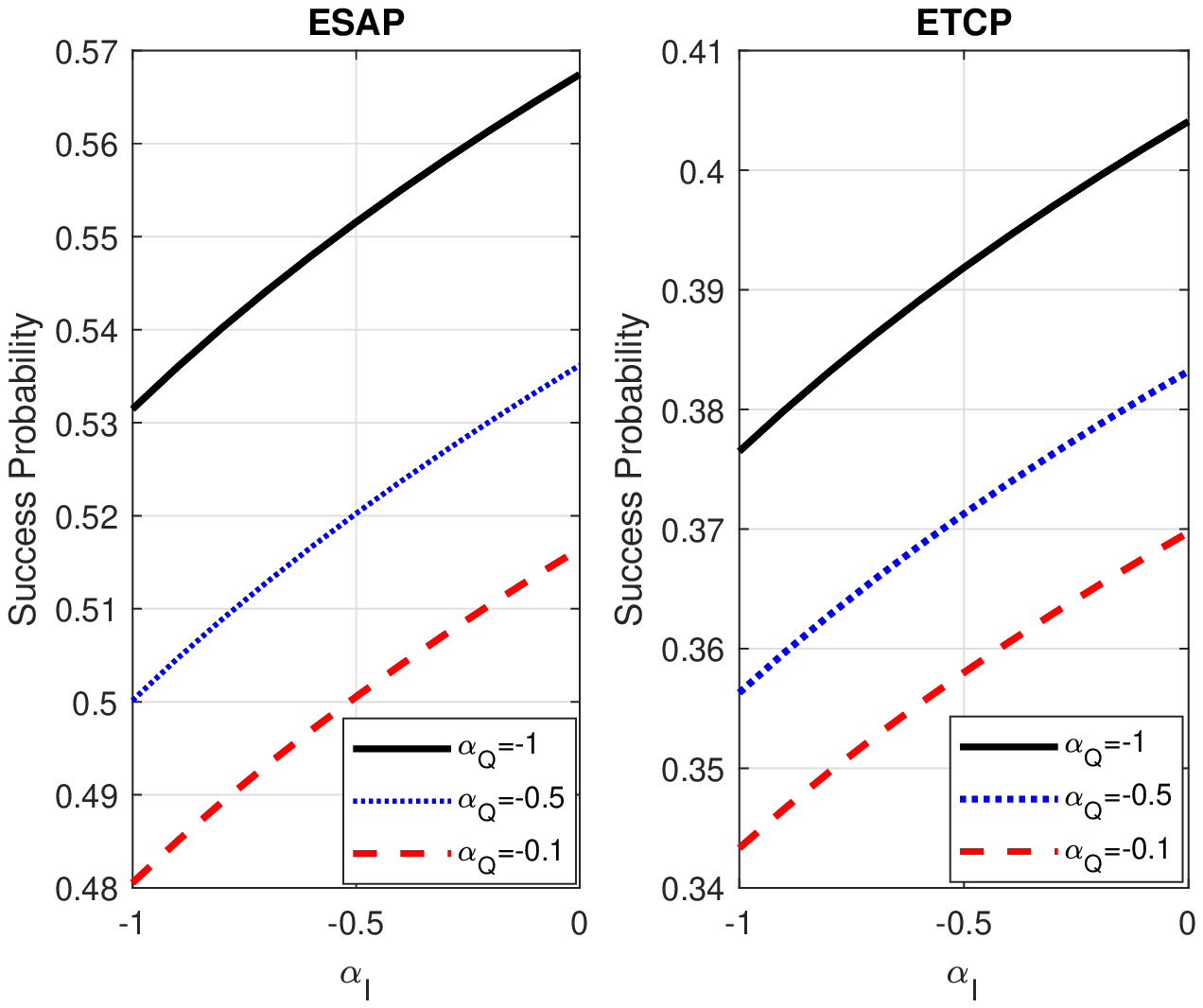} %\vspace{-5mm}
  \caption{Success probability as a function of $\alpha$. %($\widetilde{\zeta}=2000$/km$^2$).
  } \label{fig:SP_alpha}
  \end{minipage} %\vspace{-3mm}
   \begin{minipage}[c]{0.49\textwidth}
   \includegraphics[width=0.95\textwidth]{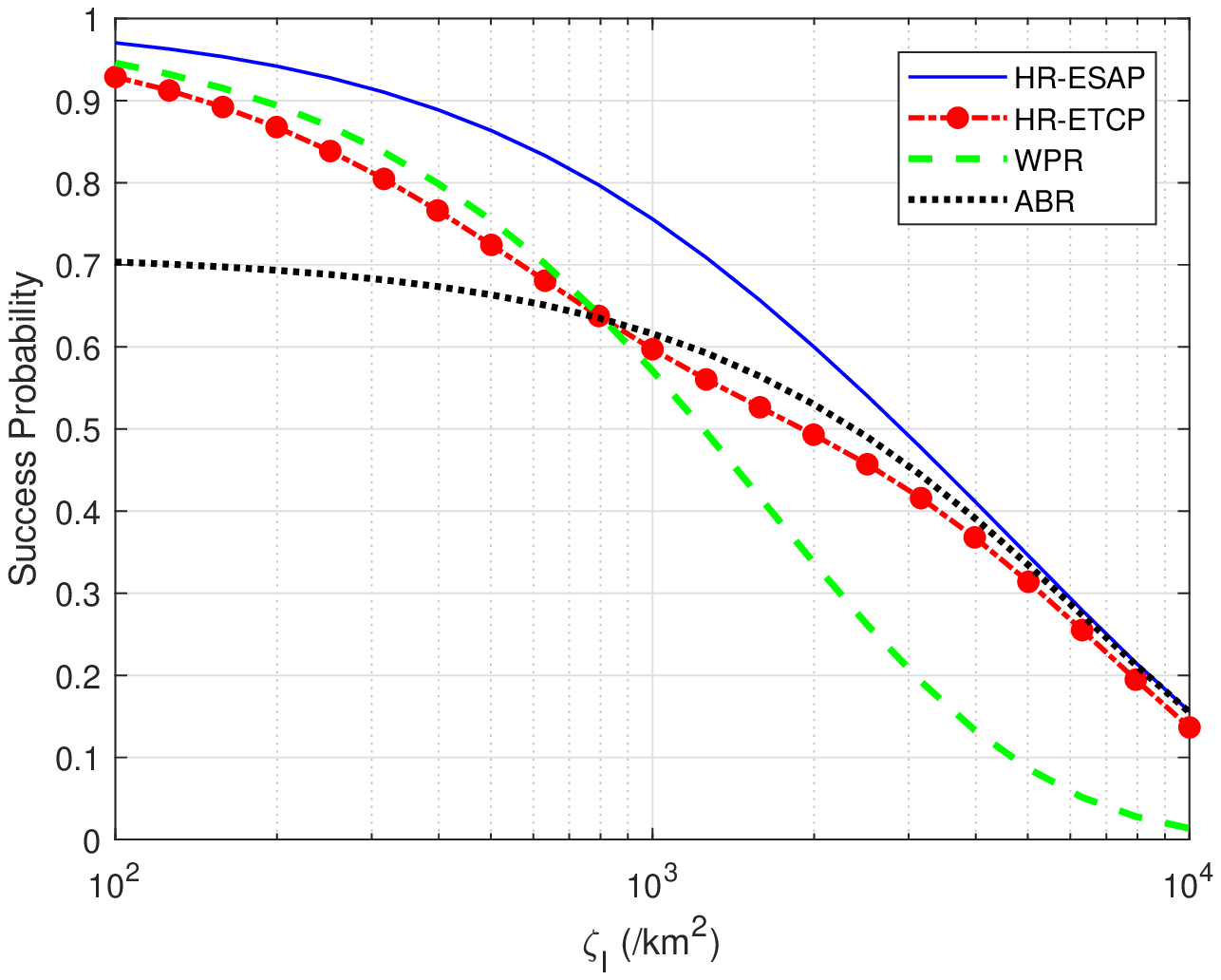}  %\vspace{-5mm}
   \caption{Success probability as a function of $\zeta$.   ($\widetilde{\zeta}=2000$/km$^2$) }\label{fig:SP_Rho_I} %\vspace{2mm}
   \end{minipage}
  %\vspace{-6mm} 
  \end{figure}

%Next, we compare the performance of $\mho_{\text{HR}}$ in (\ref{eqn:theorem_HR})  $\mho_{\text{WPR}}$ in (\ref{eqn:theorem_WPR}) and $\mho_{\text{ABR}}$ in (\ref{eqn:theorem_ABR}) to illustrate the advantage of the hybrid relay.

\subsubsection{Impact of System Environment (i.e., repulsion factors $\widetilde{\alpha}$ and $\alpha$ and densities $\widetilde{\zeta}$ and $\zeta$ of $\Psi$ and $\Phi$) on Success                 Probability}
%We then investigate the impact of system environment, i.e., repulsion factors $\alpha_{Q}$ and $\alpha_{I}$ and densities $\widetilde{\zeta}$ and $\zeta$ of $\Psi$ and $\Phi$, respectively. 
Fig.~\ref{fig:SP_alpha} shows the success probabilities $\mathcal{S}^{\textup{ESAP}}_{\textup{HR}}$ and $\mathcal{S}^{\textup{ETCP}}_{\textup{HR}}$ as functions of $\alpha$ under different $\widetilde{\alpha}$. 
It can be found that greater repulsion among the ambient emitters $\Psi$, i.e., smaller $\widetilde{\alpha}$, increases the success probabilities.
By contrast, greater repulsion among the interferers $\Phi$ decreases the success probabilities.
This comes from the fact that, given the spatial density, larger repulsion among the transmitters in $\Psi$ ($\Phi$)
results in the higher probability that some transmitters in $\Psi$ ($\Phi$) locate near $\mathrm{R}$, and thus stronger received signals from $\Psi$ ($\Phi$). As a result, smaller $\widetilde{\alpha}$ leads to more carrier signals for energy harvesting and smaller $\alpha$ generates more interference at $\mathrm{R}$. 

Fig.~\ref{fig:SP_Rho_I} studies how the success probabilities vary under different density of interferers $\zeta$. As expected, the success probabilities monotonically deceases with $\zeta$. 
We observe that %by integrating wireless-powered transmission together with ambient backscattering, 
the hybrid relaying with ESAP achieves a higher success probability %among the three types of relaying, 
than that of the pure ABR and the pure WPR, 
which corroborates the $\mathbf{Remark}$ $\mathbf{1}$  and $\mathbf{Remark}$ $\mathbf{3}$, %in Section~\ref{sec:SP}, 
respectively. 
Given the knowledge of receive SINR at the destination node, ESAP can switch the hybrid relay to the ABR %ambient backscatter relaying 
mode when the detected interference is high. Therefore, %the success probabilities of the hybrid relaying exceeds those of the wireless-powered relaying and approaches that of ambient backscatter relaying 
the hybrid relaying with ESAP outperforms the pure WPR %wireless-powered relaying 
and achieves comparable performance with the pure ABR %ambient backscatter relaying 
when $\zeta$ is large, e.g.,  $\zeta=10^4$/km$^2$, as shown in Fig.~\ref{fig:SP_Rho_I}. 
For the hybrid relaying with ETCP, %over a wide range of $\rho_{I}$, %the hybrid relaying with ETCP 
%$\mathcal{S}^{\textup{ETCP}}_{\textup{HR}}$ outperforms the worse one between the %pure ABR and WPR
%$\mathcal{S}_{\textup{ABR}}$ and $\mathcal{S}_{\textup{WPR}}$. %, which corroborates the remark 5. %It can be seen that when  
%Moreover,
$\mathcal{S}^{\textup{ETCP}}_{\textup{HR}}$ approaches the better-performed one than the worse-performed one in most conditions, which demonstrates the effectiveness of the exploration period in determining the better-performed mode.  
For the pure ABR, it is worth noting that, though the interference on the transmit frequency of $\Phi$ does not affect the %ambient backscatter relaying 
ABR 
link on the transmit frequency of $\Psi$, the interference still affects the %information reception link of the ambient backscatter relay,
source-to-relay link and thus $\mathcal{S}_{\textup{ABR}}$. 
Since only the source-to-relay link is affected, $\mathcal{S}_{\textup{ABR}}$ %the performance of the pure ABR %ambient backscatter relaying 
is more robust to the impact of increased interference than  $\mathcal{S}_{\textup{WPR}}$. %that of the pure WPR. %wireless-powered relaying.
This is evident from Fig.~\ref{fig:SP_Rho_I} that $\mathcal{S}_{\textup{ABR}}$ decreases with the increase of $\rho_{I}$ at a much slower rate than that of $\mathcal{S}_{\textup{WPR}}$.

Fig.~\ref{fig:SP_Rho_Q} examines the impact of the density of the ambient emitters $\widetilde{\zeta}$. In contrast to the influence of $\zeta$, the increase of $\widetilde{\zeta}$ augments the success probabilities of all types of relaying. The reason is that a larger $\widetilde{\zeta}$ increases not only the %probability that 
harvested energy at $\mathrm{R}$ for the circuit operation but also the transmit power of $\mathrm{R}$, either for wireless-powered transmission or ambient backscattering.
%With the increase of
When $\widetilde{\zeta}$ is relatively large, (e.g., above $2000$/km$^2$), 
the pure ABR outperforms
the pure WPR. %and approach $\mho_{\text{HR}}$. 
The reason is that when the signal power from ambient emitters is strong, the transmit power of the %wireless-powered relaying 
pure WPR
is largely limited by its capacitor capacity while that of the 
%ambient backscatter relaying 
pure ABR
does not have such a limitation. In other words, $P^{\mathrm{W}}_{\mathrm{R}}$ in (\ref{eqn:AR_R_dual}) stops increasing with $\widetilde{\zeta}$ once the capacitor is fully charged. By contrast, $P^{\mathrm{A}}_{\mathrm{R}}$ in (\ref{eqn:MB_R}) keeps increasing with $\widetilde{\zeta}$.
%Another observation is that the performance of the pure WPR largely deteriorates in the case with fast-varying interference compared to the case with quasi-static interference. However, the performance of the hybrid relaying is much less affected by the types of interference. This reveals that by taking advantage of ambient backscattering, hybrid relaying can avoid the performance degradation in the WPR mode due to fast-varying interference. 
 
%Two observations can be drawn here.
%Second, the performance gain of the hybrid relaying is more significant in case with fast-varying interference than in case with quasi-static interference, which collaborate 

\begin{figure} 
	%\vspace{-5mm}
\centering
 \begin{minipage}[c]{0.49\textwidth}
 \includegraphics[width=0.95\textwidth]{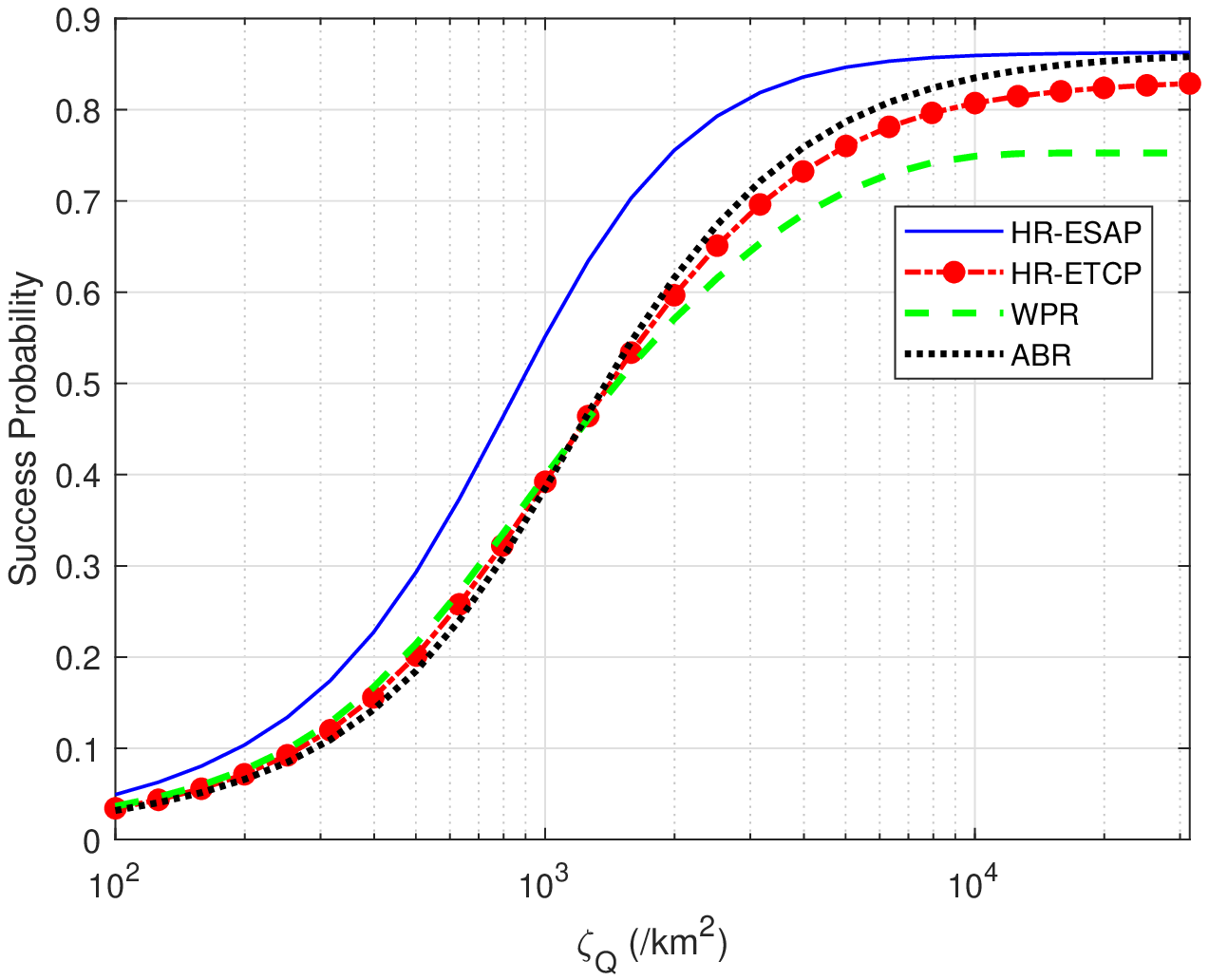} % \vspace{-5mm}
 \caption{Success probability as a function of $\widetilde{\zeta}$. % ($\rho_{C}=0.01$ Joule/second).
 } %\vspace{3mm} 
 \label{fig:SP_Rho_Q}
 \end{minipage}
 \begin{minipage}[c]{0.48\textwidth}
 \includegraphics[width=0.95\textwidth]{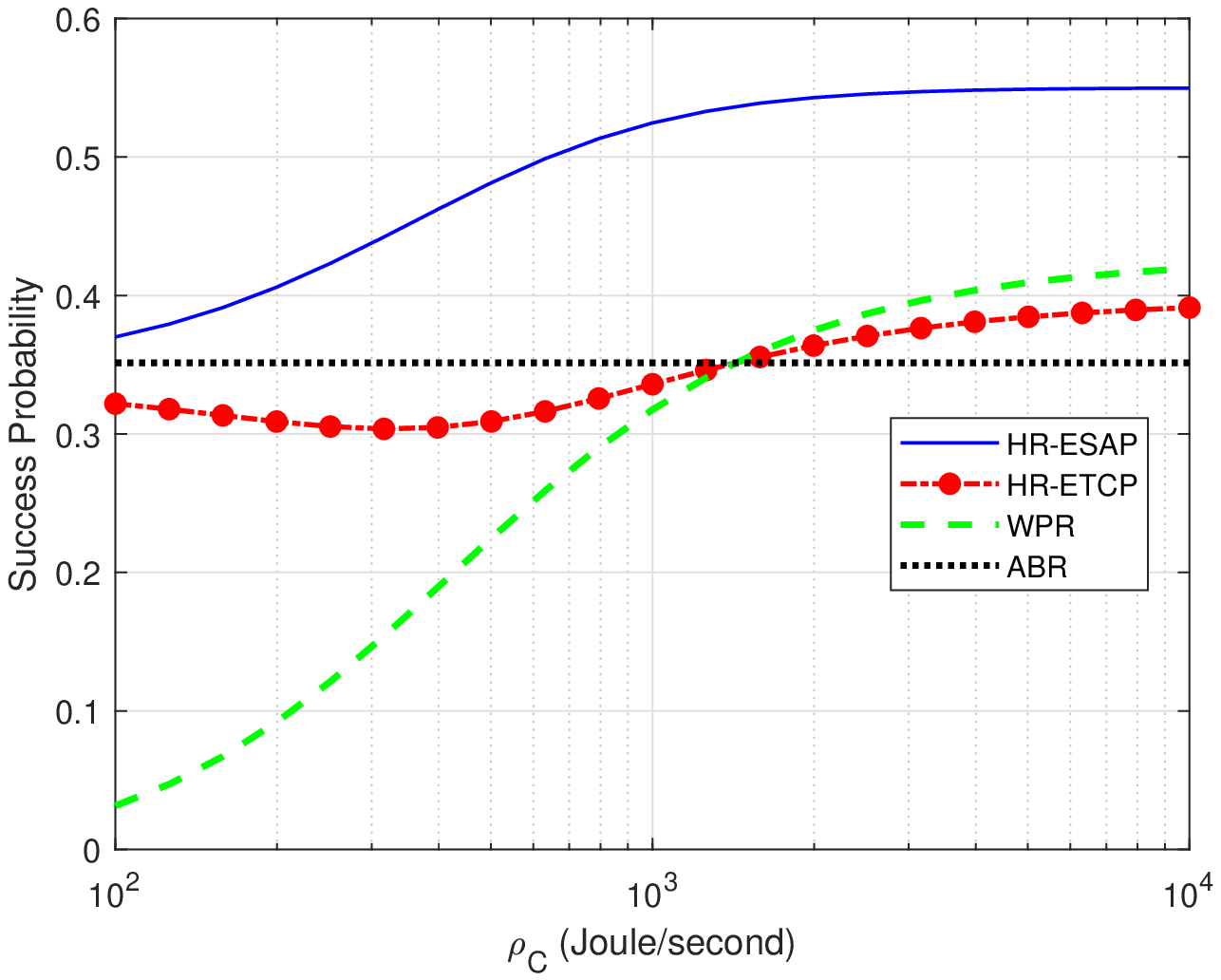}  %\vspace{-5mm}
 \caption{Success probability as a function of $\rho_{C}$. %($\mu=3.6$).
 } \label{fig:P_C} 
 \end{minipage}
% \vspace{-10mm} 
 \end{figure}

\subsubsection{Impact of Normalized Capacitor Capacity $\rho_{C}$ on Success Probability}
  
In Fig.~\ref{fig:P_C}, we investigate the impact of the normalized capacitor capacity %$E_{C}$ normalized over a time slot duration $T$ of the relaying protocol, 
i.e., $\rho_{C}=\frac{E_{C}}{T}$. It can be found that the capacity of the capacitor has a considerable impact on the success probabilities of the hybrid relaying with ESAP and the pure WPR. %wireless-powered relaying. 
The reason is that the capacitor capacity is directly related to the transmit power of the wireless-powered transmission. %wireless-powered relay.    
By contrast, $\mathcal{S}_{\textup{ABR}}$ remains steady with the variation of $\rho_{C}$, as the transmit power of the pure ABR %ambient backscatter relaying
is not related to the capacitor capacity. 
As $\mathcal{S}_{\textup{ABR}}$ outperforms $\mathcal{S}_{\textup{WPR}}$ in most of the shown range,  $\mathcal{S}^{\textup{ETCP}}_{\textup{HR}}$ approaches $\mathcal{S}_{\textup{ABR}}$ closely due to the exploration process, and thus less affected by $\rho_{C}$ compared to $\mathcal{S}^{\textup{ESAP}}_{\textup{HR}}$ and $\mathcal{S}_{\textup{WPR}}$.
We note that both 
the success probabilities of the hybrid relaying with ESAP and  
the pure WPR
are saturated when $\varrho_{C}$ becomes large. This implies that it is not necessary to equip an capacitor with oversized capacity %as large as possible 
for the hybrid relaying and the pure WPR. In practice, the capacitor capacity can be properly chosen to achieve a certain objective of success probability in the target network environment.  
Furthermore, by integrating ambient backscattering, the hybrid relaying with either ESAP or ETCP can relieve the requirement on capacitor capacity compared with the pure WPR. %This is evident from Fig.~\ref{fig:P_C} that the coverage probabilities of the hybrid relaying reaches the steady growth region at a much smaller $\rho_{C}$ than those of the pure WPR. %wireless-powered relaying. 

\subsubsection{Effect of Energy Harvesting Time Fraction $\omega$ on Ergodic Capacity}
Fig.~\ref{fig:capacity_omega} examines the impact of the energy harvesting time fraction $\omega$ of the relaying protocol on the ergodic capacity performance. 
We first validate the analytical expressions of $\mathcal{C}^{\textup{ESAP}}_{\textup{HR}}$, $\mathcal{C}^{\textup{ETCP}}_{\textup{HR}}$, $\mathcal{C}_{\textup{ABR}}$, and $\mathcal{C}_{\textup{WPR}}$ obtained in $\textbf{Theorem}$~\ref{the:3}, $\textbf{Theorem}$~\ref{the:4}, $\textbf{Corollary}$~\ref{corollary3},  and $\textbf{Corollary}$~\ref{corrollary4}, respectively.
It can be seen that our analytical results of ergodic capacity well match the Monte Carlo simulation results over a wide range of $\omega$. In terms of the ergodic capacity, the hybrid relaying with ESAP still achieves the higher ergodic capacity than those of the pure ABR and the pure WPR. %, while the hybrid relaying with ETCP still outperforms the worse one between the pure ABR and the pure WPR. 
%Moreover, %
We can observe that the plots of the hybrid relaying with ESAP and 
the pure WPR %wireless-powered relaying 
are unimodal functions of $\omega$ within the shown range. This reveals that there can be an optimal value of $\omega$ to maximize the ergodic capacity of the hybrid relaying and the pure WPR. %wireless-powered relaying. 
It is noted that the ergodic capacity of the pure ABR %ambient backscatter relaying 
is also a unimodal function of $\omega$. Due to the ultra-lower circuit power consumption, the maximal $\mathcal{C}_{\textup{ABR}}$ is achieved at a value much smaller than the shown range, and thus we omit displaying it. 
Moreover, %with the increase of the energy harvesting time fraction, 
the performance gap between the hybrid relaying with ESAP and the pure WPR %wireless-powered relaying 
becomes larger with the decrease of the energy harvesting time fraction.
This implies that the smaller the energy harvesting time, the greater the performance gain of the hybrid relaying with ESAP over the pure WPR.
%wireless-powered relaying. 
The reason is that ABR is adopted when the harvested energy is deficient for active transmission, which largely lowers the demand for capacitor compared to the pure WPR.

\begin{figure} 
	%\vspace{-5mm}
\centering
\begin{minipage}[c]{0.48\textwidth}
\includegraphics[width=0.95\textwidth]{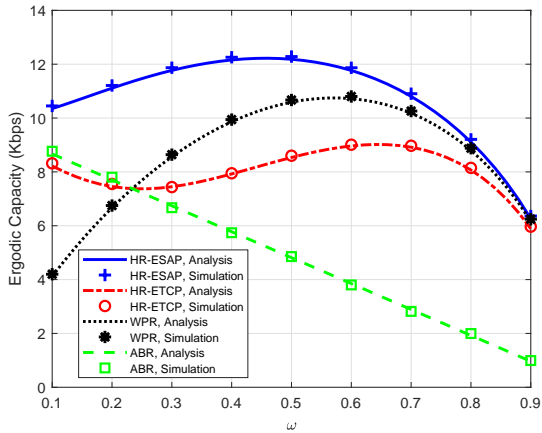}  %\vspace{-5mm}
\caption{Success probability as a function of the energy harvesting time fraction  %the energy harvesting coefficient
$\omega$.}\label{fig:capacity_omega}
\end{minipage}
\begin{minipage}[c]{0.48\textwidth}
\includegraphics[width=0.95\textwidth]{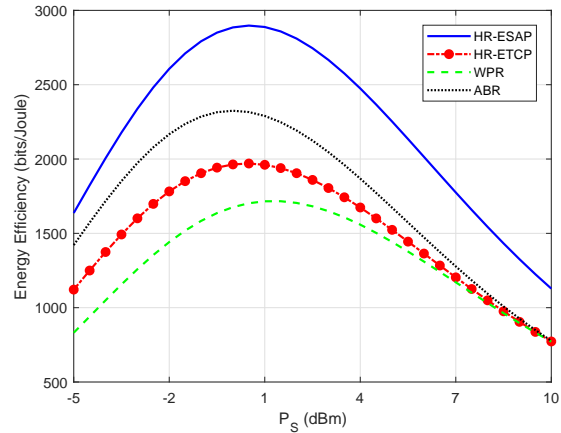}  %\vspace{-5mm}
\caption{Energy efficiency as a function of $P_{\mathrm{S}}$ ($\tau_{\mathrm{A}}=10\tau_{\mathrm{A}}$).} %\vspace{3mm}  \label{fig:EE_PS}
\end{minipage} 
%\vspace{-5mm} 
\end{figure}

\subsubsection{Impact of Transmit Power $P_{\mathrm{S}}$ on Energy Efficiency}

Based on the analytical results of ergodic capacity, we also evaluate the energy efficiency of the relaying system,  defined as the ergodic capacity versus the transmit power of the source node, i.e., $\mathcal{E}_{\textup{HR}}=\frac{\mathcal{C}_{\textup{HR}}}{P_{\mathrm{S}}}$.
As shown in Fig.~\ref{fig:EE_PS}, %We note that 
the ergodic capacities of all types of the relaying are unimodal functions of $P_{\mathrm{S}}$. Specifically, increasing $P_{\mathrm{S}}$ at first enhances the energy efficiency because of an increase of the ergodic capacity. However, as $P_{\mathrm{S}}$ keeps increasing, it becomes more dominant than the ergodic capacity which consequently
deteriorates the energy efficiency. It can also be observed that the maximum energy efficiencies for different types of the relaying are achieved at different values of $P_{\mathrm{S}}$.
The reason is that the energy efficiency depends not only on $P_{\mathrm{S}}$, but also $\omega$ which determines the amount of harvesting energy of $\mathrm{R}$ and the transmission time of each hop. 
Thus, other than $P_{\mathrm{S}}$, $\omega$ is also another key design parameter to improve the energy efficiency.

\subsubsection{Applications of Analytical Framework}

Furthermore, we demonstrate applications of the derived analytical framework in optimizing system parameters. %for optimal resource allocation of the hybrid relaying system with an objective 
In %wireless-powered
energy-constrained communication systems, power allocation is a crucial design issue. Therefore, we consider two power allocation-related design problems: transmit power minimization and energy efficiency maximization. %, referred to as $\textbf{P1}$ and $\textbf{P2}$, respectively. 
For the first problem, we minimize the transmit power of the source node with constraints on the minimum capacity in order to optimize the energy harvesting time fraction $\omega$. The formulation is expressed as follows: %\vspace{-3mm}
\begin{align*}
& %\vspace{-6mm}  
\mathbf{P1}: \hspace{3mm} \min_{\omega  }    \quad \quad P_{\mathrm{S}}  
\\   &  \begin{array}{r@{\quad}r@{}l@{\quad}l}
\text{subject to} &     \quad   \mathbf{C1}:  \quad   \mathcal{C}_{\textup{HR}}  \hspace{1mm}     \geq \hspace{1mm}  \mathcal{C}^{\textup{Target}}_{\textup{HR}}, \hspace{8mm} & %\leq b_i,  &i=1,2,3\ldots,n
\\
 &    \mathbf{C2}: \hspace{3mm} 0 \hspace{1mm}  \leq \hspace{1mm}  P_{\mathrm{S}}  \leq \hspace{1mm}  \bar{P}, \hspace{8.6mm}& \hspace{3mm} 0 \hspace{1mm}  \leq  \hspace{1mm}  P^{\mathrm{W}}_{\mathrm{R}} \hspace{1mm}  \leq \hspace{1mm}  \bar{P}, \\
  &     \mathbf{C3}: \hspace{3.5mm} 0 \hspace{1mm}  \leq  \hspace{1mm}  \omega  \hspace{1mm}  \leq \hspace{1mm}  1 . \hspace{9.8mm} & %\vspace{-3mm}  
  \\   
\end{array} 
\end{align*}
where $\mathcal{C}^{\textup{Target}}_{\textup{HR}}$ denotes the target ergodic capacity and $\bar{P}$ denotes the maximum transmit power for the source node and the relay node in the WPR mode. $\mathbf{C2}$ denotes the transmit power constraints for the source node and the hybrid relay, and $\mathbf{C3}$ denotes the time allocation constraint.
 
%to maximize the energy efficiency of the hybrid relaying system, 
For the second problem, we maximize the energy efficiency of the relaying system with the reliability constraint that the success probability of the hybrid relaying should be above some target value, denoted as $\mathcal{S}^{\textup{Target}}_{\textup{HR}}$. %under the constraints on the minimum success probability % of the hybrid relaying should be above some certain target, denoted as $T_{\textup{HR}}$. 
%on the minimum success probability 
%in order to jointly optimize the energy harvesting time fraction $\omega$ and transmit power $P_{\mathrm{S}}$.
%We consider a reliability requirement that the success probability of the hybrid relaying should be above some certain target, denoted as $T_{\textup{HR}}$. This problem is referred to as the optimal resource allocation problem, %. The optimal resource allocation problem 
The formulation is shown as follows: %\vspace{-2mm}
\begin{align*}
& \mathbf{P2}:  \hspace{3mm} \max_{\omega, P_{\mathrm{S}} }    \quad \quad \mathcal{E}_{\textup{HR}} = \frac{\mathcal{C}_{\textup{HR}}}{P_{\mathrm{S}}} \\ 
& \begin{array}{r@{\quad}r@{}l@{\quad}l}
\text{subject to} &     \quad   \mathbf{C1}:  \quad   \mathcal{S}_{\textup{HR}}  \hspace{1mm}     \geq \hspace{1mm}  \mathcal{S}^{\textup{Target}}_{\textup{HR}}, \hspace{4.7mm} & %\leq b_i,  &i=1,2,3\ldots,n
\\
 &    \mathbf{C2}: \hspace{3mm} 0 \hspace{1mm}  \leq \hspace{1mm}  P_{\mathrm{S}}  \leq \hspace{1mm}  \bar{P}, \hspace{6mm}& \hspace{3mm} 0 \hspace{1mm}  \leq  \hspace{1mm}  P^{\mathrm{W}}_{\mathrm{R}} \hspace{1mm}  \leq \hspace{1mm}  \bar{P}, \\
  &     \mathbf{C3}: \hspace{3.4mm} 0 \hspace{1mm}  \leq  \hspace{1mm}  \omega  \hspace{1mm}  \leq \hspace{1mm}  1 . \hspace{7.5mm} &  \\
\end{array} 
\end{align*}
where $\mathbf{C1}$ is the reliability requirement and $\mathbf{C2}$ and $\mathbf{C3}$ are the same as those in $\mathbf{P1}$.
Solving this problem provides us optimal choices of the energy harvesting time fraction $\omega$ and transmit power $P_{\mathrm{S}}$.

\begin{figure} 
 \centering  
 \subfigure  [ Maximum energy efficiency as a function of target success probability.
 ] {
\label{fig:MinPS_EC}
 \centering
\includegraphics[width=0.45 \textwidth]{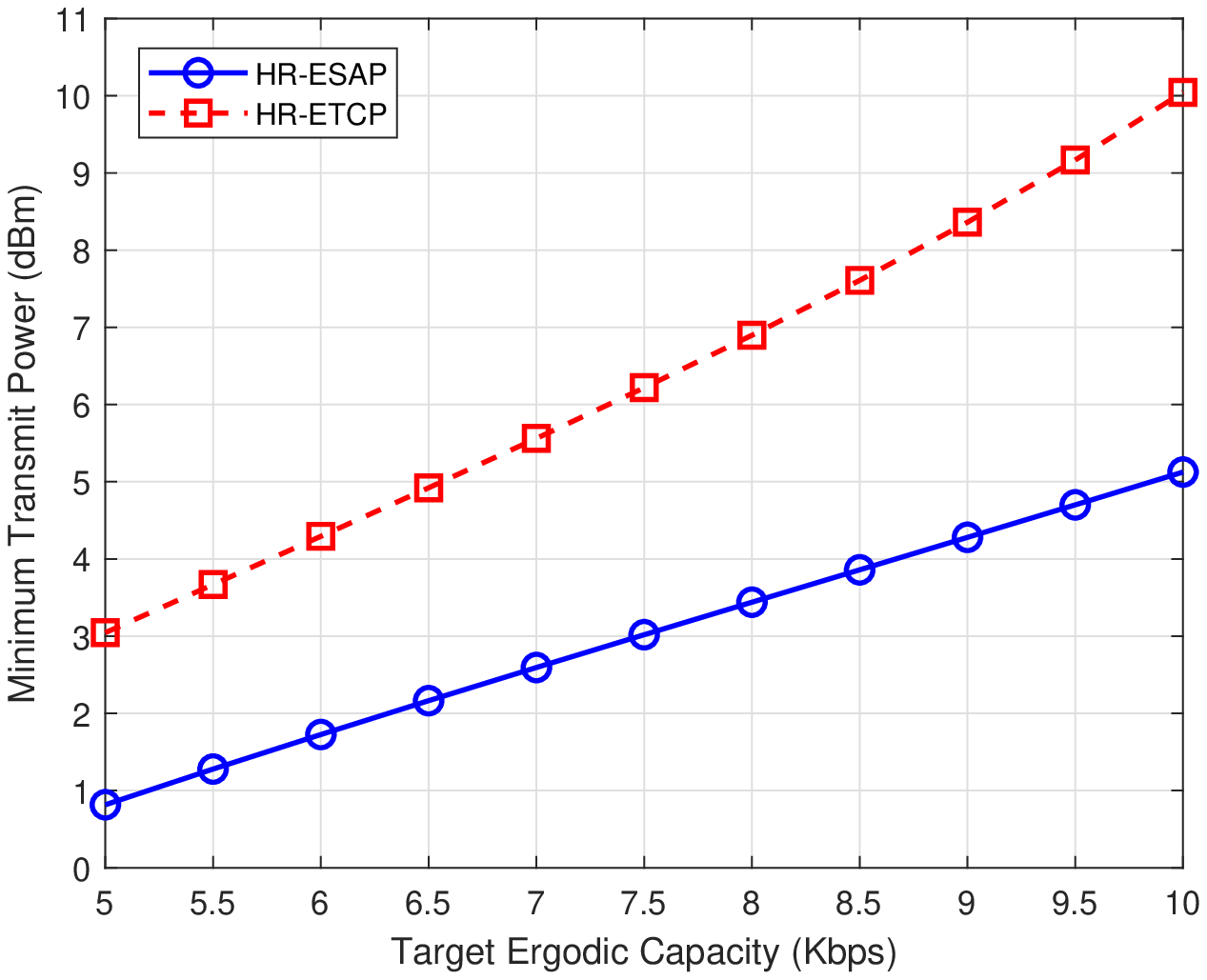}}
\centering
 %\vspace{-2mm}  
 \subfigure [ Optimal energy harvesting time fraction as a function of target success probability. \vspace{0mm}]
  {
\label{fig:Opt_omega_EC}
 \centering   
 \includegraphics[width=0.45 \textwidth]{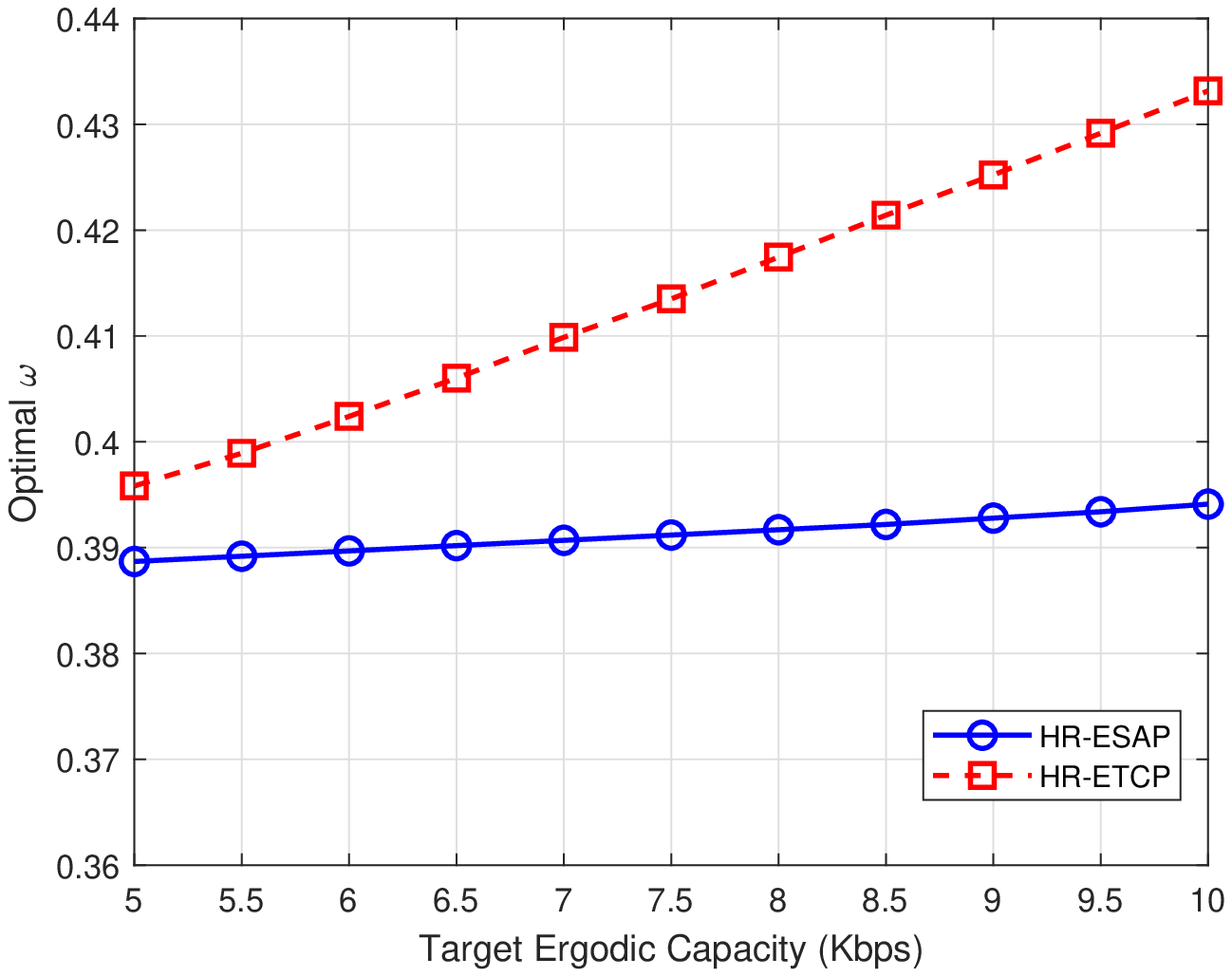}}
\caption{Optimal energy harvesting time fraction and maximum energy efficiency with reliability constraints. ($\rho_Q=2000/ $km$^2$, $\bar{P}=25$ dBm)} 
\centering
\label{fig:optimization_P1}
\vspace{0mm}
\end{figure}

\subsubsection{Numerical Solutions of Optimization Problems}
Next, we numerically solve the formulated optimization problems.  
Fig.~\ref{fig:MinPS_EC} and Fig.~\ref{fig:Opt_omega_EC} illustrates the minimum transmit power allocation and corresponding energy harvesting time allocation as functions of the target ergodic capacity, respectively, for $\mathbf{P1}$. As expected, the minimum transmit power is an increasing function of the target ergodic capacity. By utilizing the instantaneous CSI, the hybrid relaying with ESAP achieves the same target ergodic capacity with much lower transmit power than that with ETCP. Minimizing the transmit power has a larger impact on the optimal solutions of $\omega$ for the hybrid relaying with ETCP than that with ESAP.

Fig.~\ref{fig:MEE_reliability} and Fig.\ref{fig:omega_reliability} demonstrate the maximum energy efficiency and the corresponding energy harvesting time allocation as functions of the target success probability, respectively, for $\mathbf{P2}$. %the optimal resource allocation problem. 
Again, the hybrid relaying with ESAP outperforms that with ETCP in terms of maximum energy efficiency due to higher achieved capacity. It is observed that with the increase of the target success probability the maximum energy efficiency first remains steady and then decreases. 
%The reason is that the reliability requirement induces a constraint on the minimal value of $P_{\mathrm{S}}$. As aforementioned, the energy efficiency is a unimodal function of $P_{\mathrm{S}}$. When the reliability requirement is low, the maximum energy efficiency can alway achieve the peak value.
The reason is that larger
$P_{\mathrm{S}}$ is required to ensure higher reliability, which may sacrifice the energy efficiency.
Moreover, the optimal $\omega$ increases with the reliability requirement. This can be understood straightforwardly that more harvested energy is needed for the hybrid relay to perform WPR in order to guarantee a high reliability. 
The results of the optimal resource allocation problem can be used as guidance for setting the hybrid relaying to balance the tradeoff between energy efficiency and the reliability.

%Fig.~\ref{fig:MEE_reliability} reveals the tradeoff between the maximum energy efficiency and reliability. 

%We then demonstrate 
%For the hybrid relaying with ETCP, the ergodic capacity of it also manages to reach a performance closer to that of the better performed mode in most cases. %similarly to $\mathcal{S}^{\textup{ETCP}}_{\textup{HR}}$, 

%$\mathcal{C}^{\textup{ETCP}}_{\textup{HR}}$ 

%wireless-powered relaying. %This observation agrees with the finding in Fig~\ref{fig:P_C}.

\begin{figure} 
 \centering  
 \subfigure  [ Maximum energy efficiency as a function of target success probability.
 ] {
\label{fig:MEE_reliability}
 \centering
\includegraphics[width=0.45 \textwidth]{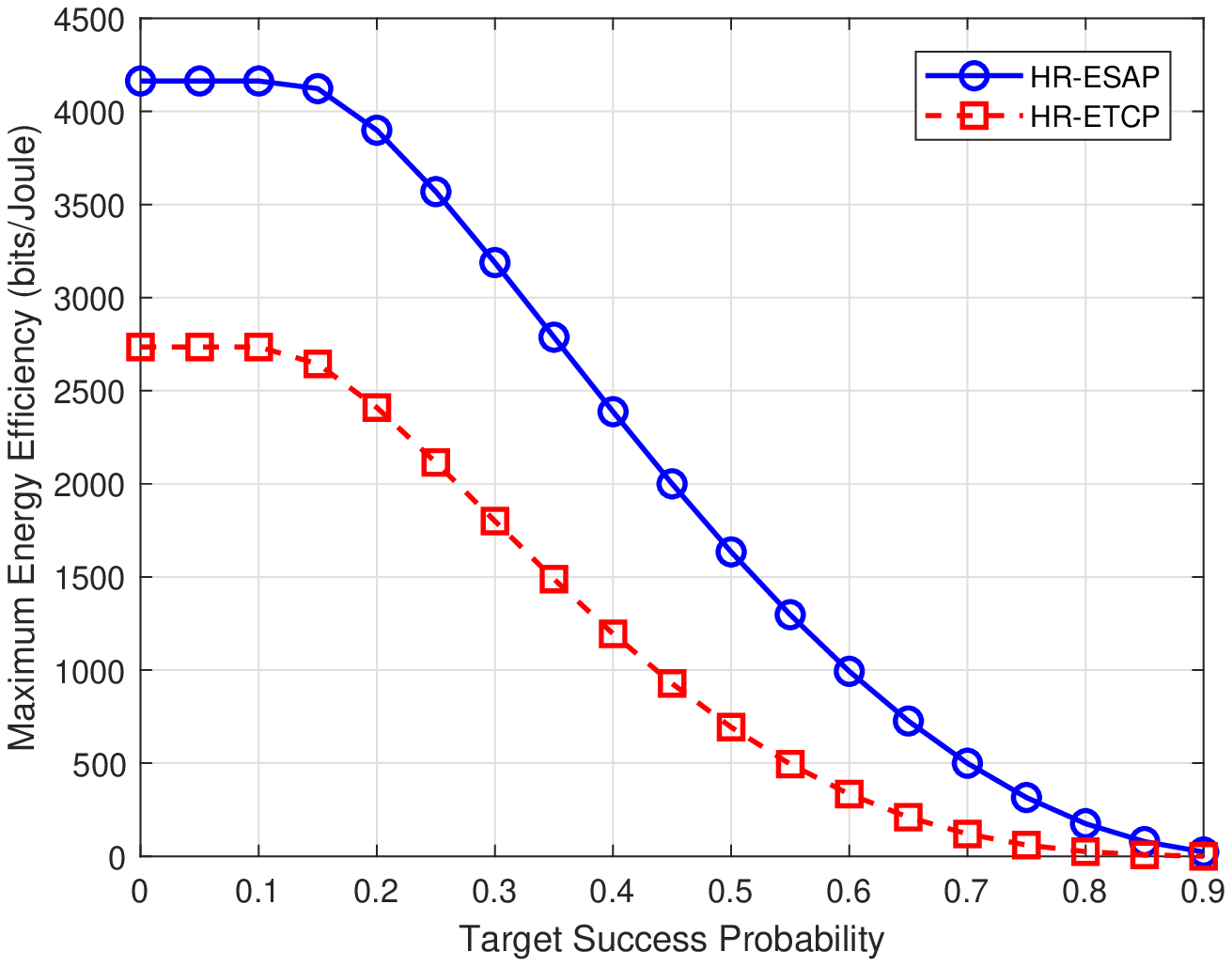}}
\centering
 %\vspace{-5mm}  
 \subfigure [ Optimal energy harvesting time fraction as a function of target success probability. \vspace{0mm}]
  {
\label{fig:omega_reliability}
 \centering   
 \includegraphics[width=0.45 \textwidth]{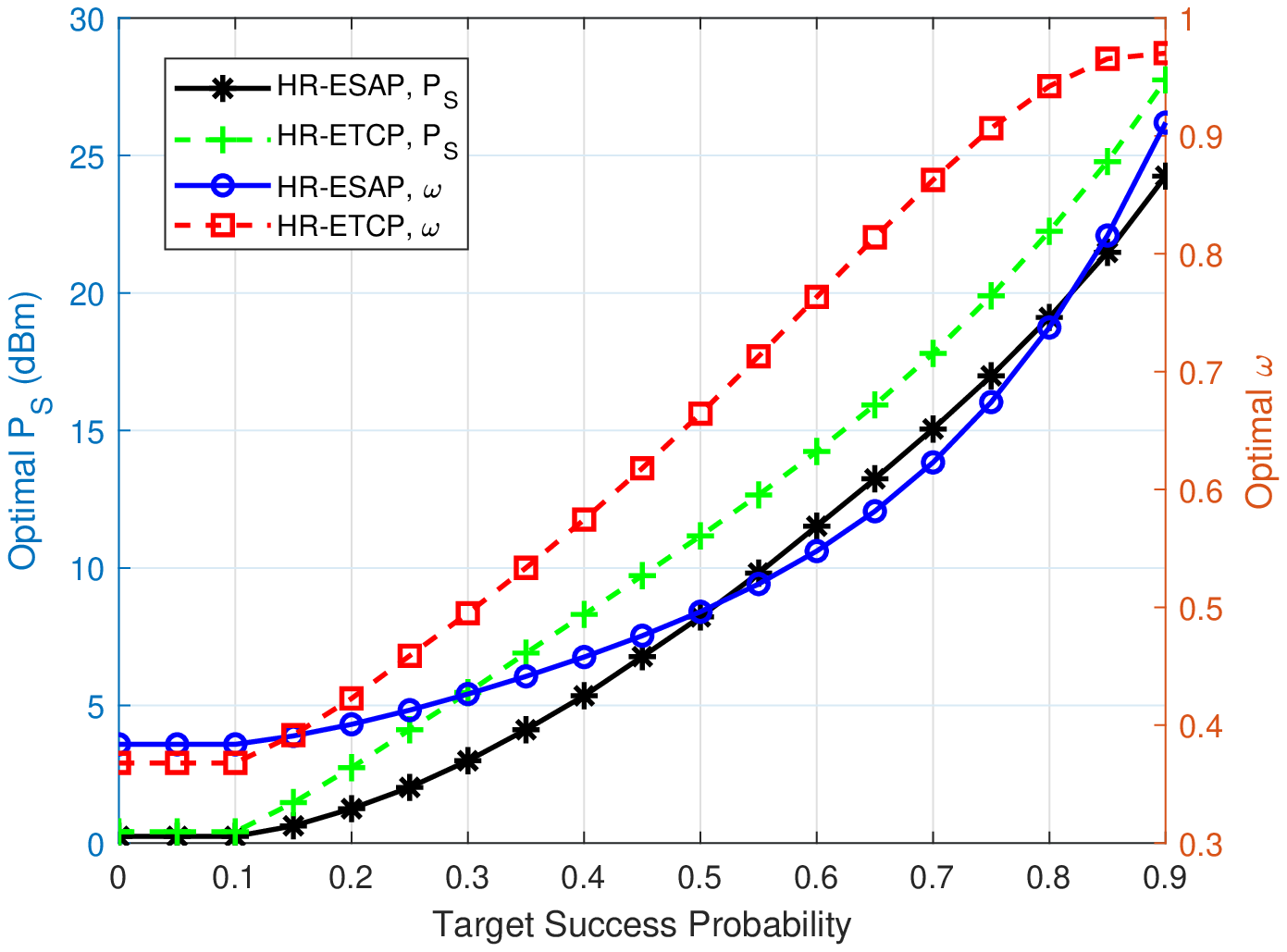}}
\caption{Optimal energy harvesting time fraction and maximum energy efficiency with reliability constraints. ($\rho_Q=2000/ $km$^2$, $\bar{P}=30$ dBm)} 
\centering
\label{fig:optimization}
\vspace{0mm}
\end{figure}

\begin{comment}
\begin{figure}
\centering
\includegraphics[width=0.45\textwidth]{Max_EE2.eps} %{multihop_models.eps} 
\vspace{-3mm}
\caption{Maximal energy efficiency versus reliability.} \label{fig:max_EE}
\vspace{-8mm}
\end{figure}
\end{comment}

%The power-splitting architecture has a significant performance advantage over the time-switching architecture, especially when ρ is around

%The optimal values of $\omega$ to reach the maximal capacity for the hybrid relaying are smaller than those for the wireless-powered relaying.  

%\vspace{-3mm}

\section{Conclusion}
%\vspace{-2mm}
We have proposed a hybrid relaying paradigm that is capable of operating %able to %perform two communication modes, i.e., work 
in either ambient backscatter relaying mode or wireless-powered relaying mode.  %paradigm, i.e., using
%namely, ambient backscatter assisted wireless-powered relaying. 
%ambient backscattering to assist wireless-powered relaying. 
%As the proposed paradigm relies on %ambient is highly dependent on the 
%ambient RF signals to operate, how to select the operational mode between WPR and ABR is
Both relaying modes are based on, but have different ways of utilizing ambient RF signals. %determining which one to adopt 
Therefore, how to %select  one relaying mode %communication function to adopt 
switch between the two relaying modes
under different network environment largely determines the performance of the hybrid relaying.
%To enable the hybrid relaying better adapt to different network environment with and without prior network information, 
%For the operation of the hybrid relaying, 
To address this issue,
we have devised two %operational mode selection 
protocols for the hybrid relaying to %choose a proper communication function 
perform operational mode selection
with and without instantaneous CSI. 
Considering the use of the hybrid relaying in a dual-hop relay system with spatially randomly located %carrier emitters and interferers, 
ambient transmitters,
we have derived the end-to-end success probabilities and ergodic capacity of the system under different mode selection protocols based on stochastic geometry analysis.
%The numerical results 
We have demonstrated analytically and numerically the superiority of the hybrid relaying %with prior information is superior to 
over the pure wireless-powered relaying and the pure ambient backscatter relaying, %, while the hybrid relaying without prior information outperforms the . 
as the proposed mode selection protocols effectively select the proper operation to adapt to the system environment. 
%One of our future directions is 
The analytical results reveal %the relationship among different system parameters and their
the impacts of different system parameters on the studied performance metrics and allow us to optimize the system parameters based on the objective.
Our analytical framework can be extended to investigate intelligent reconfigurable surface \cite{X.June2020Lu,S.2020Gong} assisted relaying.
Another promising direction is to design online learning (e.g., bandit learning~\cite{G.Aug.2020Li}) based mode selection protocol for the hybrid relaying.

%\vspace{-2mm}

\section*{Appendix} 
 
%\vspace{-2mm}
\subsection{Proof of Theorem 1}

% P1    
\begin{proof}

According to the  
criteria of ESAP 
described in Section \ref{sec:network_model}, we have the probability that the hybrid relay being in the WPR mode and the ABR mode, respectively, expressed as 
%\vspace{-2mm}
\begin{align}  
\mathbb{P}\big[\mathrm{M}_{\textup{ESAP} }=\mathrm{W}\big] =  \mathbb{P} \Big[\nu_{\mathrm{R}} \!>\! \tau_{\mathrm{W}}, E_{\mathrm{R}} \! > \! E_{\mathrm{W}}, \nu^{\mathrm{W}}_{\mathrm{D}} \! > \! \tau_{\mathrm{W}}  \Big], \label{eqn:CP_HR2}
\end{align}
%\vspace{-12mm}
\begin{align}
\text{and} \quad & \mathbb{P}\big[\mathrm{M}_{\textup{ESAP} }=\mathrm{A}\big] =  \mathbb{P}\Big[\nu_{\mathrm{R}} \!>\! \tau_{\mathrm{W}}, E_{\mathrm{R}} \!>\! E_{\mathrm{W}}, \nu^{\mathrm{W}}_{\mathrm{D}} \! \leq \! \tau_{\mathrm{W}}\Big]  + \mathbb{P} \Big[ E_{\mathrm{W}}  \! \geq \! E_{\mathrm{R}} \! > \! E_{\mathrm{A}} \Big]. \label{eqn:CP_HR3} 
\end{align} 
 
By inserting (\ref{eqn:CP_HR2}) and (\ref{eqn:CP_HR3}) into (\ref{definition_SP}), we have   
\begin{align}
& \mathcal{S}^{\textup{ESAP} }_{\textup{HR}} = \! % \prod_{ \substack{i \in \mathcal{N}\!/ \{ \mathrm{D}  \}, \\ i^+ \in \mathcal{N}\!/ \{ \mathrm{R}  \} }  } \! \!\! 
\underbrace{\mathbb{P} \Big [  \nu_{\mathrm{R}} \! > \! \tau_{\mathrm{W}}, \nu^{\mathrm{W}}_{\mathrm{D}} \! > \! \tau_{\mathrm{W}}, E_{\mathrm{R}}  \! > \! E_{\mathrm{W}}   \Big ] }_{ \mathcal{S}^{\textup{ESAP} }_{\mathrm{W}} }  
\!   %\underbrace{ }_{ \mathcal{P}^1_{\!\mathrm{A}} } 
\nonumber \\
& \hspace{15mm} + \underbrace{\mathbb{P} \Big [ \nu_{\mathrm{R}} \! > \! \tau_{\mathrm{W}}, \nu^{\mathrm{A}}_{\mathrm{D}} \! > \! \tau_{\mathrm{A}},   \nu^{\mathrm{W}}_{\mathrm{D}} \! \leq \! \tau_{\mathrm{W}},   E_{\mathrm{R}}       \!>\! E_{\mathrm{W} } \Big ] + \mathbb P \Big[ \nu_{\mathrm{R}} \! > \! \tau_{\mathrm{W}}, \nu^{\mathrm{A}}_{\mathrm{D}} \!>\! \tau_{\mathrm{A}}, E_{\mathrm{W}} \! \geq \! E_{\mathrm{R}} \! > \! E_{\mathrm{A}} \Big]}_{  \mathcal{S}^{\textup{ESAP} }_{\mathrm{A}}  } ,  \label{eqn:C_ARP} 
\end{align}
where $\mathcal{S}^{\textup{ESAP} }_{\mathrm{W}}$ ($\mathcal{S}^{\textup{ESAP} }_{\mathrm{A}}$) represents the joint probability that the transmission is successful and the hybrid relay is in the WPR (ABR) mode under ESAP. We first obtain $\mathcal{S}^{\textup{ESAP} }_{\mathrm{W}}$ as follows:
\begin{align} 
& \mathcal{S}^{\textup{ESAP} }_{\mathrm{W}} 
\!\! = \mathbb{P} \bigg[  h_{\mathrm{S},\mathrm{R}} > \frac{   d^{\mu}_{\mathrm{S}, \mathrm{R}} \tau_{\mathrm{W}} (I_{\mathrm{R}} \! + \! \sigma^2 ) }{ P_{\mathrm{S}} } ,  h_{\mathrm{R},\mathrm{D}}  > \frac{   d^{\mu}_{\mathrm{R}, \mathrm{D}} \tau_{\mathrm{W}} (I_{\mathrm{D}} \! + \! \sigma^2 ) }{ P^{\mathrm{W}}_{\mathrm{R}} } ,  \omega \beta T Q_{\mathrm{R}} > E_{\mathrm{W}}  \bigg] \nonumber \\ 
& 
\!\! \overset{\text{(a)}}{=}  \mathbb E  
\bigg[ \exp \! \bigg( \! \! - \! \frac{   d^{\mu}_{\mathrm{S}, \mathrm{R}} \tau_{\mathrm{W}} %(I_{\mathrm{R}} \! + \! \sigma^2 )
}{ P_{\mathrm{S}} } \Big( \! \sum_{j \in \Phi} \!\! \frac{P_{T} h_{j,\mathrm{R}}}{ d^{\mu}_{j,\mathrm{R}} }  \! + \! \sigma^2  \! \Big) \! \! \bigg) \exp \! \bigg( \! \! - \! \frac{   d^{\mu}_{\mathrm{R}, \mathrm{D}} \tau_{\mathrm{W}}  }{ P^{\mathrm{W}}_{\mathrm{R}} } \Big( \! \sum_{j \in \Phi} \!\! \frac{P_{T} h_{j,\mathrm{D}}}{ d^{\mu}_{j,\mathrm{D}} }  \! + \! \sigma^2  \! \Big) \! \! \bigg)   {\bf 1}_{\{Q_{\mathrm{R}} > \varrho_{\mathrm{W}} \} } \bigg]    \nonumber \\  
& \!\! \overset{\text{(b)}}{=} \exp \! \Big( \!\! - \! %\frac{  d^{\mu}_{\mathrm{S},\mathrm{R}}  \tau_{\mathrm{W}} }{ P_{\mathrm{S}}  } 
\kappa(\tau_{\mathrm{W}}) \sigma^2  \!  \Big) \! \Bigg( \! \!
 \exp \! \Big( \!\! -    
  \!  \ell(\tau_{\mathrm{W}},\rho_{C})     \sigma^2   \!   \Big)    \mathbb{E}_{   \Phi} \!  \Bigg[   \prod_{j \in \Phi} \! \Big( \! 1 \! + \! \frac{\kappa(\tau_{\mathrm{W}}) P_{T}}{ d^{\mu}_{j,\mathrm{R}} } \Big)^{\!\!\!-1} \!   \Big(\! 1 \! + \! \frac{ \ell(\tau_{\mathrm{W}},\rho_{C}) P_{T}  }{  d^{\mu}_{j,\mathrm{D}} } \Big)^{\!\!-1} \! \Bigg] \!\!\int^{\infty}_{\!\!\varrho_{\mathrm{W}}+\varrho_C} \!\!\!\!\! f_{Q_{\mathrm{R}}}(q) \mathrm{d} q  \nonumber \\
& \!\! \hspace{2mm} %\times \! {\bf 1}_{\{  Q_{\mathrm{R}} > \varrho_{\mathrm{W}} + \varrho_{C}  \} } \!   \Bigg] \!\! 
+ \!\! %\! \exp \! \big( \!  -    \!\kappa(\tau_{\mathrm{W}}) \sigma^2      \big) \!  %\mathbb{E}_{Q_{\mathrm{R}}} \Bigg[   
\int^{\varrho_{\mathrm{W}}+\varrho_C}_{\varrho_{\mathrm{W}} } \!\!\! \exp \! \Big( \!\! - \! \ell(\tau_{\mathrm{W}}, \omega \beta q - \rho_{\mathrm{W}}) \sigma^2   \! \Big) 
 \mathbb{E}_{\Phi}  \Bigg[   \prod_{j \in \Phi} \! \Big( \! 1 \! + \! \frac{\kappa(\tau_{\mathrm{W}}) P_{T}}{   d^{\mu}_{j,\mathrm{R}} } \!\Big)^{\!\!\!-1} \!   \Big( \! 1 \! + \! \frac{ \ell(\tau_{\mathrm{W}}, \omega \beta q - \rho_{\mathrm{W}}) P_{T}   }{  d^{\mu}_{j,\mathrm{R}} } \! \Big)^{\!\!-1} \!    \Bigg] f_{Q_{\mathrm{R}}}(q) \mathrm{d} q \!\nonumber \\
%& \hspace{80mm} \times \!    {\bf 1}_{\{ \varrho_{\mathrm{W}} + \varrho_{C} \geq Q_{\mathrm{R}} >  \varrho_{\mathrm{W}}  \} } \!   \Bigg]  \nonumber \\ 
& \!\! \overset{\text{(c)}}{=}  \exp \! \big( \! - \!   \kappa(\tau_{\mathrm{W}}) \sigma^2      \big) \! \Bigg( \! \exp \! \Big( \! \! - \!  \ell(\tau_{\mathrm{W}},\rho_{C}) \sigma^2 
%\frac{  d^{\mu}_{\mathrm{R},\mathrm{D}} \sigma^2  \tau_{\mathrm{W}} (1\!-\!\omega) }{ 2 \rho_{C}  }          
 \!    \Big)  
\mathrm{Det}    \Big(\mathrm{Id} \! + \! \alpha  \mathbb{G}_{\Phi}(\mathbf{x},\mathbf{y})\varphi_\mathbf{x}\big(\ell(\tau_{\mathrm{W}},\rho_{C})\big)\varphi_\mathbf{y}\big(\ell(\tau_{\mathrm{W}},\rho_{C})\big)  \! \Big)^{\! - \frac{1}{ \alpha}  } \nonumber \\
&  \times \! \Big( 1\!-\!F_{Q_{\mathrm{R}}} \big(\varrho_{\mathrm{W}}\!+\!\varrho_{C}\big) \! \Big) %\int^{\infty}_{  \varrho_{\mathrm{W}} \! + \! \varrho_{\mathrm{C}}  } \!    \! f_{Q_{\mathrm{R}}}(q) \mathrm{d} q 
%\nonumber \\
%& \hspace{6mm}  
\! + \! \! \int^{\varrho_{\mathrm{W}}   +   \varrho_{\mathrm{C}}}_{\! \varrho_{\mathrm{W}}} \! \! \! \mathrm{Det} \Big( \mathrm{Id} \! + \! \alpha 
\mathbb{G}_{\Phi}(\mathbf{x},\mathbf{y})\varphi_\mathbf{x} \big(\ell(\tau_{\mathrm{W}},\omega\beta q \! - \! \rho_{\mathrm{W}}) \varphi_\mathbf{y}\big(\ell(\tau_{\mathrm{W}},\omega\beta q \! - \! \rho_{\mathrm{W}})\big)
\! \Big)^{\!\!- \frac{1}{ \alpha}  }  \nonumber \\
& \hspace{85mm} \times \exp \! \Big( \! \! - \! \ell(\tau_{\mathrm{W}}, \omega \beta q -\rho_{\mathrm{W}} )  \sigma^2  \!  \Big)   f_{Q_{\mathrm{R}}}(q) \mathrm{d} q \! \Bigg), \!\! \label{eqn:SP_HR_A}
\end{align}  
where (a) follows the complementary CDF of the exponential random variable, i.e., $\mathbb{P}[h > x]=\exp(-x)$ for $h \sim \exp(1)$, (b) follows the Laplace transform of an exponential random variable, i.e., 
$\mathcal{L}_{h}(Z) = (1+Z)^{-1}$ for $h \sim \exp(1)$ % holds by inserting $P^{\mathrm{W}}_{\mathrm{R}}$ given in (\ref{eqn:AR_R_dual}) 
and applies the substitutions $\kappa(v)=\frac{d^{\mu}_{\mathrm{S},\mathrm{R}} v}{P_{\mathrm{S}}}$ and $\ell(v,p)=\! \frac{  d^{\mu}_{\mathrm{R},\mathrm{D}} v (1-\omega) }{ 2 p } $, and (c) applies $\textbf{Theorem}$ $1$ of~\cite{L.August2015Decreusefond}. % and $\mathbb{A}_{\Phi}$ is defined in (\ref{eq:kernal_I_R1}).  %$f_{Q_{\mathrm{R}}}(q)$ and $F_{Q_{\mathrm{R}}}(x)$ are the PDF and CDF of $Q_{\mathrm{R}}$, given in (\ref{pdf}) and (\ref{CDF}), respectively. 

Subsequently, we continue to derive $\mathcal{S}^{\textup{ESAP} }_{\mathrm{A}}$ as follows:
\begin{align}
 %& \mathbb{P} \Big [ \nu^{\mathrm{W}}_{\mathrm{R}} \! \geq \! \tau_{\mathrm{W}}, \nu^{\mathrm{A}}_{\mathrm{D}} \! \geq \! \tau_{\mathrm{A}},   \nu^{\mathrm{W}}_{\mathrm{D}} \! < \! \tau_{\mathrm{W}},   E_{\mathrm{R}}       \!\geq\! E_{\mathrm{W} } \Big ] \nonumber \\
 & \mathcal{S}^{\textup{ESAP} }_{\mathrm{A}} = \mathbb{P} \bigg[   h_{\mathrm{S},\mathrm{R}} \! > \! \frac{  d^{\mu}_{\mathrm{S},\mathrm{R}} \tau_{\mathrm{W}} (I_{\mathrm{R}} \!  + \! \sigma^2 ) }{ P_{\mathrm{S}}  } ,    \widetilde{h}_{\mathrm{R},\mathrm{D}} \! > \! \frac{ d^{\mu}_{\mathrm{R},\mathrm{D}} \widetilde{\sigma}^2  \tau_{\mathrm{A}}  }{ \eta \xi Q_{\mathrm{R}} } , h_{\mathrm{R},\mathrm{D}} \! \leq \! \frac{  d^{\mu}_{\mathrm{R},\mathrm{D}}\tau_{\mathrm{W}} (I_{\mathrm{D}} \! + \! \sigma^2 ) }{ P^{\mathrm{W}}_{\mathrm{R}}  } ,  \omega \beta T Q_{\mathrm{R}} \! > \! %\frac{ }{\omega \beta}  
 Q_{\mathrm{W}}
 \bigg] \nonumber \\
 & \hspace{50mm}  + \mathbb P \bigg[  h_{\mathrm{S}, \mathrm{R}} \! > \! \frac{ \tau_{\mathrm{W}}d^{\mu}_{\mathrm{S},\mathrm{R}} (I_{\mathrm{R}} \! + \! \sigma^2  )  }{   P_{\mathrm{S}} }  , \widetilde{h}_{\mathrm{R},\mathrm{D}} \! > \! \frac{ d^{\mu}_{\mathrm{R},\mathrm{D}} \widetilde{\sigma}^2  \tau_{\mathrm{A}}  }{ \eta \xi Q_{\mathrm{R}} }   ,   E_{\mathrm{W}} \! \geq \! \omega \beta T Q_{\mathrm{R}}  \! >  \! E_{\mathrm{A}}   \bigg].
 \nonumber \\
& \overset{\text{(d)}}{=}   
\exp \! \big( \!  - \!  \kappa(\tau_{\mathrm{W}})\sigma^2    \big)
\! \bigg ( \!  
%\mathrm{Det} \Big( \mathrm{Id} + \alpha_{I} \mathbb{B}_{\Phi}\big(\kappa(\tau_{\mathrm{W}}) \big)  \! \Big)^{\!-\frac{1}{\alpha_{I}}} 
\mathrm{Det} \Big( \mathrm{Id} \! + \! \alpha \mathbb{G}_{\Phi}(\mathbf{x},\mathbf{y}) \psi_{\mathbf{x}}\big(\kappa(\tau_{\mathrm{W}})\big)   \psi_{\mathbf{y}} \big(\kappa(\tau_{\mathrm{W}}) \big) \! \Big)^{\!\!-\!\frac{1}{\alpha}}
\!\!\!  \int^{\infty}_{  \varrho_{\mathrm{A}}   } \!\!  
\delta(q)
 f_{Q_{\mathrm{R}}}(q) \mathrm{d}q   \!   \nonumber \\
 & \hspace{30mm} - % \exp \! \big( \! \! - \!   \ell(\tau_{\mathrm{W}}, \rho_C) \sigma^2    \big)  \mathrm{Det}    \Big(\mathrm{Id} \! + \! \alpha_{I}   \mathbb{G}_{\Phi}(\mathbf{x},\mathbf{y})\varphi_\mathbf{x}\big(\ell(\tau_{\mathrm{W}},\rho_{C})\big)\varphi_\mathbf{y}\big(\ell(\tau_{\mathrm{W}},\rho_{C})\big)  \! \Big)^{\!\! - \frac{1}{ \alpha_{I}}} 
 \chi_{\mathbb{A}} (\tau_{\mathrm{W}}, \rho_{C})   \!  \int^{\infty}_{  \varrho_{\mathrm{W}} + \varrho_{C}  } \! \!\! \delta(q) f_{Q_{\mathrm{R}}}(q) \mathrm{d}q     %\nonumber \\
 %& \hspace{35mm}  
 -   \!\int^{ \varrho_{\mathrm{W}} \! + \! \varrho_{C}  }_{ \varrho_{\mathrm{W}}  } \!\!\!    %\mathrm{Det}  \Big( \mathrm{Id} \! + \! \alpha_{\mathrm{I}} \mathbb{G}_{\Phi}(\mathbf{x},\mathbf{y})\varphi_\mathbf{x} \big(\ell(\tau_{\mathrm{W}},\omega\beta q \! - \! \rho_{\mathrm{W}}) \varphi_\mathbf{y}\big(\ell(\tau_{\mathrm{W}},\omega\beta q \! - \! \rho_{\mathrm{W}})\big) \! \Big)^{\!-\frac{1}{\alpha_{I}}}  \! %\bigg) 
 \chi_{\mathbb{A}} (\tau_{\mathrm{W}}, \omega\beta q \! - \! \rho_{\mathrm{W}} ) \delta(q) \!  f_{Q_{\mathrm{R}}}(q) \mathrm{d} q \! \bigg).
 %\nonumber \\  & \hspace{85mm} \times 
 \label{eqn:SP_HR_B} 
  \end{align} 
where (d) follows similar derivation steps of $\mathcal{S}^{\textup{ESAP}}_{\mathrm{W}}$ in (\ref{eqn:SP_HR_A}) and $\chi_{\mathbb{A}}(v, p)$ has been defined in (\ref{chi1}).

Finally, by plugging (\ref{eqn:SP_HR_A}) and (\ref{eqn:SP_HR_B}) into (\ref{eqn:C_ARP}), we have %the success probability of the hybrid relaying 
$\mathcal{S}^{\textup{ESAP}}_{\textup{HR}}$ expressed in (\ref{eqn:QI_HR}).
\end{proof}

% \vspace{-5mm}
\subsection{Proof of Corollary 2}    
    
\begin{proof}
The case when the ambient emitters and interferers are distributed following PPPs can be modeled as a special case of our adopted $\alpha$-GPP when $\alpha \to \infty$.
% When there exists no repulsion,
In particular, by using the expansion~\cite{ShiraiTakahashi}
\begin{eqnarray} \label{eqn:expansion}
	\mathrm{Det} \big(\mathrm{Id}+\alpha \mathbb{B}_{\Phi}(s) \big)^{-\frac{1}{\alpha} }
	%=1+ \alpha_{T}\int_{\mathbb{O}} \mathbb{A}_{\phi}(\mathbf{x},\mathbf{x})\mathrm{d}\mathbf{x}
	\stackrel{\alpha \to 0}{\longrightarrow} \exp \left( -\int_{\mathbb{O} } \mathbb{B}_{\Phi} (\mathbf{x},\mathbf{x} )\mathrm{d} \mathbf{x}\right),
	\end{eqnarray}
	we can simplify 
	the Fredholm determinant in (\ref{eqn:theorem_ABR}) as follows when $\alpha \to \infty$, $\widetilde{\alpha} \to \infty$, and
	$\mu=4$: %\begin{align} 
	\begin{align}  
	\mathrm{Det} \Big(\mathrm{Id}\!+\!\alpha \mathbb{B}_{\Phi}\big( \kappa(\tau_{\mathrm{W}}) \big) \!  \Big)^{-\frac{1}{\alpha}}
	&\!= \! \exp\! \left ( \! - 2 \pi \zeta   \int^{R\to\infty}_{0} \! \bigg( 1- \Big( 1\! + \! \frac{ d^{4}_{\mathrm{S},\mathrm{R}} \tau_{\mathrm{W}}  P_{T}}{ P_{\mathrm{S} } r^4} \Big)^{-1} \bigg) r \mathrm{d}r \right)  \nonumber \\ &
	 =  \exp \left (\!-  \frac{ 1}{ 2 } \pi^2 \zeta \sqrt{\frac{ d^{4}_{\mathrm{S},\mathrm{R}} \tau_{\mathrm{W}}  P_{T}}{ P_{\mathrm{S} }  }   }  \right). \label{eqn:Det_PPP}
	\end{align} 
Similarly, following the expansion in (\ref{eqn:expansion}),  $f_{Q_{\mathrm{R}}}$(q) can be simplified as follows: 
\begin{align} 
f_{Q_{\mathrm{R}}}(q) 
&\!= \!\mathcal{L}^{-1} \! \left\{ \exp\! \left [- 2 \pi \widetilde{\zeta} \!  \int^{R\to\infty}_{0} \! \! \frac{r}{ 1\! + \! r^{4} (s  \widetilde{P}_{T})^{-1} } \mathrm{d}r \right]  \right\}(q) %\nonumber \\ &
 =  \mathcal{L}^{-1}\! \left\{   \exp \left (\!-  \frac{\pi^2 }{ 2 }  \widetilde{\zeta} \sqrt{s \widetilde{P}_{T} }   \right) \! \right\}(q) \nonumber \\
&\overset{\text{(e)}}{=} \frac{1}{2\pi i} \lim_{T \to \infty} \int^{C+i T}_{C-i T} \exp \left (q s \right) \exp \left(\! - \frac{\pi^2 }{ 2 } \widetilde{\zeta} \sqrt{s \widetilde{P}_{T} }  \right) \mathrm{d} s \nonumber \\
&\overset{\text{(f)}}{=} \frac{1}{ \pi } \int^{\infty}_{0} \exp(- q t) \left[ \frac{\exp\Big( \frac{i}{2} \pi^2 \widetilde{\zeta} \sqrt{t \widetilde{P}_{T} } \Big)- \exp \Big( \! - \! \frac{i}{ 2 } \pi^2 \widetilde{\zeta} \sqrt{t \widetilde{P}_{T} }  \Big) }{2i} \right]\mathrm{d} t \nonumber \\
%&\overset{\text{(c)}}{=} 1-\frac{1}{\pi}\int^{\infty}_{0} \exp(-\rho t) \sqrt{1-\mathrm{cosh}\left( 2\pi^2\zeta_{A} (-\beta \eta P_{T}t)^{\frac{2}{\mu}} \left (\mu \sin\left( \frac{\pi(\mu-2)}{\mu} \right) \right)^{-1} \right)^2} \mathrm{d}t \nonumber \\
&\overset{\text{(g)}}{=}\frac{1}{ \pi } \int^{\infty}_{0} \exp (- q t ) \sin \Big( \frac{1}{2}\pi^2 \widetilde{\zeta} \sqrt{t \widetilde{P}_{T}} \Big)   \mathrm{d}u
 %\nonumber \\  & 
 = \frac{ 1 }{4} \left(\frac{\pi}{q}\right)^{\!\!\frac{3}{2}} \widetilde{\zeta} \sqrt{\widetilde{P}_{T}} \exp \left(- \frac{\pi^4 \widetilde{\zeta}^2 \widetilde{P}_{T}}{16 q} \right),  \label{eqn:PDF_PPP}  
\end{align} 
%(c) holds because $ \frac{1}{2 i}\exp(i \cdot x)-\exp(-i \cdot x)= \sin(x) $. 
where (e) follows from Mellin's inverse formula~\cite{P.1995Flajolet} which transforms the inverse Laplace
transform into the complex plane, and $C$ is a fixed constant
greater than the real parts of the singularities of 
$\exp \left(\! - \frac{\pi^2 }{ 2 } \widetilde{\zeta} \sqrt{s \widetilde{P}_{T} }  \right)$, 
(f) applies the Bromwich inversion theorem with the modified contour, and (g) holds as $\frac{\exp(ix)-\exp(-ix)}{2i}=\sin(x)$.

Finally, %can be obtained 
by plugging (\ref{eqn:Det_PPP}) and (\ref{eqn:PDF_PPP}) into (\ref{eqn:theorem_ABR}), we can obtain $\mathcal{S}_{\textup{ABR}}$ in  (\ref{eqn:SP_ABR}) after some mathematical manipulations.
\end{proof}
      
 %\vspace{-3mm}  
\subsection{Proof of Theorem 2}      

\begin{proof}
With ETCP, the mode selection of $\mathrm{R}$ is based on achieved performance in previous time slots  instead of the current one. Therefore, the steady-state success probability of the hybrid relaying in a certain mode is independent of the mode selection probability. 
%Based on the definition in (\ref{definition_SP}), 
Recall that the hybrid relay under ETCP eventually selects the ABR mode when the number of successful transmission in the ABR mode is higher than that in the WPR mode during the exploration period.
According to this mode selection criteria,
the corresponding success probability can be expressed as 
\begin{align} \label{SP_ETCP} 
& \mathcal{S}^{\textup{ETCP}}_{\textup{HR}} \! = \! \mathbb{P}\Big[\nu_{\mathrm{R}} \! > \! \tau_{\mathrm{W}}, \nu^{\mathrm{W}}_{\mathrm{D}} \! >\! \tau_{\mathrm{W}},  E_{\mathrm{R}} \! > \! E_{\mathrm{W}} \Big]  \mathbb{P} \Big[ \mathrm{M}_{\textup{ETCP}} \!=\! \mathrm{W}  \Big] \!+\!  \mathbb P \Big[ \nu_{\mathrm{R}} \! >\! \tau_{\mathrm{W}}, 
  \nu^{\mathrm{A}}_{\mathrm{D}} \! > \! \tau_{\mathrm{A}},  E_{\mathrm{R}} \! > \! E_{\mathrm{A}} \Big]   \mathbb{P} \Big[ \mathrm{M}_{\textup{ETCP}} \!=\! \mathrm{A}   \Big] \nonumber \\
& + \! \frac{1}{2}  \bigg ( \! 1 \! - \! \mathbb{P}\Big[\nu_{\mathrm{R}} \! > \! \tau_{\mathrm{W}}, \nu^{\mathrm{W}}_{\mathrm{D}} \! >\! \tau_{\mathrm{W}},  E_{\mathrm{R}} \! > \! E_{\mathrm{W}} \Big] \! -  \mathbb P \Big[ \nu_{\mathrm{R}} \! >\! \tau_{\mathrm{W}}, 
  \nu^{\mathrm{A}}_{\mathrm{D}} \! > \! \tau_{\mathrm{A}},  E_{\mathrm{R}} \! > \! E_{\mathrm{A}} \Big] \! \bigg) \! \bigg( \! \mathbb{P} \Big[ \mathrm{M}_{\textup{ETCP}} \!=\! \mathrm{W}  \Big] \!\! + \!  \mathbb{P} \Big[ \mathrm{M}_{\textup{ETCP}} \!=\! \mathrm{A}  \Big] \! \bigg)  \nonumber \\
& =   \frac{1}{2} (\mathcal{S}_{\textup{WPR}}       +  \mathcal{S}_{\textup{ABR}})  +\frac{1}{2}\bigg( \mathbb{P}   \Big[ N_{\textup{WPR}}  \!>\! N_{\textup{ABR}} \Big]  - \mathbb{P} \Big[ N_{\textup{ABR}} \!>\! N_{\textup{WPR}}\Big]\bigg) (\mathcal{S}_{\textup{WPR}} - \mathcal{S}_{\textup{ABR}}).
\end{align}   

The number of successful transmissions in the ABR mode and that in the WPR mode are dependent on $\mathcal{S}_{\textup{ABR}}$ and $\mathcal{S}_{\textup{WPR}}$, obtained in (\ref{eqn:theorem_ABR}) and (\ref{eqn:theorem_WPR}), respectively. Based on these results, we have 
\begin{align} \label{P_METCP_B}
& \mathbb{P} \big[  N_{\textup{ABR}} \!>\! N_{\textup{WPR}}  \big] = \sum^{n}_{i=1}  {n \choose i } \mathbb{P} \big[ N_{\textup{ABR} } = i \big] \sum^{i}_{j=1} {n \choose i\! -\! j }  \mathbb{P} \big[ N_{\textup{WPR} } = i - j \big]      \nonumber \\
& = \sum^{n}_{i=1} {n \choose i } \mathcal{S}^{i}_{\textup{ABR}} (1-\mathcal{S}_{\textup{ABR}} )^{n-i}   \sum^{i}_{j=1} 
{n \choose i-j}
\mathcal{S}^{i-j}_{\textup{WPR}} \big(1- \mathcal{S}_{\textup{WPR}} \big)^{n-i+j} . 
\end{align}
 
Similarly, we have %\vspace{-3mm}
\begin{align}  \label{P_METCP_A}
\mathbb{P} \big[  N_{\textup{WPR}} \!>\! N_{\textup{ABR}}  \big] = \sum^{n}_{i=1} {n \choose i } \mathcal{S}^{i}_{\textup{WPR}} (1-\mathcal{S}_{\textup{WPR}} )^{n-i}  \sum^{i}_{j=1} {n \choose i\!-\!j}
\mathcal{S}^{i-j}_{\textup{ABR}} \big(1- \mathcal{S}_{\textup{ABR}} \big)^{n-i+j} . %\vspace{-3mm}
\end{align}

Then, by inserting (\ref{P_METCP_B}) and (\ref{P_METCP_A}) into  (\ref{SP_ETCP}), we have the expression of $\mathcal{S}^{\textup{ETCP}}_{\textup{HR}}$ in  (\ref{eqn:SP_ETC}).
\end{proof}

% \vspace{-4mm}
\subsection{Proof of Theorem 3}
\label{proofThm2}

\begin{proof}

Recall that, with ESAP, the probabilities of the hybrid relaying working in the WPR mode (i.e., $\mathbb{P}[\mathrm{M}=\mathrm{W}]$) and ambient backscatter mode (i.e., $\mathbb{P}[\mathrm{M}=\mathrm{A}]$) have been obtained in (\ref{eqn:CP_HR2}) and (\ref{eqn:CP_HR3}), respectively. 
By inserting (\ref{eqn:CP_HR2}) and (\ref{eqn:CP_HR3})
into the definition in (\ref{def:ergodic_capacity}), we have the ergodic capacity of the hybrid relaying with ESAP as follows:
\begin{align} 
\mathcal{C}^{\textup{ESAP}}_{\textup{HR}} %&  =  \frac{W}{2} \mathbb{E} [ \log_{2} (1 + \nu ) | \mathrm{M}  = \mathrm{W}  ] \mathbb{P} [ \mathrm{M}  = \mathrm{W}  ]   + \frac{C_{\mathrm{A}}}{2} \mathbb{P} [\nu_{\mathrm{R}} \geq \tau_{\mathrm{W}}, \nu^{\mathrm{A}}_{\mathrm{D}} \geq  \tau_{\mathrm{A}} | \mathrm{M}  = \mathrm{A} ] \mathbb{P} [\mathrm{M}  = \mathrm{A} ] , \nonumber \\
&  =  \frac{1-\omega}{2} \Bigg( \mathbb{E} \Big[W \log_{2} (1 + \nu )  {\bf 1}_{ \{\nu_{\mathrm{R}} > \tau_{\mathrm{W}},  \nu^{\mathrm{W}}_{\mathrm{D}} > \tau_{\mathrm{W}}, E_{\mathrm{R}} > E_{ \mathrm{W} } \}} \Big]    %\nonumber \\ & \hspace{60mm} 
+ C_{\mathrm{A}} %\mathbb{P} \Big[ \nu_{\mathrm{R}} \! > \! \tau_{\mathrm{W}} \Big] 
\bigg( \mathbb{P} \Big[ \nu_{\mathrm{R}} \! > \! \tau_{\mathrm{W}}, \nu^{\mathrm{A}}_{\mathrm{D}} \! > \! \tau_{\mathrm{A}} , \nu^{\mathrm{W}}_{\mathrm{D}} \! \leq \! \tau_{\mathrm{W}},    E_{\mathrm{R}} \! > \! E_{\mathrm{W}} \Big] \nonumber \\
& \hspace{90mm} + \mathbb{P} \Big[ \nu_{\mathrm{R}} \! > \! \tau_{\mathrm{W}},
\nu^{\mathrm{A}}_{\mathrm{D}} \! > \! \tau_{\mathrm{A}} ,    E_{\mathrm{W}} \! \geq \! E_{\mathrm{R}} > E_{\mathrm{A}} \Big] \bigg)  \! \Bigg) \nonumber \\
& = \frac{1-\omega}{2} \bigg( \mathbb{E} \Big[W \log_{2} (1 + \nu )  {\bf 1}_{ \{\nu_{\mathrm{R}} > \tau_{\mathrm{W}},  \nu^{\mathrm{W}}_{\mathrm{D}} > \tau_{\mathrm{W}}, E_{\mathrm{R}} > E_{ \mathrm{W} } \}} \Big]    
+ C_{\mathrm{A}} \mathcal{S}^{\textup{ESAP}}_{\mathrm{A}} \! \bigg),
 \label{def:HR_QI}  
\end{align} 
where $\mathcal{S}^{\textup{ESAP}}_{\mathrm{A}}$ has been obtained in~(\ref{eqn:SP_HR_B}).
 
The first term in the brackets of (\ref{def:HR_QI}) can be calculated as
\begin{align} 
& %\frac{1\!-\!\omega}{2} 
\mathbb{E} \Big[ W \log_{2} (1 + \nu )  {\bf 1}_{ \{  \nu_{\mathrm{R}} > \tau_{\mathrm{W}},  \nu^{\mathrm{W}}_{\mathrm{D}} > \tau_{\mathrm{W}},   E_{\mathrm{R}} > E_{\mathrm{W}}   \}} \Big]  \overset{\text{(h)}}{=} %\frac{(1-\omega)W}{2} 
W \! \int^{\infty}_{0}   \! \mathbb{P} \Big[  \log_{2} ( 1 + \nu ) {\bf 1}_{ \{\nu_{\mathrm{R}} > \tau_{\mathrm{W}},  \nu^{\mathrm{W}}_{\mathrm{D}} > \tau_{\mathrm{W}}, E_{\mathrm{R}} > E_{\mathrm{W}}   \}} > t \Big] \mathrm{d} t \nonumber \\ 
& \overset{\text{(i)}}{=}  %\frac{(1\!-\!\omega)W}{2 \ln(2)} 
\frac{ W}{  \ln(2)}   \! \int^{\infty}_{0} \!  \mathbb{P} \Big[   \nu  > v,   \omega T \beta  Q_{\mathrm{R}} \! > \!   E_{\mathrm{W}} \Big]   \frac{1}{1+v} \mathrm{d} v  %\nonumber \\
%\frac{(1\!-\!\omega)W}{2 \ln(2)} 
%&
=\frac{ W}{  \ln(2)} \! \int^{\infty}_{0} \!  \mathbb{P} \Big[  \min ( \nu_{ \mathrm{R}}, \nu^{\mathrm{W}}_{\mathrm{D}}) > v ,   Q_{\mathrm{R}} > \varrho_{\mathrm{W}}  \Big] \frac{1}{1+v} \mathrm{d} v  \nonumber \\
& = %\frac{(1\!-\!\omega)W}{2 \ln(2)} 
\frac{ W}{  \ln(2)}\! \int^{\infty}_{0} \!  \mathbb{P} \Big[   \nu_{\mathrm{R} }  > v ,   \nu^{\mathrm{W}}_{\mathrm{D} } > v ,  Q_{\mathrm{R}} > \varrho_{\mathrm{W}} \Big] \frac{1}{1+v} \mathrm{d} v %\label{eqn:SAP_CP} 
\nonumber \\
%& \overset{\text{(a)}}{=}  \frac{W}{2 \ln(2)}  \int^{\infty}_{0} \!  \mathbb{P} \big[ \nu_{\mathrm{R} }  > v ,   \nu^{\mathrm{W}}_{\mathrm{D} }  > v ,  E_{\mathrm{R}} > E_{\mathrm{W}} \big] \frac{1}{1+v} \mathrm{d} v    
& \overset{\text{(j)}}{=} \! %\frac{(1\!-\!\omega)W}{2\ln(2)} \!\!
\frac{ W}{  \ln(2)} \! \int^{\infty}_{0}  
\! \exp \! \Big( \! \! - \!   \kappa(v) \sigma^2    \!   \Big) \!  \bigg( \!  
%\exp \! \Big( \! \! - \! \ell(v,\rho_C) \sigma^2  \!   \Big) \mathrm{Det} \Big( \mathrm{Id} \! + \! \alpha_{I} \mathbb{A}_{\Phi}\big( \ell (v, \rho_C) \big) \! \Big)^{\!- \frac{1}{ \alpha_{I}}  }  
\chi_{\mathbb{A}} (\tau_{\mathrm{W}}, \rho_{C})  
\Big( 1 \!-\! F_{Q_{\mathrm{R}}} (\varrho_{\mathrm{W}}\!+\!\varrho_C ) \! \Big)   \nonumber \\  %\exp \! \Big( \!\! - \!  \ell\big(v,\omega \beta q - \rho_{A} \big) \sigma^2   \!  \Big) %\mathrm{d} q
%\bigg[ 1 \!- \! \mathcal{L}^{-1} \bigg\{ \! \frac{1}{s}\mathrm{Det}\big(\mathrm{Id}\!+\!\alpha_{Q}\mathbb{Q}_{\Psi}(s)\big)^{\!-\frac{1}{\alpha_{Q}}} \! \bigg \} \! \big( \varrho_{\mathrm{W}}\!+\!\varrho_C \big) \! \bigg]  
%\nonumber  \\ 
& \hspace{85mm} + \! \int^{ \varrho_{\mathrm{W}} \!  + \! \varrho_{C}  }_{ \varrho_{\mathrm{W}} } \!\! \! \! \chi_{\mathbb{A}} (\tau_{\mathrm{W}}, \omega\beta q \! - \! \rho_{\mathrm{W}} )  f_{Q_{\mathrm{R}}}(q) \mathrm{d} q \! \bigg) \frac{1}{1+v} \mathrm{d}v,  %\mathrm{Det} \Big( \mathrm{Id} \! + \! \alpha \mathbb{A}_{\Phi}\big( \ell(v,\omega\beta q \! - \! \rho_{\mathrm{W}} ) \big) \! \Big)^{\!- \frac{1}{ \alpha}  }
\label{eqn:C_A_QI} 
\end{align}
where (h) follows $ \mathbb{E}[A] = \int^{\infty}_{0} \mathbb{P} [A > t ] \mathrm{d}t $  \cite{G.2011Andrews}, (i) %follows by replacing $v=2^t-1$, 
replaces $t$ with $\log_{2}(1+v)$,
and (j) follows the derivation of  $\mathcal{S}^{\textup{ESAP}}_{\mathrm{W}}$ in (\ref{eqn:SP_HR_A}) with $\tau_{\mathrm{W}}$ replaced by $v$.  
 
Then, by plugging (\ref{eqn:SP_HR_B}) and (\ref{eqn:C_A_QI}) into (\ref{def:HR_QI}), we have the ergodic capacity of the hybrid relaying with %the quasi-static interference 
ESAP in (\ref{the:QI_Capacity_HR}).
\end{proof}

  \end{document}